\pdfoutput=1
\documentclass[%
reprint,
superscriptaddress,
nofootinbib,
 amsmath,amssymb,
 aps,
 prx,
 longbibliography,
floatfix,
]{revtex4-1}

\usepackage{graphicx}
\usepackage{dcolumn}
\usepackage{bm}
\usepackage{appendix}
\usepackage{hyperref}
\usepackage{graphicx}
\usepackage{amsmath}
\usepackage{blkarray, bigstrut}
\usepackage{amsthm}
\usepackage{algorithm}
\usepackage{algpseudocode}
\newtheorem{theorem}{Theorem}

\newtheorem{lemma}[theorem]{Lemma}
\theoremstyle{definition}
\newtheorem{definition}{Definition}

\usepackage{mathtools}

\usepackage{array,amssymb,booktabs}
\newcolumntype{C}{>{$}c<{$}}
\AtBeginDocument{
    \heavyrulewidth=.08em
    \lightrulewidth=.05em
    \cmidrulewidth=.03em
    \belowrulesep=.65ex
    \belowbottomsep=0pt
    \aboverulesep=.4ex
    \abovetopsep=0pt
    \cmidrulesep=\doublerulesep
    \cmidrulekern=.5em
    \defaultaddspace=.5em
}
\usepackage{makecell}
\setcellgapes{5pt}

\usepackage{cancel}
\usepackage[usenames, dvipsnames]{color}

\newcommand\mybigstrut[1][4pt]{\setlength\bigstrutjot{#1}\bigstrut[t]}
\newcommand\mynegbigstrut[1][2pt]{\setlength\bigstrutjot{#1}\bigstrut[b]}

\begin{document}

\preprint{APS/123-QED}

\title{Scalable and self-correcting photonic computation using balanced photonic binary tree cascades}

\author{Sunil Pai}
\email{sunilpai@stanford.edu}
\affiliation{Department of Electrical Engineering, Stanford University, Stanford, CA 94305, USA}
\author{Olav Solgaard}
\affiliation{Department of Electrical Engineering, Stanford University, Stanford, CA 94305, USA}
\author{Shanhui Fan}
\affiliation{Department of Electrical Engineering, Stanford University, Stanford, CA 94305, USA}
\author{David A.B. Miller}
\affiliation{Department of Electrical Engineering, Stanford University, Stanford, CA 94305, USA}

\begin{abstract}
Programmable unitary photonic networks that interfere hundreds of modes are emerging as a key technology in energy-efficient sensing, machine learning, cryptography, and linear optical quantum computing applications. In this work, we establish a theoretical framework to quantify error tolerance and scalability in a more general class of ``binary tree cascade'' programmable photonic networks that accept up to tens of thousands of discrete input modes $N$. To justify this scalability claim, we derive error tolerance and configuration time that scale with $\log_2 N$ for balanced trees versus $N$ in unbalanced trees, despite the same number of total components. Specifically, we use second-order perturbation theory to compute phase sensitivity in each waveguide of balanced and unbalanced networks, and we compute the statistics of the sensitivity given random input vectors. We also evaluate such networks after they self-correct, or self-configure, themselves for errors in the circuit due to fabrication error and environmental drift. Our findings have important implications for scaling photonic circuits to much larger circuit sizes; this scaling is particularly critical for applications such as principal component analysis and fast Fourier transforms, which are important algorithms for machine learning and signal processing.
\end{abstract}

\maketitle


\section{Introduction}

Reconfigurable photonic networks or meshes of interferometers are capable of transforming optical modes for quantum computation, energy-efficient deep learning, mode unscrambling, sensing, and beamforming \cite{Bogaerts2020ProgrammableCircuits}. In such applications, a set of $N$ waveguide modes containing data or bits are fed into a mesh network of reconfigurable photonic interferometers and phase shifters whose output represents a matrix multiplication, where the matrix is represented by the transmission matrix of the optical device in the basis of waveguide modes. The ability to perform matrix multiplication in an ultra-fast and energy efficient manner is an attractive property of reconfigurable photonic circuits, which significantly broadens the application space for such devices. This approach also promises a new way of performing self-aligning optics with no moving parts which has applications in high-speed sensing and optical phased arrays \cite{Miller2013Self-aligningCoupler, Miller2020AnalyzingNetworks}. However, a key problem with this architecture is the presence of systematic errors that build up over various stages of the photonic circuit, requiring the definition of a theory of sensitivity.

In this paper, we discuss a common class of reconfigurable photonic networks that contain a network or ``mesh'' of Mach-Zehnder interferometers (MZIs), which can be used to progressively interfere pairs of waveguide modes. These modes may be generated in a number of ways; machine learning and photonic computing might use integrated modulators on a photonic chip operated ideally at sub-nanosecond time scales using electro-optic modulation \cite{Wang2018NanophotonicModulatorsb} and using a laser source tuned to a wavelength of interest (e.g. 1550 nm). Alternatively, these modes can arrive directly from free space as outputs of photonic sensors. This idea is still relatively unexplored but is a major focus of recent work \cite{Miller2020AnalyzingNetworks} showing that multimode fields can be analyzed directly by training self-configurable networks of interferometers arranged in a ``binary tree'' configuration. Such modes convert a continuous free space intensity pattern into a discrete sequence of complex numbers with amplitude and phase characterizing light in the fundamental mode of a single-mode waveguide. Pairs of modes can be constructively and/or destructively interfered sequentially as they propagate through feedforward meshes of interferometers. This can be achieved programmatically by adjusting active phase shifting elements within MZIs placed in the photonic circuit. Such phase shifting elements, the pathway for programming analog matrix multiplication and signal processing in photonics, can operate by modulating the effective mode index thermally \cite{Harris2014EfficientSilicon}, electromechanically \cite{Errando-Herranz2020MEMSCircuits, Edinger2020CompactPhotonics}, using phase change co-integration \cite{Wuttig2017Phase-changeApplications}, or electrooptically \cite{Wang2018NanophotonicModulators}.

As shown in Fig. \ref{fig:node}(a), a triangular network of MZIs can be programmed such that the process of light propagating forward through the network is a physical (analog) implementation of any desired unitary matrix multiplication \cite{Reck1994ExperimentalOperator, Miller2013Self-configuringInvited} \footnote{The triangular unitary model was actually first discovered in 1897 by Hurwitz \cite{Hurwitz1897UberIntegration} and only popularized for physics applications over a century later.}. Here, we propose an alternate method for implementing matrix multiplication in classical photonic circuits. Our approach relies on ``nonlocal'' interactions that can be recursively defined using binary trees, a data structure commonly deployed in computer science applications. Depictions of ``unbalanced'' and``balanced'' binary trees of MZIs are shown in Fig. \ref{fig:node}(d). The recursive definition of such structures is further elaborated in Fig. \ref{fig:tree}, demonstrating the construction of arbitrary binary tree designs by connecting every MZI to an output or up to two binary subtrees.

At a high level, this binary tree ``vector unit'' construction can be programmed to route any incoming set of waveguide modes into a single waveguide \cite{Miller2013Self-aligningCoupler} (Fig. \ref{fig:node}(e, f, g)). The reverse is also true by reciprocity as light can be sent into a single mode waveguide and routed out such that any complex vector output may be generated in the device basis (Fig. \ref{fig:node}(a, b, c)). We specifically show that a vector unit can be represented by a recursive binary tree definition (Fig. \ref{fig:node}(d, h)), which can ultimately be used to explore a deeper mathematical framework for the error tolerance and dispersion of such networks. Such units have important roles for state preparation and readout of optical modes which can be invaluable for machine learning in hybrid neural networks \cite{Pai2022ExperimentallyNetworks} and cryptocurrency hash functions \cite{Pai2022ExperimentalCryptocurrency} as well as telecommunications and optical phased arrays.

Photonic networks that are universal can implement any unitary matrix, contain $N^2$ degrees of freedom, and typically include triangular and rectangular meshes \cite{Clements2016AnInterferometers}, the latter of which is more compact but not self-configurable. However, these networks typically only interact waveguide modes locally (only modes of neighboring waveguides interact), which can limit overall photonic bandwidth for which a certain matrix accuracy is achievable. This motivates exploring whether nonlocal architectures increase the bandwidth in a photonic mesh, which is explored in some detail in Ref. \citenum{Hamerly2021InfinitelyInterferometers}. While crossings for nonlocal connections may be a concern, recent advances in multilayer photonics and silicon nitride-on-silicon in commercial CMOS foundries make such designs plausible \cite{Chiles2017Multi-planarLoss}. Balanced ``binary tree''-based architectures are nonlocal architectures that are a key building block for low-depth photonic circuits such as the butterfly architecture, the Benes network (an arbitrary switching network that is two butterfly architectures back-to-back). They also have mappings to quantum architectures such as state preparation circuits \cite{Araujo2021APreparation}, quantum fast Fourier transform (FFT) \cite{Flamini2017BenchmarkingProcessing}, and the cosine-sine architecture, which is a useful architecture in quantum computation \cite{Mottonen2004QuantumGates}.

Our main contribution is to provide a new error model to explain why increased bandwidth and robustness arises from nonlocal connectivity in balanced binary trees. We generally propose ``wide'' or ``splay'' photonic architectures that provide more tolerance than ``deep'' universal architectures because light has to propagate through fewer devices on average, motivating a core design principle for photonic circuits. To prove these claims, we organize the paper as follows:
\begin{enumerate}
    \item In Section \ref{sec:node}, we define the key unit cell or node for the vector unit, specifically the self-configuration or self-error correction step, which automatically corrects for any fabrication and/or hardware error without requiring off-chip calculations \cite{Miller2013Self-configuringInvited, Bandyopadhyay2021HardwarePhotonics, Hamerly2021AccurateInterferometers}.
    \item In Section \ref{sec:network}, we define the balanced and unbalanced vector units mathematically, laying the foundation for our error analysis.
    \item In Section \ref{sec:error}, we define a second-order perturbation theory error model for general feedforward networks. More strikingly, we show that sensitivity of each phase shifter is proportional to the power going through it, which is the basis of our error model. We can leverage this property to non-invasively monitor powers in any feedforward photonic circuit using binary trees at the input and output of the device.
    \item In Section \ref{sec:results}, we bring together the results of Sec. \ref{sec:network} and Sec. \ref{sec:error} to derive a statistical power model that explains the increased robustness scaling of balanced compared to unbalanced vector units. We also perform simulations that compare various phase and coupling errors in the mesh that agree with the derived scaling properties.
    \item In Section \ref{sec:apps}, we propose a new ``binary tree cascade'' model, a generalization of the the triangular architecture \cite{Reck1994ExperimentalOperator, Miller2013Self-configuringInvited} which can be useful for error-tolerant mode conversion and principal component analysis-based signal processing. The Appendix later expands on the binary tree cascade to propose other new error-tolerant splay architectures and relates our work to existing well-known butterfly architectures.
\end{enumerate}

\section{Component model} \label{sec:node}

Before we discuss a network model, we propose a component model for the individual nodes or building blocks of the network. First we note that a node does not have to be explicitly an MZI, which has been recognized more recently \cite{Bogaerts2020ProgrammableCircuits} with the development of tunable coupling elements, so we need to describe the basic functionality of the node component that abstracts away the MZI functionality. 

\subsection{Photonic node core functionality}

In general the core functionality of a photonic node in the network is:
\begin{enumerate}
    \item The node is a $ 2 \times 2 $ component (2 left ports and two right ports).
    \item Given any input into the left ports of a node, the node arbitrarily redirects the light entirely into either one of its right ports.
    \item By reciprocity, given an input into either right port, the node arbitrarily redirects the light into its left ports in both amplitude and \textit{relative} or \textit{differential} phase.
\end{enumerate}
Here, we refer to left and right ports separately from input and output ports, following the naming convention of Ref. \citenum{Miller2020AnalyzingNetworks} because the above definition assumes that there is a phase shifter present on the left side of the device and that light can enter from either direction. The ratio of normalized powers in the output ports is known as a ``split ratio.'' We define the reflectivity $r = 1 - s$ and transmissivity $s = 1 - r$ respectively as the fractional power in the bar port and cross port from the port in which light is sent in. ``Bar state'' means $r = 1, s = 0$ (all light goes to ``same side'' port, e.g. lower input to lower output port), and ``cross state'' means $s = 1, r = 0$ (all light goes to the ``opposite side'' port, e.g. lower input to upper output port). Ultimately, the main idea of this paper is to work with photonic circuits that connect these nodes up to implement arbitrary multi-waveguide modes, as discussed in Ref. \citenum{Miller2020AnalyzingNetworks}, as well as other photonic mesh feedforward networks \cite{Pai2020ParallelNetwork}.

In this paper, our error models work for two implementations of nodes depicted in Fig. \ref{fig:node}(b): the tunable directional coupler (TDC) node and the Mach Zehnder interferometer (MZI) node \cite{Bogaerts2020ProgrammableCircuits}, the latter of which is already quite ubiquitous in commercial and academic implementations of these systems \cite{Shen2017DeepCircuits, Annoni2017UnscramblingModes, Harris2018LinearProcessors, Taballione20188x8Waveguides}. The theme here is in general that there are two elements to a node: a phase shifter and a tunable split ratio. In the MZI case, the tunable ratio consists of two directional couplers and a phase shifter. In the TDC case, there is an explicit physical mechanism to modulate the splitting ratio of a $2 \times 2$ coupler by directly perturbing the active coupling region, side-stepping potential errors in the individual passive directional couplers of the MZI node, though there are nodes that can be corrected for these errors in triangular and rectangular architectures \cite{Hamerly2021InfinitelyInterferometers}.

\begin{figure*}
    \centering
    \includegraphics[width=\linewidth]{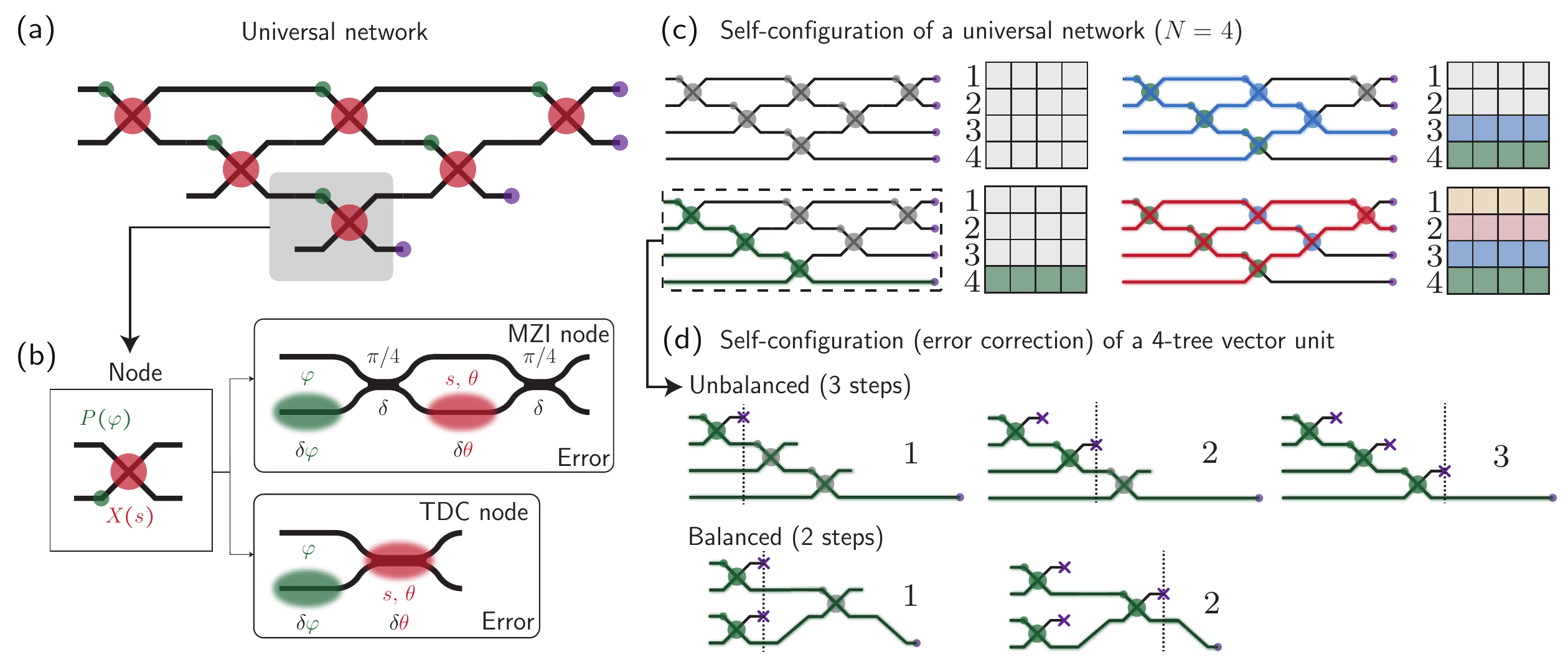}
    \caption{(a) We show a universal $4 \times 4$ photonic network that arises from the well-known Reck architecture. (b) The individual nodes of the universal network are either MZIs or TDCs with errors in phase shifts and coupling. A benefit of the TDC is that it does not rely on fundamental limits based on errors in the fixed splitters $\delta$. (c) Self-configuration of a universal network of four inputs consists of three steps, where each step involves self-configuration to encode a row vector of $U$, a $4 \times 4$ unitary device operator. (d) Self-configuration of a 4-tree vector unit involves nullifying columns in parallel, similar to the proposals in Refs. \cite{Miller2020AnalyzingNetworks, Pai2020ParallelNetwork}.}
    \label{fig:node}
\end{figure*}

\subsection{Node transmission matrix}

It is often mathematically convenient to represent the transmission matrix acting on the two element vector representing the input modes approaching from the left side $\bm{x} \equiv (x_0, x_1) = (x_\mathrm{lower}, x_\mathrm{upper})$. When nodes are connected together, this approach ultimately allows us to define a matrix $U$ to represent the transmission matrix. This may sound familiar, as this formulation is a condensed version of the $S$-matrix formulation, which considers both reflection into the input ports and transmission into the output ports. Since we assume the reflection is sufficiently small, the transmission matrix formulation is also sufficient since it considers transmission into the output ports given an input port excitation. Therefore, based on our definition above, the general ideal $2 \times 2$ node matrix is defined in terms of a phase matrix component $P$ and coupling component $X(s)$:
\begin{equation}\label{eqn:node}
    T(s, \phi) \equiv X(s) P(\phi) \equiv  \begin{bmatrix} -\sqrt{1 - s} & \sqrt{s}\\ \sqrt{s} & \sqrt{1 - s} \end{bmatrix} \begin{bmatrix} e^{i \phi} & 0\\ 0 & 1 \end{bmatrix},
\end{equation}
where any node implements some functionally equivalent form of the $XP$ representation above. In this paper, we abstract away the details of exactly how $s$ and $\phi$ are implemented, though generally it may take the form of an MZI or a tunable coupler \cite{Bogaerts2020ProgrammableCircuits}. In such cases, we typically can find a physical parameter that behaves like a phase (e.g. an arm of the MZI) or the inverse beat length (difference between the first and second mode propagation constants in an MMI or directional coupler) such that the transmissivity $s = \cos^2 \frac{\theta}{2}$ for $\theta \in [0, \pi]$. More explicitly, we will also consider the following more standard parametrization which for MZIs, assumes a ``differential mode phase shift'' \cite{Pai2020ParallelNetwork}:
\begin{equation}\label{eqn:mzinode}
    T_{\mathrm{MZI}}(\theta, \phi) \equiv X(\theta) P(\phi) \equiv  ie^{i \frac{\theta}{2}}\begin{bmatrix} -\sin \frac{\theta}{2} & \cos \frac{\theta}{2}\\ \cos \frac{\theta}{2} & \sin \frac{\theta}{2} \end{bmatrix} \begin{bmatrix} e^{i \phi} & 0\\ 0 & 1 \end{bmatrix}.
\end{equation}
A TDC node, on the other hand, has the parametrization, based on tuning a beat length of the directional or MZI:
\begin{equation}\label{eqn:tdcnode}
    T_{\mathrm{TDC}}(\theta, \phi) \equiv X(\theta) P(\phi) \equiv  ie^{i f(\theta)}\begin{bmatrix} \cos \frac{\theta}{2} & i\sin \frac{\theta}{2}\\ i\sin \frac{\theta}{2} & \cos \frac{\theta}{2} \end{bmatrix} \begin{bmatrix} e^{i \phi} & 0\\ 0 & 1 \end{bmatrix},
\end{equation}
where we include a $e^{i f(\theta)}$ overall phase term that depends on the exact TDC design and coupled mode theory \cite{Huang1994Coupled-modeOverview}. Ultimately, we will focus specifically on architectures of MZI nodes, but our analysis can be extended to uncover error models of architectures of TDC nodes as well.

\subsection{Feedforward architectures}

In this paper, we are specifically interested in how error from individual nodes affects the overall error in entire ``feedforward networks'' of nodes given by Eq. \ref{eqn:node}  \cite{Pai2020ParallelNetwork}, as opposed to networks containing no cyclic loops as in Ref. \citenum{Perez2017HexagonalInterferometers}. As previously discussed, binary tree architectures, butterfly/Benes nonlocal architectures, and rectangular or triangular universal architectures all fall under the umbrella of feedforward networks. As with the component definition above, we assume monochromatic light and light propagating from left-to-right in the network.

When a feedforward photonic circuit is operated, an $N$-dimensional input vector (or input data) $\bm{x}$ enters the left waveguide ports of the device and propagates forward (left-to-right) through the MZI network until it reaches the right side of the network where output amplitude and phase $\bm{y}$ is measured at a set of photodetectors \cite{Miller2013Self-configuringInvited}. Since ideally no light is lost in the circuit (the coupling nodes are ideally just rearranging the light), the propagation of the $N$-dimensional input through the network may be modelled as a \textit{unitary} or norm preserving transformation $U$. Thus our measured $\bm{y}$ is a result of a change in mode basis operation performed by the photonic circuit, which can be represented mathematically as a matrix product $\bm{y} = U \bm{x}$. To recover the mode basis, the column vectors of $U$ can be determined by measuring outputs given inputs into individual input waveguides of the circuit. The row vectors are recovered by sending light back in from the right and measuring amplitudes and conjugate phases on the left of the circuit. Note that there is a final set of phase shifts (a tunable phase screen) placed at the end of the feedforward network that can be generally useful to achieve unitary architectures. For simplicity, we generally do not consider the contribution of these phase shifts as they are not required for self-configuration of vector unit architectures \cite{Miller2020AnalyzingNetworks}.


As shown in Fig. \ref{fig:node}(c, d), the property that MZIs can guide any incoming mode into a single waveguide means that the universal triangular architecture in Fig. \ref{fig:node}(a) can self-configure itself to program any unitary operator in an $N \times N$ optical spatially multiplexed system \cite{Miller2013Self-configuringInvited}. The protocol consists of three steps to self-configure a $4 \times 4$ unitary $U$, where each step self-configures a vector unit of successively smaller inputs (4, 3, 2). This protocol is useful because it can also correct for any fabrication errors in the device, owing to the model-free optimizations of phase shifts at each step of the process, discussed in further detail in the Appendix. Fig. \ref{fig:node}(d) suggests that depending on the architecture, self-configuration can be completed in fewer steps (here shown for a ``balanced'' vector unit example).

As it pertains to this paper, self-configuration uses model-free feedback optimizations to automatically correct for any error while programming, a convenient property we will use later to analyze error-corrected bounds of binary tree architectures. These errors appear in the form of phase errors (wavelength tuning or environmental perturbations), coupling errors and optical losses due to fabrication process variation. In Refs. \cite{Bandyopadhyay2021HardwarePhotonics, Hamerly2021InfinitelyInterferometers}, strategies for error tolerance are proposed that enable error correction in the presence of realistic errors in the splitters, which can actually be mapped to corrected phase errors. These issues are potentially avoided by using a tunable TDC node because such nodes include couplers that are nominally tunable from cross to bar state. We will focus specifically on modelling the MZI node in this paper for simplicity, because as previously mentioned, we avoid a global phase term in the TDC definition that varies with $\theta$ and is thus less straightforward to model. However, assuming such a model is found, similar concepts can be transferred from our analysis in this paper.

\section{Binary tree networks} \label{sec:network}

\begin{figure*}
    \centering
    \includegraphics[width=\textwidth]{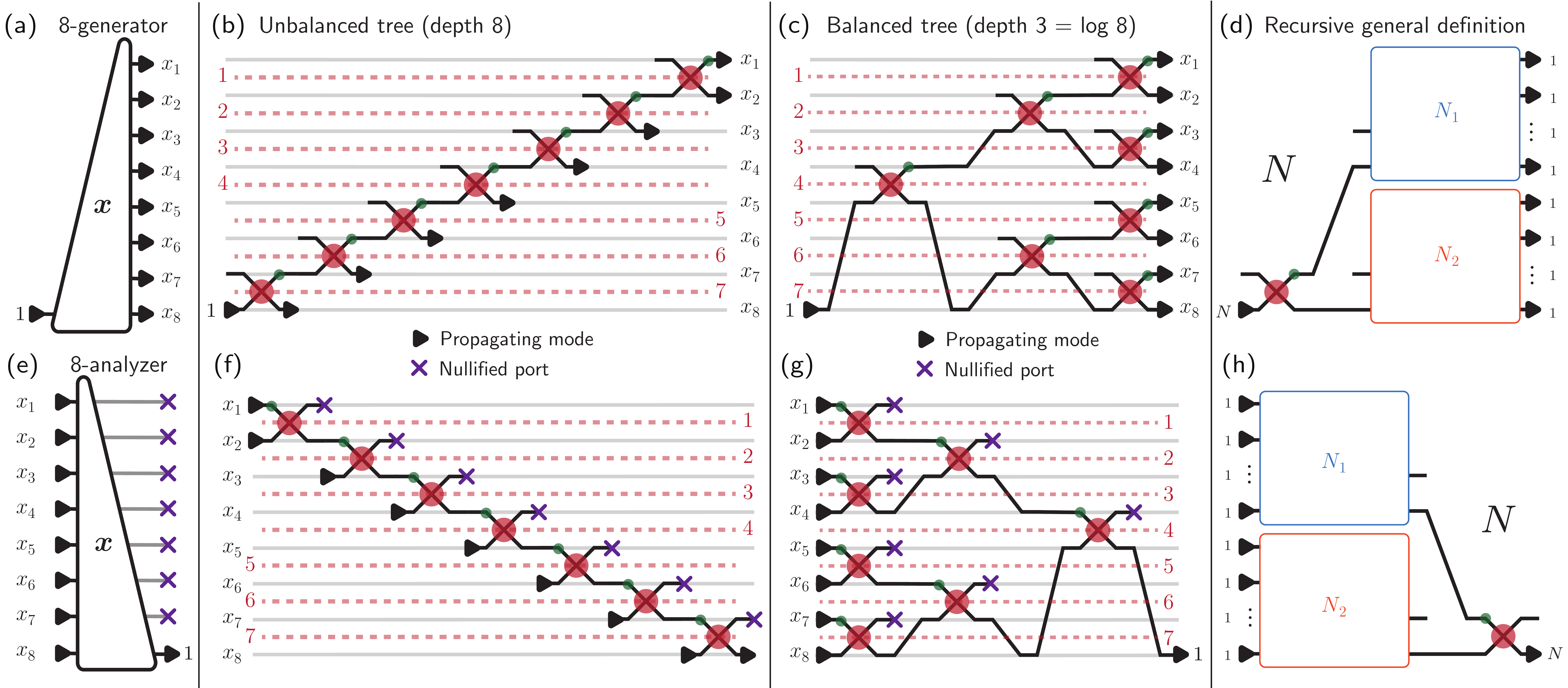}
    \caption{(a) Generator symbol for $N = 8$. (b) Unbalanced generator network for $N = 8$. (c) Balanced generator network for $N = 8$. (d) Recursive generator network definition to generate any vector unit structure. (e) Analyzer symbol for $N = 8$, including nullified ports indicated by purple crosses. (f) Unbalanced analyzer network for $N = 8$. (g) Balanced analyzer network for $N = 8$. (h) Recursive analyzer network definition to generate any vector unit structure. Note: all red dotted lines and red labels denote the index assignments for the individual nodes (both the $\theta_n, \phi_n$ phase shifters in those nodes) compatible with the definition in our Simphox framework \cite{Pai2022Simphox:Library}.}
    \label{fig:tree}
\end{figure*}

In this section, we describe a particular implementation of $U$ that allows for self-configuration. For this, we assume lossless and ideal nodes as described in Section \ref{sec:node}.

\subsection{Embedding nodes in a larger circuit}
In an $N$-waveguide photonic circuit that is feedforward, we can consider modes propagating through a circuit along $N$ ``rails'' indexed $1$ to $N$. This framework is defined explicitly in our previous work Ref. \cite{Pai2020ParallelNetwork} and also in the original proposal for photonic networks \cite{Reck1994ExperimentalOperator}.

We now need some formalism to place $2 \times 2$ elements along these rails in a specified order for light to propagate. For this, we define an embedding matrix for $T$ of Eq. \ref{eqn:node} of the form $T_{m, n}^{[N]}$, where modes $m$ and $n$ are the waveguide indices for the modes to interfere and $N$ is the number of waveguides or ``circuit size.'' This is a unitary operator known in the mathematics literature as a \textit{Givens rotation}:

\begin{equation}\label{eqn:givensrotation}
T_{m, n}^{[N]} \,:=\,\, \begin{blockarray}{ccccccccc}
 &  & \small{m} & & \small{n} &  & N & \\
\begin{block}{[ccccccc]cl}
1   & \cdots &    0   &   0  &   0   & \cdots &    0 & & \mybigstrut\\
\vdots & \ddots & \vdots &  \vdots &  \vdots &        & \vdots & &\\
0   & \cdots &    T_{11}  & \cdots &   T_{12}   & \cdots &    0 & & \small{m} \\
0   & \ddots &    \vdots  & \ddots &   \vdots   & \ddots &    0 & &  \\
0   & \cdots &    T_{21}  & \cdots &   T_{22}  & \cdots &    0   & & \small{n}\\
\vdots &   & \vdots  & \cdots & \vdots & \ddots & \vdots & &\\
0   & \cdots &   0 & 0 &    0   & \cdots &    1 & & N\mynegbigstrut\\
\end{block}
 &  & & &  & & \\
\end{blockarray},\\
\end{equation}
where $T$ is defined as in Eq. \ref{eqn:node}.

In summary, we have a unitary matrix that is an identity matrix, with 1's along the diagonal except in row/column $m, n$ where there are also off-diagonal terms $T_{21}, T_{12}$. This deviates from the definition for embedded photonic node in Ref. \citenum{Pai2020ParallelNetwork}, but ultimately simplifies the representation for defining tree networks.

We refer to a sequence of embedded nodes $\{..., T_{m, n}^{[N]}, ...\}$ that commute, or equivalently do not share any common index $m, n$ as a ``level'' or ``column'' of nodes. We showed previously in Ref. \citenum{Pai2020ParallelNetwork} that any feedforward circuit can be decomposed into efficient columns of nodes that can be tuned (or self-configured) in parallel. For example, the nodes $T_{1, 2}^{[N]}$ and $T_{2, 3}^{[N]}$ \textit{cannot} be evaluated simultaneously (do not commute) because they share a common index $2$ and thus are connected to each other differently depending on the specified order, i.e. $T_{1, 2}^{[N]}T_{2, 3}^{[N]} \neq T_{2, 3}^{[N]}T_{1, 2}^{[N]}$. On the other hand, $T_{1, 2}^{[N]}$ and $T_{3, 4}^{[N]}$ operators do commute ($T_{1, 2}^{[N]}T_{3, 4}^{[N]} = T_{3, 4}^{[N]}T_{1, 2}^{[N]}$) and equivalently can be applied simultaneously. Grouping all columns of nodes is critical in any feedforward architecture to build compact designs and representations for self-configuration or calibration, and this can be achieved via topological sorting by time-order traversal \cite{Pai2020ParallelNetwork}.

\subsection{Definition}

Vector units are architectures that encode arbitrary complex vectors up to an overall magnitude. A vector unit is also capable of implementing a single row or column vector of $U$ in a photonic circuit, the first step in generating an architecture that computes $\bm{y} = U \bm{x}$ matrix-vector products. Such vector units can be fully specified using a \textit{binary tree} data structure consisting of nodes of the form in Eq. \ref{eqn:givensrotation}. At a high level, binary tree structure ensures a single waveguide has a connection path to any $N$ waveguides, and can be used to generate either a generator or an analyzer configuration \cite{Miller2020AnalyzingNetworks}. Vector units are currently primarily deployed as arbitrary state preparation architectures as an input to a general photonic circuit. State preparation (generation) is the inverse of self-configuration (analysis). Thus, by reciprocity as in Fig. \ref{fig:tree}(a), a tree network self-configured (programmed) to an $N$-dimensional complex vector $\bm{y}_N \in \mathbb{C}^N$ can be physically flipped (input enters from root node) and used as a generator to implement $\bm{y}_N^*$. In other words, we can send light into the root node of a balanced tree circuit in generator orientation and achieve any prepared state. In this vein, we will refer to the ``orientation'' of a tree network as an ``analyzer configuration'' if the root node is at the output of the network, and as a ``generator configuration'' if the root node is at the input of the network, assuming light always goes from left to right, following the convention of Ref. \citenum{Miller2020AnalyzingNetworks}.

Balanced vector units in photonic circuits generally require $\log_2 N$ layers, and consist of either a generator binary tree \cite{Miller2020AnalyzingNetworks} or a splitting tree with programmable phase and attenuating elements to define the individual elements of the vector. A generator binary tree (as we will define in this paper) is less lossy than the attenuating element version, which can incur an additional loss of up to $-20$ dB for $N = 100$ when ``one-hot" vectors (standard basis vectors) are programmed on the device.

\subsection{Binary tree}
Analyzer-oriented and generator-oriented vector units are ``binary tree'' graph networks decorated with input and output node-edge pairs, so it is useful to define a binary tree in this context. In this paper, we define a ``binary tree'' using the standard constructive definition definition: starting from a root node, we add at most two child nodes and repeat this process on the child nodes which can also be considered binary trees (or ``subgraphs'') of the original tree. If node $A$ is a child of node $B$, then there is an edge from node $B$ to node $A$. Optionally, we may use directed edges in binary trees, which in our context, are equivalent in function to a waveguide with photons or light travelling in the direction of the edge. Any node that has no children is referred to as a ``leaf'' of the binary tree. We refer to an $N$-tree as any tree that has $N - 1$ nodes. For instance, a $2$-tree is just a single root node with no children. We will show later that an $N$-tree can fully parametrize a vector unit that parametrizes any $\bm{y}_N$.

As shown in Fig. \ref{fig:tree}(b, c), there are two structural or graph-topological extremes for defining an $N$-tree. On one extreme, each node has a single child node: this is also referred to as a ``unbalanced'' binary tree network since the resulting structure is a maximally unbalanced tree. On the other extreme, each node above the ``leaf level'' has two children, in which case we get a \textit{balanced} binary tree, which is also known as a ``divide and conquer'' scheme. Note our definition of binary tree diverges from that of Ref. \citenum{Miller2020AnalyzingNetworks}, in that linear chains are also considered to be extreme forms of binary trees and ``hybrid'' architectures are just binary trees of varying structure or graph topology. We will also use the terms ``balanced tree'' and ``unbalanced tree'' in place of ``binary tree'' and ``diagonal line.'' In other words, our use of ``binary tree'' here maps to a broader definition of a binary tree than in Ref. \citenum{Miller2020AnalyzingNetworks}, where we allow varying extremes of balance. 

Binary trees, as defined here, are the minimal representation to define any self-configuring layer. (This representation is consistent with the criteria for self-configuring layers in Ref. \citenum{Miller2020AnalyzingNetworks} (Appendix).) However, such data structures do not consider other types of nodes in the actual photonic circuit that need to be defined such as input and output nodes (representing optical interconnects such as edge or grating couplers), because these nodes are not necessary to define to arrive at a specific MZI network topology. This is because once a binary tree subgraph is defined, we ``fill'' any missing edges and assign them to input (node with outgoing edge) or output (node with incoming edge) node-edge pairs respectively. In particular, all nodes should have a $2 \times 2$ structure (two incoming and two outgoing edges), and the remaining unfilled edges must be ``filled'' by any missing input or output node-edge pairs. In the final circuit graph, all nodes with a parent have one output node-edge pair filled and the root node has two such pairs filled. All nodes with a single child node have missing outgoing edges filled by an input node-edge pair, and all nodes with no children (leaf nodes) have both missing outgoing edges filled by input node-edge pairs.

\subsection{Generator and calibration}
The elements of a generated vector can be expressed in terms of the phases in a recursive manner, where a node $D_N = T_{N_1, N}^{[N]}(\theta_{N_1}, \phi_{N_1})$ splits vector magnitudes across subtrees that have $N_1$ and $N_2$ outputs as follows:
\begin{equation} \label{eqn:generator}
    \begin{aligned}
        \bm{x}_N &:= D_N^\dagger \bm{e}_N = \begin{bmatrix}
               \cos\left(\frac{\theta_{N_1}}{2}\right) e^{i \phi_{N_1}} \bm{x}_{N_1} \\
               \sin\left(\frac{\theta_{N_1}}{2}\right) \bm{x}_{N_2}
        \end{bmatrix},
    \end{aligned}
\end{equation}
where we define the device operator $D_N^\dagger$ with input to the $N$th (``bottom'') input of the generator device, and we prune off $\theta, \phi$ for the root node of $\bm{x}$ at each recursive step. This simple representation gives a direct formula for each vector element and allows for straightforward calculation of error sensitivities with respect to phases in the specific programmed vector. This is also a convenient formula to use when computing the sensitivity model of the architecture in terms of the control voltage of the phase shifters rather than the phase shifts alone.

The generator configuration can also be used to \textit{calibrate} the $s$ split ratio or $\theta$ phase shifters of the device to generate a lookup table for each of the phase shifter settings in the network as a function of voltages or some other control parameter. Such procedures are defined in Ref. \citenum{Miller2020AnalyzingNetworks} for any binary tree network. One key point in the paper is that rearranging phase shifts into a symmetric configuration (with phase shifts in the two arms of the MZI), can allow one to effectively also calibrate the $\phi$ phase shifters via a reparametrization of the system. However, the $\phi$ phase shifters as defined in this paper cannot be calibrated this way. That can instead be achieved using an analyzer configuration and self configuration as we now discuss.

\subsection{Analyzer and self-configuration}
We now consider the self-configuring analyzer architecture, first proposed in Ref. \citenum{Miller2013Self-aligningCoupler}. We can write a more explicit formula for the unitary matrix $D_N$ first used in Eq. \ref{eqn:generator} in terms of products of matrices of the form in Eq. \ref{eqn:givensrotation}.

Shifting from a graph topology to a matrix framework, we define a formula for \textit{any} binary tree resulting in the unitary matrix $D_N$ (unitary binary tree formulation with $N$ outputs) by a recursive definition:
\begin{equation} \label{eqn:binarytree}
    \begin{aligned}
        D_{N}(\mathbf{s}_N,\bm{\phi}_N) &= T_{N_1, N}^{[N]}(s_{N_1}, \phi_{N_1})\begin{bmatrix}
                D_{N_1} & 0 \\
                0 & D_{N_2}
        \end{bmatrix} \\
        \mathbf{s}_N &= [\ldots, s_{N_1}, \ldots] \\
        \bm{\phi}_N &= [\ldots, \phi_{N_1}, \ldots]
    \end{aligned}
\end{equation}
where the only requirement is that $1 < N_1 < N$ and $N_1 + N_2 = N$ ($N_1$ decides how many outputs each subtree gets at each split, the rail index of the top waveguide, and can be considered the index of the node within a given subtree). We refer to $D_{N_1}$ as the ``top subtree'' and $D_{N_2}$ as the ``bottom subtree.'' To label the nodes further down in the tree (the ``descendants'' of node $N_1$), we apply an offset $N_1$ to any node index in the bottom subtree $D_{N_2}$ as part of the recursive step. As a part of defining the block matrix, we also let ``$0$'' denote the setting of all off-diagonal block-matrix elements to zero. In general $\mathbf{s}_N$ contains $N - 1$ values in the range $[0, 1]$ and $\bm{\phi}_N$ contains $N - 1$ values in the range $[0, 2\pi)$. This, along with the overall magnitude or power and the overall phase accumulated, accounts for all the degrees of freedom in a complex vector, thus $2N$ total degrees of freedom to define an $N$-dimensional complex vector.

Any self-configuring graph architecture may be defined by specifying $N_1$ for recursively defined subtree branches until the base case $D_1 = 1$ is reached. For a (``fully'') balanced tree, we have $N_1 = \lfloor \frac{N}{2} \rfloor$ (i.e., the ``floor'' or integer part of $\frac{N}{2}$), and for a (``fully'') unbalanced tree (chain or ``diagonal line'') we have $N_1 = N - 1$. Variations on the binary tree structure are specified by recursively specifying $N_1$ starting from the root node into repeated invocations of Eq. \ref{eqn:binarytree} until the base case is reached. We will mostly be concerned with the extreme two cases in this paper, because they exhibit key differences in how they scale, but arbitrary structures may also be defined, whose structures fall in the spectrum between balanced and (fully) unbalanced.

A proof of self-configuration can also be done straightforwardly using recursion and the methods of Eqs. \ref{eqn:phisc} and \ref{eqn:ssc}. Consider a lab setting where we are given some complex vector $\bm{x}_N \in \mathbb{C}^N$ and we do not know anything about the settings of the network implementing $D_N$. Through successive minimizations of powers in detectors, we can experimentally self-configure the device so that all the light exits the bottom waveguide indexed at $N$. In mathematical terms $D_N \bm{x}_N = \widetilde{x}_{N} \bm{e}_N$, where $\bm{e}_N$ is the $N$th standard basis vector in $\mathbb{C}^N$ and $\widetilde{x}_{N}$ represents some  phase and amplitude output by the self-configuration with $|\widetilde{x}_{N}|^2 = \|\bm{x}\|^2 \equiv P_N$ (power conservation). Along with these equations, we can use induction to prove that Eq. \ref{eqn:binarytree} implements self-configuration. Specifically, we first split $\bm{x} = [\bm{x}_{N_1}, \bm{x}_{N_2}]$ and apply our inductive hypothesis that $D_{N_1} \bm{x}_{N_1} = \widetilde{x}_{N_1}\bm{e}_{N_1}$ and $D_{N_2} \bm{x}_{N_2} = \widetilde{x}_{N_2}\bm{e}_{N_2}$. Note that power conservation ensures $P_N = P_{N_1} + P_{N_2}$, i.e. the total power in the first and second subtrees of the overall binary tree is always additive (no power is lost or gained).

To complete the proof, we note that we can apply the procedure in the methods of Eqs. \ref{eqn:phisc} and \ref{eqn:ssc} to program the matrix $T_{N_1, N}^{[N]}$ to achieve the result:
\begin{equation} \label{eqn:selfconfiguringproof}
    \begin{aligned}
        D_N \bm{x}_N &= T_{N_1, N}^{[N]}(D_{N_1}, \phi_{N_1})\begin{bmatrix}
                D_{N_1} & 0 \\
                0 & D_{N_2}
        \end{bmatrix} \bm{x}_N\\
        &= T_{N_1, N}^{[N]}(s_{N_1}, \phi_{N_1}) (\widetilde{x}_{N_1}\bm{e}_{N_1} + \widetilde{x}_{N_2}\bm{e}_{N}) \\
        &\equiv  \widetilde{x}_{N} \bm{e}_{N}.
    \end{aligned}
\end{equation}

Note that the base case for the inductive proof is trivial: $D_1 \bm{x}_1 = \bm{x}_1 = \widetilde{x}_{1} \bm{e}_{1}$ since $\bm{e}_{1} = U_1 = 1$. We have now proven that using our recursive definition, any binary tree (not just ``fully" balanced or ``fully" unbalanced) can be self-configured according to the physical process outlined in Eqs. \ref{eqn:phisc} and \ref{eqn:ssc}.

Now that we have proven self configuration is always possible for a (general) binary tree network, such a network can be parametrized in one of two ways: in terms of node settings (which we have defined above) or in terms of the inputs into the layer $\bm{u}_N$ satisfying the condition $D_N(\mathbf{s}_N, \bm{\phi}_N) \bm{u}_N = \bm{e}_{N}$. To define the latter, we define a new matrix function $R_N(\bm{u}_N) = D_N(\mathbf{s}_N, \bm{\phi}_N)$, where $\bm{u}_N \in \mathbb{C}$ satisfies $\|\bm{u}_N\|^2 = 1$ and $\mathrm{arg}(u_1) = 0$ so that the remaining $2N - 2$ degrees of freedom match the number provided by $N - 1$ nodes ($s, \phi$ for each node).

Finally, the time to self-configure relies on the number of columns in the vector unit since MZIs in a column can be self-configured simultaneously \cite{Pai2020ParallelNetwork}. The self-configuration time would need to be sufficiently small to adjust to any incoming training signal, e.g. for a sensor reading modulated optical modes. The number of columns in a binary tree is given by $\log_2 N$, whereas the number of columns in a fully unbalanced tree is given by $N$, so therefore self-configuration is faster by an order of $N / \log_2 N$, which grows quickly with $N$. Hereafter we will use $\log N$ to refer to $\log_2 N$ for ease of notation.

\section{Error model} \label{sec:error}

In this section, we set up the core results of this paper by deriving phase error sensitivity properties for feedforward photonic networks rigorously defined using the ``Hessian'' of a least squares error function describing the intended and measured behavior of the devices. This more general result is needed to ultimately compare the overall performance of various vector unit architectures, including both phase and coupling errors. Unlike in previous sections, here we assume that the feedforward network is already programmed to some desired setting of phase shifts, so the actual implemented unitary operator on the device is already intended to be some $U$. Given this ideal state, we would like to analyze some perturbed ``error'' state $\widehat{U}$ where we vary one or two of the phases.

\subsection{Error function and Hessian}

Here, we proceed to define architecture-dependent ``error sensitivities'' of vector units. We define the error sensitivity of a component to be a ratio between the overall mean square error (due to that component alone) and component error. Here, we seek to describe how various errors (other than loss) in various elements of the circuit affect the \textit{overall} circuit error. Such errors (e.g., fabrication and environmental errors) ultimately can be modelled as coupling errors and phase errors; since coupling errors can be effectively reduced to phase errors in an ideal node \cite{Bandyopadhyay2021HardwarePhotonics}, the phase error can effectively be used to describe either.

We begin by defining the mean square circuit error: 
\begin{equation} \label{eqn:vectorerror}
    \epsilon^2 = \|\bm{y} - \hat{\bm{y}}\|^2= 2 - 2\mathcal{R}(\bm{y}^\dagger \hat{\bm{y}}),
\end{equation}
where $\hat{\bm{y}}$ is the measured output vector, $\bm{y}$ is the predicted output vector, and $\mathcal{R}$ denotes taking the real part. Note again that we are now dealing with the output of a general feedforward mesh.

The first step in calculating the sensitivity of any device is to realize the gradient of the error function in Eq. \ref{eqn:vectorerror} is zero when $\bm{y} = \hat{\bm{y}}$, hence the use of a second-order term called a ``Hessian'' to describe the errors in the circuit. Given that we are aiming for generality, note that this holds regardless of whether the intended input/output behavior $\bm{x}, \bm{y}$ is known a priori. We could in principle run an experiment using an ideal analyzer and generator to determine the full matrix $U$ implemented by the device without any knowledge of the internal settings of the device (assuming Hermitian operator). However, the internal settings of the device, whatever they are, will still have some sensitivity based on the \textit{current} implemented settings.

Assume that the vectors $\bm{\theta}, \bm{\phi}$ are the true phases and the vectors $\widetilde{\bm{\theta}}, \widetilde{\bm{\phi}}$ are the phases with error. We define the error vector $\bm{\Delta} := [\bm{\Delta}_{\theta}, \bm{\Delta}_{\phi}] := [\bm{\theta} - \widetilde{\bm{\theta}}, \bm{\phi} - \widetilde{\bm{\phi}}] := \bm{\eta} - \bm{\eta}'$. This gives the following expression for the error in the network phases:
\begin{equation}\label{eqn:hessian}
    \begin{aligned}
        \epsilon^2(\bm{\Delta}) &= \cancelto{0}{\epsilon^2(\bm{0})} + \bm{\Delta}^T \cancelto{0}{\frac{\partial \epsilon^2}{\partial \bm{\Delta}}} + \frac{1}{2}\bm{\Delta}^T \mathcal{H}_{\epsilon^2} \bm{\Delta} + \cdots \\
        &= \frac{1}{2}\bm{\Delta}^T \mathcal{H}_{\epsilon^2} \bm{\Delta} \\
        &:= \frac{1}{2}[\bm{\Delta}^T_\theta, \bm{\Delta}^T_\phi]\begin{bmatrix}
               \mathcal{H}_{\theta \theta} & \mathcal{H}_{\theta \phi} \\
               \mathcal{H}_{\phi \theta} & \mathcal{H}_{\phi \phi}
        \end{bmatrix} \begin{bmatrix}
               \bm{\Delta}_{\theta} \\
               \bm{\Delta}_{\phi}
        \end{bmatrix},
    \end{aligned}
\end{equation}
where as we have just claimed, the first order gradient term evaluates to zero leaving us with a Hessian $\mathcal{H}_{\epsilon^2}$ where the first rows and columns specify all $\theta$ phase shifts and the final rows and columns specify all $\phi$ phase shifts.

We now can determine the matrix elements of $\mathcal{H}_{\epsilon^2}$, which describes how both the individual phase shifter sensitivities and correlations among phase shifters contribute to the overall error. Note that based on the properties of the second-order derivative terms, we have $\mathcal{H}_{\phi\theta} = \mathcal{H}_{\theta\phi}^T$.

For phase shifters indexed at $i, j$ (corresponding to the concatenated $\bm{\phi}, \bm{\theta}$ vectors $\bm{\eta}$), we can write the formula for each element of the Hessian. We will also need to compare to simulation. To evaluate the Hessian given errors $\delta_i = \eta_i - \hat{\eta_i} \ll \eta_i$, the Hessian can be written in terms of central finite difference (which is required to compute the diagonal terms correctly):
\begin{equation}\label{eqn:hessianfd}
    \begin{aligned}
        \mathcal{H}_{ij} &:= \frac{\partial^2\epsilon^2}{\partial\delta_i \partial\delta_j} \approx \frac{\epsilon^2(\delta_i\bm{e}_i + \delta_j\bm{e}_j) - \epsilon^2(\delta_i\bm{e}_i - \delta_j\bm{e}_j)}{2\delta_i\delta_j}.
    \end{aligned}
\end{equation}
Note that to compute the Hessian here we are subtracting error contributions where perturbations go in the same direction from those where they go in the opposite direction. Note that if $i = j$, the first term is nonzero and the second term is zero.

However, we can also use the generator formula of Eq. \ref{eqn:generator} to derive an exact formula for the Hessian to avoid needing to perform highly computationally intensive tasks as Hessian finite differences in Eq. \ref{eqn:hessianfd}. We will now explicitly evaluate the Hessian matrix elements, and in the process derive the sensitivity and correlations across many phase elements in any given feedforward mesh network and expected input/output mode pair.

\begin{figure*}
    \centering
    \includegraphics[width=\linewidth]{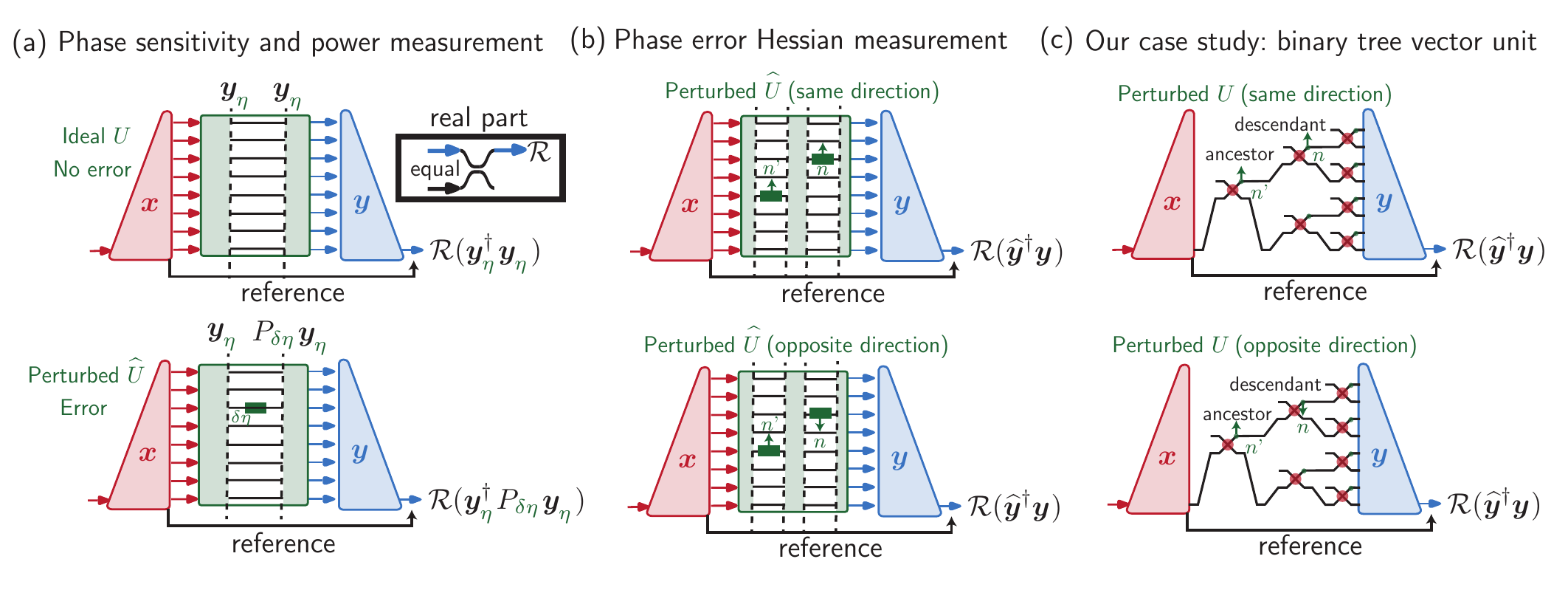}
    \caption{We show how we might experimentally implement a phase sensitivity or power monitor in a feedforward mesh (green) using an error detection circuit (blue and red). The error detection circuit assumes a perfectly ideal input generator (red) and output (blue) analyzer set to input/output mode pair $\bm{x}, \bm{y}$ respectively. As shown in the ``real part'' inset in (a), we also use an extra reference path to interferometrically determine the real part of the output signal corresponding to our mean square fidelity (or $1 - \epsilon^2$) assuming equal amplitudes at the input between reference and mesh signal (This can be adjusted using attenuation on the reference path.). (a) To measure power or sensitivity at any phase shifter, we simply perturb a phase shifter $\eta$ by some fixed amount $\delta \eta$ in the middle of the mesh and measure the corresponding change in the output power. (b) Measuring the Hessian in any feedforward mesh requires perturbing a descendant $n'$ and ancestor $n$ phase shift by the same amount and subtracting the  response resulting from going in the same and opposite directions. If there are no paths between the phase shifts, the Hessian contribution is zero, and if they are the same phase shift or $n = n'$, reduce to the case in (a). (c) The binary tree mesh problem that we consider involves sending just a single mode into a generator binary tree and evaluating the Hessian on the output vector.}
    \label{fig:sensitivityhessian}
\end{figure*}

\begin{figure*}
    \centering
    \includegraphics[width=\linewidth]{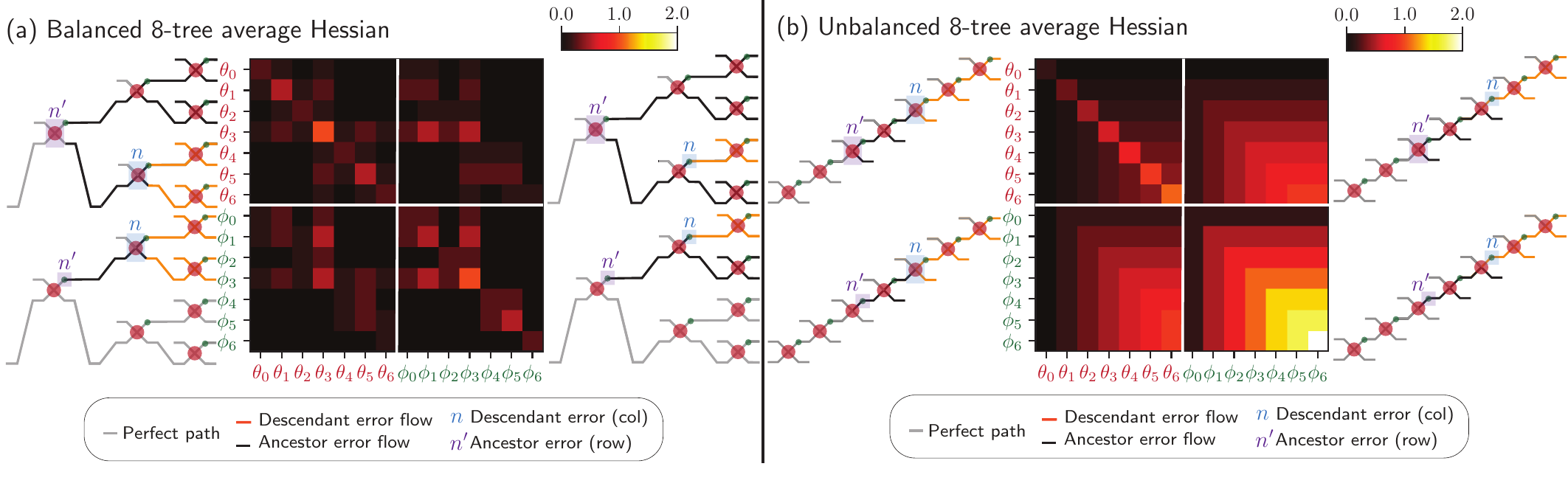}
    \caption{Assuming an MZI node as in Fig. \ref{eqn:node}, we calculate relative magnitudes of balanced (a) and unbalanced (b) Hessian terms based on Eqs. \ref{eqn:hessianbt}, \ref{eqn:hessianpower}. These results clearly depict larger correlations in unbalanced architectures and also larger sensitivities for the individual nodes (along the diagonal). Next to each plot, we include example diagrams of the correlation links between nodes lined up with the quadrants of the Hessian $\mathcal{H} = \begin{bmatrix} \mathcal{H}_{\theta \theta} & \mathcal{H}_{\theta \phi} \\ \mathcal{H}_{\phi \theta} & \mathcal{H}_{\phi \phi} \end{bmatrix}$ so it is possible to reason out why certain Hessian terms are zero (nodes not connected) or why some are larger than others (more expected power in the descendant node).}
    \label{fig:hessian}
\end{figure*}

\subsection{Sensitivity in feedforward networks}

We want to prove the following claim: ``Given some phase shifter $\eta$, a photonic feedforward mesh implementing $U$, and the input/output mode pair $\bm{x}, \bm{y}$, the sensitivity of any individual phase shifter $\eta$ in the device is equal to the power going through that phase shifter $p_\eta$ given a linear square error function as in Eq. \ref{eqn:vectorerror}.'' Note that we are given a single input and output mode pair because Eq. \ref{eqn:vectorerror} considers a single vector error, but we will consider a more general case later. Additionally, note that in previous scenarios, we considered the specific case where the desired $\bm{y} = \bm{e}_N$ ($N$th standard basis vector) for a vector unit successfully programmed to ideal $\bm{x}$. More generally, given any input $\bm{x}$ we attain some ideal outcome $\bm{y} = U \bm{x}$.

We define operators such that $U = R_\eta P_\eta L_\eta$, where $L_\eta, R_\eta$ represent operators before and after (to the left and right) of the phase shifter $\eta$ in a given device and $P_\eta$ is the operator for some applied $\eta$ phase shift, i.e., a diagonal unitary matrix where a $e^{i \eta}$ phase shift is applied to any single waveguide mode of the system as shown in Fig. \ref{fig:sensitivityhessian}(a).

Substituting into Eq. \ref{eqn:vectorerror}, we calculate the error of that phase shifter for some $\hat{\eta} = \eta + \delta \eta$ error:
\begin{equation} \label{eqn:powererror}
    \begin{aligned}
        \epsilon^2(\delta \eta) &= 2 - 2 \mathcal{R}(\bm{y}^\dagger \hat{\bm{y}})
        \\ &= 2 - 2 \mathcal{R}(\bm{x}^\dagger L_\eta^\dagger P_{-\eta} \cancelto{1}{R_\eta^\dagger R_\eta} P_{\hat{\eta}} L_\eta \bm{x}) \\
        &= 2 - 2 \mathcal{R}(\bm{x}^\dagger L_\eta^\dagger P_{\delta \eta} L_\eta \bm{x}) \\
        &= 2 - 2 \mathcal{R}(\bm{y}_\eta^\dagger P_{\delta \eta} \bm{y}_\eta) \\
        &:= 2 (1 - p_\eta \cos\delta \eta) \approx p_\eta \delta \eta^2,
    \end{aligned}
\end{equation}
where $\delta \eta = \hat{\eta} - \eta$, $p_\eta$ is the relative power in the phase shifter $\eta$ and $\bm{y}_\eta = L_\eta \bm{x}$ is the vector preceding that phase shifter as light propagates through the feedforward device. Again, note that this Eq; \ref{eqn:powererror} holds for any feedforward programmable optical device given any individual input/output mode pair $\bm{x}$ and $\bm{y} = U \bm{x}$. It is important to note that while Eq. \ref{eqn:powererror} holds specifically for $\bm{x}, \bm{y}$, but all phase shifters contribute equally to error in the overall matrix $U$ \cite{Bandyopadhyay2021HardwarePhotonics} as we later address. However, since the emphasis in this paper is that we care about only a subset of modes and not the full Hilbert space spanned by the rows of $U$, the expression of Eq. \ref{eqn:powererror} is of increased importance.

To help with understanding this concept, we show that it is possible to apply the results of Eq. \ref{eqn:powererror} to perform a direct measurement of the sensitivity (and thus monitor intermediate powers) in an arbitrary feedforward network in a direct way experimentally. This is shown diagrammatically in Fig. \ref{fig:sensitivityhessian}(a). We dump half of the light into a reference path between the original generator for $\bm{x}$ and the output of the final analyzer for $\bm{y}$ measured at the output of our feedforward mesh. We can define the ``real'' part of a signal as the contribution in the top waveguide of the final 50/50 coupler. This trivially follows from the definition of a 50/50 beamsplitter matrix when the output mode of the analyzer and the reference path light are out-of-phase by $\pi / 4$.

\subsection{Hessian sensitivity for vector units}

As shown in Fig. \ref{fig:sensitivityhessian}(b) and (c), it is possible to specify the various Hessian matrix elements by directly evaluating correlations betweden perturbations across phase shifts as they affect the mean square error. The sensitivities given by $\epsilon^2(\delta \eta)$ of Eq. \ref{eqn:powererror} are equivalent to diagonal terms $\mathcal{H}_{\eta \eta} / 2$ for the Hessian for any feedforward network. For binary tree vector unit architectures, we now consider the more general case of Hessian sensitivities that involve correlations across different phase shifters rather than just the individual phase shifters, i.e., where $\eta \neq \eta'$.

The Hessian off-diagonal terms relate to correlated errors while the on-diagonal terms relate to uncorrelated sensitivities as in Eq. \ref{eqn:powererror}. More explicitly, we can calculate the total error due to error in the individual phase shifters as shown previously in Fig. \ref{fig:node}(b) as:

\begin{equation} \label{eqn:errorcorr}
    \begin{aligned}
        \epsilon^2(\bm{\Delta}) &\approx \frac{1}{2}\bm{\Delta}^T \mathcal{H}_{\epsilon^2} \bm{\Delta} = \frac{1}{2}\sum_{\eta, \eta'}\mathcal{H}_{\eta \eta'} \delta\eta \delta\eta' \\&= \underbrace{\frac{1}{2}\sum_{\eta}\mathcal{H}_{\eta \eta} \delta\eta^2}_{\mathrm{uncorrelated}} + \underbrace{\frac{1}{2}\sum_{\eta \neq \eta'}\mathcal{H}_{\eta \eta'} \delta\eta \delta\eta'}_{\mathrm{correlated / bias}}
    \end{aligned}
\end{equation}

The statistics of the individual phase errors may be characterized using a covariance matrix, which give the distributions of $\delta\eta^2$ and $\delta\eta \delta \eta'$ and may be measured experimentally. If we assume Gaussian error or noise $\delta \eta \sim \mathcal{N}(0, \sigma_\eta)$, then $\mathbb{E}[\delta \eta^2] = \sigma_\eta^2$, where $\mathbb{E}$ represents the average or expected value. The correlation of a pair of phase errors $\delta \eta, \delta \eta'$ is given by $\mathbb{E}[\delta \eta \delta \eta']$, which is 0 only if the errors are both completely uncorrelated and centered at zero. When using an MZI node, bias may be removed such that the errors are zero by adjusting the input wavelength until the bias is zero (i.e. the expected MZI coupling is 50/50), ensuring $\mathbb{E}[\delta \eta \delta \eta'] = 0$.

In practical linear photonic network implementations, the reason the entire Hessian needs to be considered rather than just the on-diagonal uncorrelated errors is that crosstalk and wavelength errors in phase shifters or couplers might be correlated. Measuring the bandwidth of a photonic network, for example, relies on the measurement of correlated error across phase shifters in the photonic circuit. Additionally, lithography does in many cases introduce spatially correlated errors (e.g. errors in waveguide widths and coupling gaps that are spatially closer on the photonic circuit will have more correlation than those far away).

The Hessian off-diagonal terms require two pieces of information: the ancestor and the descendant phase shifts $\eta', \eta$. If the ancestor and descendant phase shifts are not connected, then there is no Hessian contribution because those phase shifts do not affect each other (This can be also deduced from explicit evaluation of the second derivative of the error with respect to the two phase shifts.). If they \textit{are} connected it should be evident that error in one phase shift will influence the error in the other phase shift.

The computation of the off-diagonal Hessian terms is more straightforward for binary tree vector units compared to other feedforward architectures that do not obey the tree property. Examples of ancestor and descendant ``error flow'' for balanced and unbalanced trees, are provided in our evaluation of the Hessian terms for 8-tree vector units (i.e., $N = 8$) in Fig. \ref{fig:hessian}(a) and (b). In the example 8-tree diagrams, we denote connected phase shifters of types $\theta \to \theta$, $\theta \to \phi$, $\phi \to \theta$, and $\phi \to \phi$ each of which account for a matrix element in each quadrant of the Hessian matrix. Following the same theme as our individual phase shifter sensitivities, the power in the descendant node is the only quantity needed for the relevant Hessian elements.

For any binary tree, the Hessian terms of Eq. \ref{eqn:hessian} given any phase shifter $\eta$ have different expressions based on whether the ancestor is an internal phase shift $\theta$ or external phase shift $\phi$:
\begin{equation} \label{eqn:hessianbt}
    \begin{aligned}
        \mathcal{H}_{\eta \eta} &= 2 p_\eta \\
        \mathcal{H}_{\theta \eta} &= \mathcal{H}_{\eta \theta} = p_\eta \\
        \mathcal{H}_{\phi \eta} &= \mathcal{H}_{\eta \phi} = 2 p_\eta
    \end{aligned}
\end{equation}
In words, ``each element of the Hessian matrix is nonzero if the phase shifters are connected, is proportional to the power going through the descendant phase shifter closer to the leaves of the tree with scaling of 1 if the ancestor is internal phase shifter and 2 if ancestor is an external phase shifter.'' This Hessian definition gives a complete picture of error flow and is ultimately a major step towards efficient and thorough error modeling of binary tree photonic architectures. Further details are provided in the Appendix.

\begin{figure*}
    \centering
    \includegraphics[width=\linewidth]{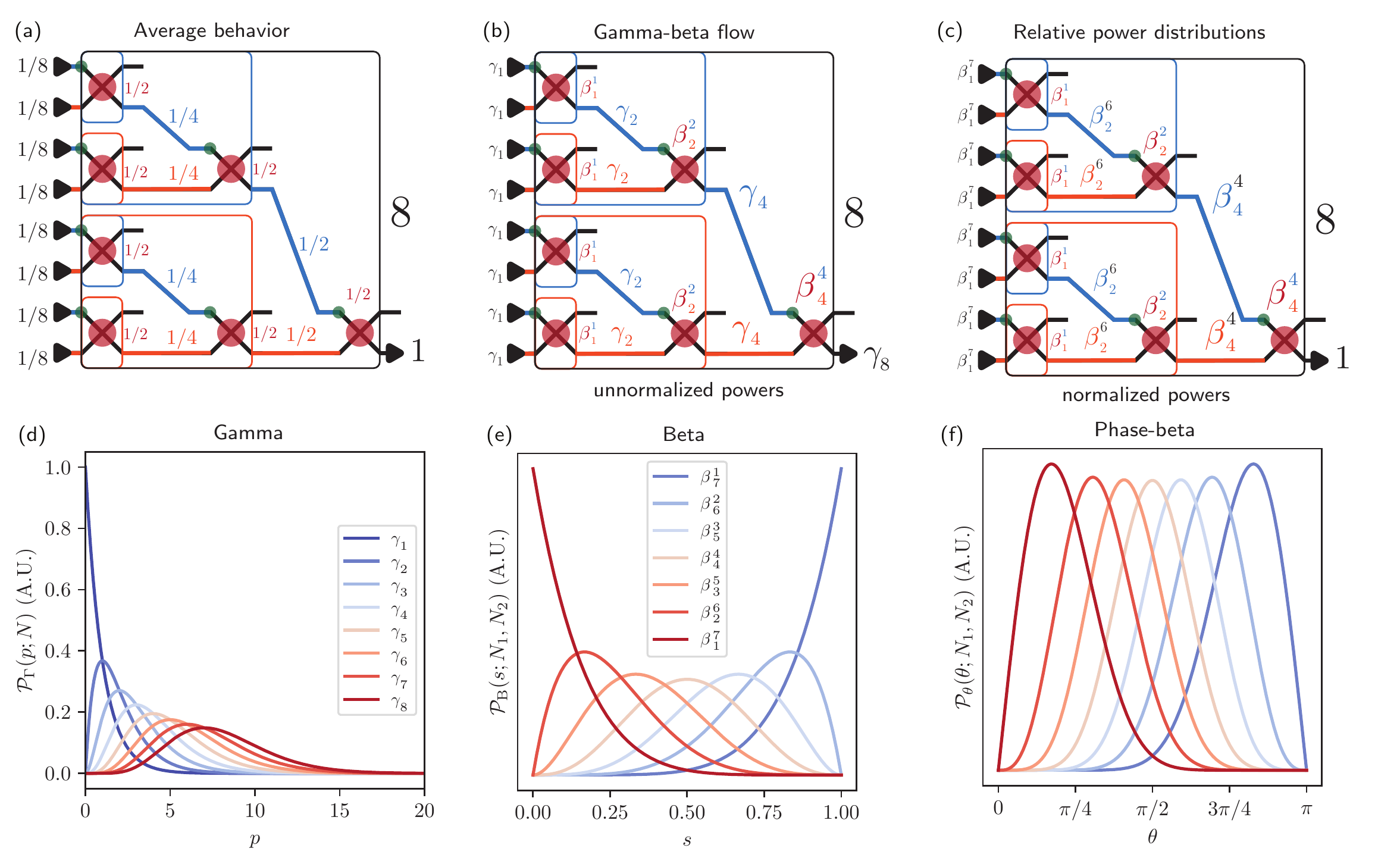}
    \caption{(a) Average distribution of power in a balanced binary tree. (b) Gamma-beta flow is a diagram labelling the statistics of $p_n$ and $s_n$ the power entering node $n$ and the transmissivity of node $n$. The statistics of the programmed analyzer settings are beta distributions ($\beta_{N_1}^{N_2}$) given gamma-distributed input powers ($\gamma_{N'}$). (c) Relative power distributions enforce the constraint that the total power in the system is 1, so the measured powers in each of the waveguides follow beta distributions ($\beta_{N'}^{N - N'}$) rather than gamma distributions ($\gamma_{N'}$). (d) Gamma distributions given an input power $p$. (e) Beta distributions given an input transmissivity or relative power $s$. (f) Phase-beta distributions (for the internal phase shifts).}
    \label{fig:treestats}
\end{figure*}

\begin{figure*}
    \centering
    \includegraphics[width=\textwidth]{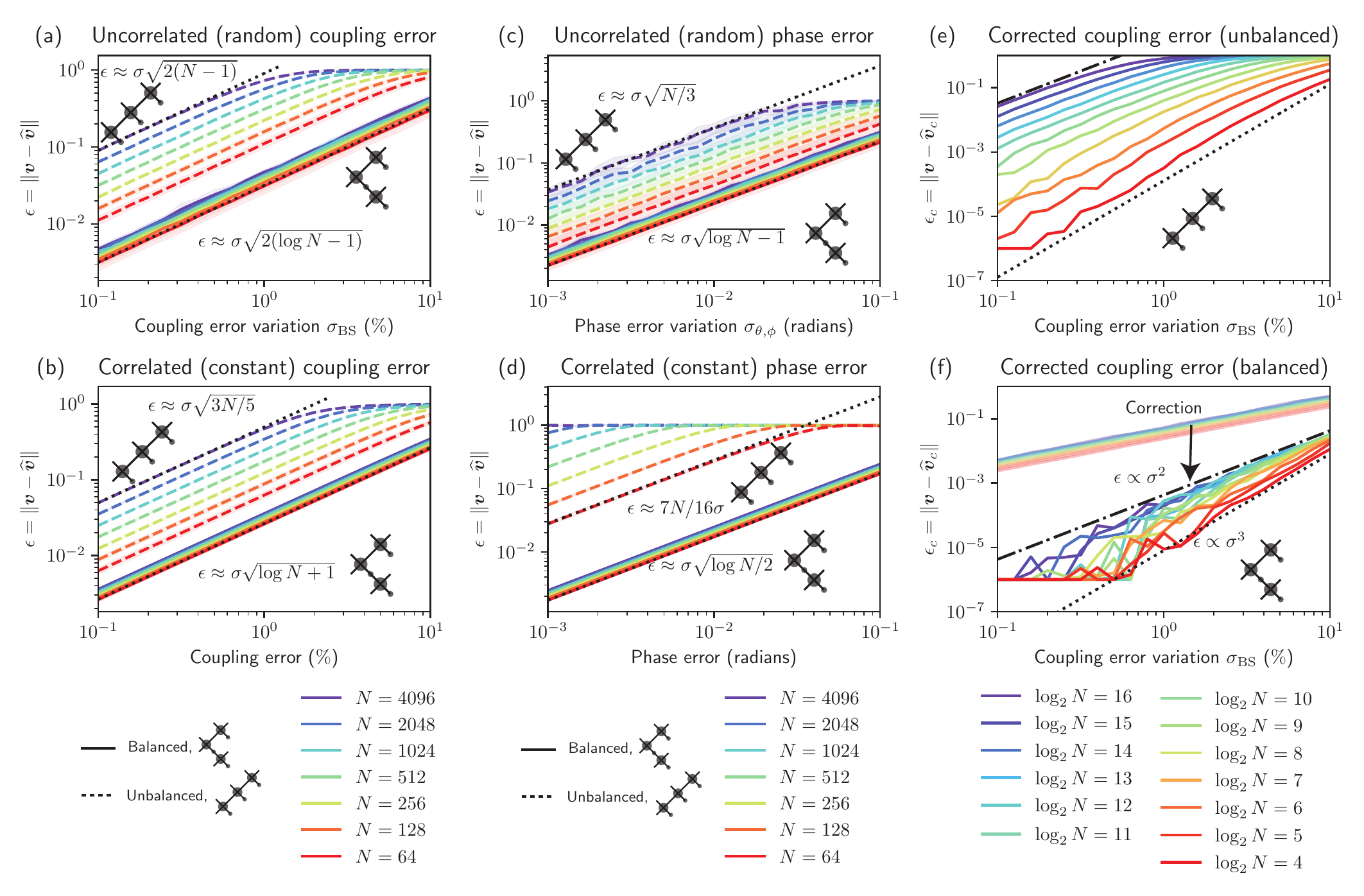}
    \caption{Here, we analyze correlated and uncorrelated error, in the presence and lack of error correction for balanced and unbalanced photonic vector units. Legends for each column of this figure are provided in the final row, specifying the various dimensions of our analysis. (a) Random coupling error for unbalanced and balanced trees from $N = 64 \to 4096$. (b) Constant coupling error for unbalanced and balanced trees from $N = 64 \to 4096$. (a) Random phase error for unbalanced and balanced trees from $N = 64 \to 4096$. (b) Constant phase error for unbalanced and balanced trees from $N = 64 \to 4096$. (e) After error correction, the unbalanced architecture now has error proportional to roughly $N \sigma^2$ for large $N$ and $N \sigma^3$ for small $N$. (f) After error correction, the balanced architecture now has error proportional to roughly $\log N \sigma^2$ for large $N$ and $\log N \sigma^3$ for small $N$.}
    \label{fig:error}
\end{figure*}

\section{Results} \label{sec:results}

We now have a prescription for determining the various sensitivities of any feedforward vector unit regardless of structure, but we still need to perform a fair comparison of robustness among various architecture choices. In this section, we define first a new statistical model of tree vector units which ultimately allows us to perform a thorough comparison of the robustness of varying photonic architectures. All calculations are performed in our simulation framework Simphox \cite{Pai2022Simphox:Library}.

\subsection{Statistical model of tree vector units}

We begin by defining how assumed random input distributions transform into powers in the binary tree architectures used to compute the Hessian as in Eq. \ref{eqn:hessian} and resulting in the average behavior of Fig. \ref{fig:hessian}(a) versus (b) (balanced versus unbalanced).

Recently, error models have been proposed for rectangular and triangular architectures under the assumption that inputs to such networks are complex normal vectors \cite{Russell2017DirectMatrices, Bandyopadhyay2021HardwarePhotonics}, which accounts for most realistic scenarios for these devices. We approach the problem from a similar angle, but with more of a statistical focus as we attempt to write a similar framework for the class of binary tree vector units, a larger class of architectures. We ultimately show that maximally balanced trees are the most robust to coupling matrix errors, both analytically and via simulation. Note we use the term ``maximally balanced'' to account for cases a binary tree is never entirely balanced, i.e. where $N \neq 2^K$ for nonnegative integer $K$.

Keeping things simple to start, consider the distribution of equal powers in what we may call the ``average'' case as shown in Fig. \ref{fig:treestats}. The average input has all equal power magnitudes entering the network. Depending on the structure of the network, for an analyzer to route the equally distributed powers to a single output, the appropriate splitting ratios must be defined at each node $n$. We need to come up with a theory to specify what these splitting ratios are as well as what happens when we no longer obey this average simple case. For this reason, we need to make some assumptions about the input distribution into the network, and this will allow us to fairly compare the error tolerances of balanced and unbalanced tree vector units.

The (standard) assumption that inputs to the network are complex normal random numbers means that the inputs are of the form $\bm{x} = \bm{a} + i \bm{b} \sim \mathcal{C}\mathcal{N}(0, 1)$, where $\mathcal{C}\mathcal{N}$ represents a complex normal distribution defined where $a_n, b_n \sim \mathcal{N}(0, 1 / 2)$ are independently distributed complex normal distributions centered at 0 with variance $1 / 2$:

\begin{equation}
    \mathcal{P}(x_n) = \frac{e^{-a_n^2 - b_n^2}}{\pi / 2}
\end{equation}

In physics, we typically work with the phasor representation $x_n = \sqrt{y_n} e^{i \varphi}$, where $y_n = a_n^2 + b_n^2$, so the above probability distribution becomes much simpler:
\begin{equation} \label{eq:exponential}
    \mathcal{P}(y_n, \varphi) = e^{-y_n}
\end{equation}
This proves that that $\varphi \sim \mathcal{U}(0, 2\pi)$ is a uniform distribution, i.e. all values between $0$ to $2\pi$ are equally likely. Additionally, the powers are exponentially distributed which by definition obeys the Gamma distribution defined as $y_n \sim \mathrm{Gamma}(1)$. More generally, the formula for a gamma-distributed random variable distributed as $\mathrm{Gamma}(N)$ is
\begin{equation} \label{eq:gamma}
    \mathcal{P}_\Gamma(y; N) := \frac{y^{N - 1}e^{-y}}{\Gamma(N)},
\end{equation}
where substituting $N = 1$ gives the exponential distribution in Eq. \ref{eq:exponential}. 

We can extend the definition in Eq. \ref{eq:gamma} even further. Specifically, we can use the convenient property that Gamma distributions, like powers in self-configured nodes, are additive. In particular, the average and variance of the distribution for $\mathrm{Gamma}(N)$ is both $N$. The sum of gamma-distributed powers add up both in average and variance to a total gamma distributed power, an interesting statistical analog of energy conservation. 

More formally in our scenario, the total power in a self-configured branch of the binary tree with $M$ inputs is distributed as $P_M \sim \mathrm{Gamma}(M)$ because it is the sum of identically distributed inputs distributed as $\mathrm{Gamma}(1)$. In Fig. \ref{fig:treestats}(b), we label edges (waveguides) with the distribution $\mathrm{Gamma}(N)$ as $\gamma_N$ edges, which gives a shorthand notation to denote the distribution of powers expected in that waveguide given a self-configured vector. The inputs are labelled $\gamma_1$ as expected by the random complex vector condition.

The final step is to find the distribution for $s, \phi$ in the binary tree nodes. A reparametrization is needed using the recursive update in Eq. \ref{eqn:selfconfiguringproof} from $(\widetilde{x}_{N_1}, \widetilde{x}_{N_2})$ to $(s, \phi, \widetilde{x}_{N})$. As we previously showed in Eq. \ref{eqn:ssc}, we can parametrize $s$ as:
\begin{equation}
    s = \frac{|\widetilde{x}_{N_1}|^2}{|\widetilde{x}_{N_1}|^2 + |\widetilde{x}_{N_2}|^2} = \frac{P_{N_1}}{P_{N_1} + P_{N_2}},
\end{equation}
where $P_{N_1} \sim \mathrm{Gamma}(N_1)$ and $P_{N_2} \sim \mathrm{Gamma}(N_2)$. Recall that $s \in [0, 1]$, so we also have $1 - s \in [0, 1]$. As before, we also have $\phi \sim \mathcal{U}(0, 2\pi)$.

The distribution for $s$ is determined using a change-of-basis, and is known as a \textit{beta distribution}, written as $s \sim \mathrm{Beta}(N_1, N_2)$. In statistics more generally, beta distribution can be thought of as a way to measure fairness of a coin given $N_1$ head trials and $N_2$ tail trials. The smaller $N_1, N_2$ are, the less certain we are of the fairness and the larger the variance of the corresponding beta distribution. Here, we propose a new analogy to optical power statistics; we apply the same concept to our physical platform where coin flip probabilities are instead represented as fractions of powers in various segments of the circuit. The beta distribution (also shown in Fig. \ref{fig:treestats}(e)) is defined as:
\begin{equation} \label{eqn:betadist}
\begin{aligned}
    \mathcal{P}_\mathrm{B}(s; N_1, N_2) &= \frac{s^{N_1 - 1} (1 - s)^{N_2 - 1}}{\mathrm{B}(N_1, N_2)} \\
    \mathrm{B}(N_1, N_2) &= \frac{\Gamma(N_1 + N_2)}{\Gamma(N_1) \Gamma(N_2)},
\end{aligned}
\end{equation}
where $\mathrm{B}(N_1, N_2)$ is just a normalization function for the beta distribution (similar to $\Gamma(N)$) that depends on the parameters of the beta distribution. Analogous to coin fairness in our above example, the beta distribution parameters $N_1, N_2$ tell us the average or  ``expected'' fraction of power expected for random variables, which is generally $\langle s\rangle = N_1 / (N_1 + N_2)$. We can further the analogy by relating coin fairness to the expected power fraction. Just as more trials increase coin fairness confidence, an increase in $N_1$ and $N_2$ (the number of inputs leading into the first and second subtrees of a node) corresponds to increased confidence (decreased variance) in allocating power to each subtree, with variance given by $\frac{N_1N_2}{(N_1 + N_2)^2 (N_1 + N_2 + 1)}$. If $N_1 = N_2 = N / 2$ (maximally balanced case), the variance is $\frac{1}{4(N+1)}$. If $N_1 = N - 1, N_2 = 1$ (maximally unbalanced case), then the variance is  $\frac{N - 1}{4N^2(N+1)}$, which is roughly a factor of $N$ less than the variance for the maximally balanced case. The increased confidence in the fraction to attribute to each subtree also decreases the error tolerance for $s$. Each node may be labelled with the notation $\beta_{N_2}^{N_1}$, which indicates that there were $N_1$ inputs that went into the top subtree and $N_2$ inputs that went into the bottom subtree, and that the results of these inputs are now being funneled or combined into this node. In summary, given input powers into a self-configured node labelled $\gamma_{N_1}, \gamma_{N_2}$ for the top and bottom, the node will be labelled $\beta_{N_2}^{N_1}$ and the output edge labelled $\gamma_{N_1 + N_2} = \gamma_N$. This process is applied recursively for $N = 8$ in Fig \ref{fig:treestats}(b, c). We also consider the change-of-variable $s \to \theta$ in Fig. \ref{fig:treestats}(f) for internal phase shifts $\theta$ based on the formula we derived earlier $s = \cos^2 \frac{\theta}{2}$:
\begin{equation} \label{eqn:phasebeta}
    \mathcal{P}_{\mathrm{B}, \theta}(\theta; N_1, N_2) = \frac{\left(\sin \frac{\theta}{2}\right)^{2N_1 - 1} \left(\cos \frac{\theta}{2}\right)^{2N_2 - 1}}{\pi \mathrm{B}(N_1, N_2)},
\end{equation}
which, in the case of $N_2 = 1$, reduces to the findings of Ref. \citenum{Russell2017DirectMatrices} for locally interacting photonic networks. We will call this the ``phase-beta distribution.''

The key difference between balanced tree and unbalanced tree architectures is in the beta distribution for leaf nodes (connecting directly to the inputs into the device) as shown in Fig. \ref{fig:treestats}(a-c).

Note that relative magnitudes in the mesh also behave like a beta distribution, which is useful for our phase and coupling sensitivity analysis. To see this, consider node $n \leq N - 1$ leading to $N'$ outputs with relative output power $p_n$:
\begin{equation} \label{eqn:waveguidestats}
    \begin{aligned}
        p_n &= \frac{P_{N'}}{P_{N'} + P_{N - N'}} \\
        p_n &\sim \mathrm{Beta}(N', N - N'),
    \end{aligned}
\end{equation}
which shows that \textit{relative} powers in the mesh follow beta distributions as well as the nodes of the mesh. Note that the powers in the mesh are not independent because of the tree structure, and therefore, these relative powers can only be used to evaluate sensitivities of the circuit to individual nodes, but not how the overall circuit responds to a perturbation of all nodes at once. We will address this point specifically later.

For specifying the power going through phase shifts $\eta$ in the device more specifically, we define the following based on previously defined $p_n, s_n$ at the various nodes:
\begin{equation} \label{eqn:phaseshiftpowers}
    \begin{aligned}
        p_{\eta} &= \begin{cases}
            p_n / 2 & \eta = \theta_n\\
            p_n s_n & \eta = \phi_n,
        \end{cases}
    \end{aligned}
\end{equation}
where the factor of $2$ comes from the fact that half the power entering the node goes through the top $\theta$ phase shift.

Interestingly, the variance of $\mathcal{P}_{\mathrm{B}, \theta}(\theta; N_1, N_2)$ scales such that it is near-constant for a given $N$, i.e. balance no longer affects the error model for $\theta$. However, the variance does depend strongly on $N$, and binary trees have a scale invariant property that half the nodes in the network always have $N = 2$, or by our notation are $\mathrm{B}_1^1$ nodes, which explains why balanced trees are so much more error tolerant and broadband as compared to unbalanced trees. More generally, balanced trees have the property that $N / 2^\ell$ nodes in column $\ell$ are $\beta_{2^\ell}^{2^\ell}$ nodes as shown in the labelled waveguides in Fig. \ref{fig:treestats}(c).

The number of possible binary tree architectures for a given $N$ is given by the Catalan number $C_{N - 1} = \frac{(2N)!}{(N + 1)! N!}$. Application of Stirling's approximation suggests that the number of possible vector unit designs scales as roughly $4^N$, and though we only consider the two extremes of these designs, our theory may be applied to any of these designs by considering an arbitrary choice of $N_1$ at each recursive step for Eqs. \ref{eqn:binarytree} and \ref{eqn:selfconfiguringproof}, and then applying the statistics of Eq. \ref{eqn:phasebeta}.


\subsection{Balanced trees are robust to phase error}

We now compare the Hessian and the error scaling in balanced and unbalanced binary tree networks, the core theoretical result of this paper. Our goal is to (1) analytically show that balanced trees are robust to phase error compared to unbalanced trees and (2) perform simulation analysis to numerically verify our scaling arguments given both coupling and phase errors.

Comparing these large class of networks requires incorporating the statistical analysis of flow of power in the network in Fig. \ref{fig:treestats}. Specifically, the Hessian elements for the 8-tree shown in Fig. \ref{fig:hessian} arise from the mean of the beta distributions representing powers in the various waveguide segments as labelled in Fig. \ref{fig:treestats}(c). As a result, we find that the Hessian matrix elements for the balanced tree are much smaller than those of the unbalanced tree. Ultimately, we attain expressions for the power in different segments of the network $p_n$ which is entirely sufficient to determine the necessary Hessian terms and even the distributions of those Hessian terms given random inputs.

Under a simplified assumption of uncorrelated error, it is possible to directly compare the performance of balanced and unbalanced trees by simply tracking the amount of light in various branches of the photonic network. In particular, we can assume that the uncorrelated error is the dominant contributing term in Eq. \ref{eqn:errorcorr} where $\delta \sim \mathcal{N}(0, \sigma^2)$. The overall sensitivity can therefore be thought of as simply the sum of optical powers in various branches of the network assuming a complex random input distribution from which $p_n$ statistics arise. 

Let us now assess the scaling relations according to this now simplified framework for uncorrelated error. We apply Eqs. \ref{eqn:waveguidestats} and \ref{eqn:phaseshiftpowers} to compare expected (average) errors $\mathbb{E}[\varepsilon^2(\bm{\Delta})]$ in balanced networks and unbalanced networks using linearity of expectation:
\begin{equation} \label{eqn:hessianpower}
    \begin{aligned}
        \mathbb{E}[\varepsilon^2(\bm{\Delta})] &= \sum_\eta \mathbb{E}[p_\eta] \\
        \mathbb{E}[\varepsilon^2_{\mathrm{bal}}(\bm{\Delta})]& = \sum_{k = 1}^{\log N} \sum_{k' = 1}^{N / 2^k} \frac{2^k}{N} \sigma^2 = \log N \sigma^2 \\
        \mathbb{E}[\varepsilon^2_{\mathrm{unbal}}(\bm{\Delta})] &= 2 \sum_{n = 1}^N \frac{n}{N} \sigma^2 \propto N \sigma^2 \\
    \end{aligned}
\end{equation}
where $N_\eta$ is the number of inputs spanned by phase shifter $\eta$'s subtree, and $\sigma$ represents uncorrelated phase shift error yielding $\log N$ scaling for balanced trees and $N$ scalings for unbalanced trees. This scaling argument is the key argument of this paper and applies not only for phases but for coupling too as we will see in our numerical analysis. This is not surprising due to the relationship between coupling and phase error uncovered by Ref. \citenum{Bandyopadhyay2021HardwarePhotonics}.

If $N$ is not a power of 2, the structure of a fully balanced architecture is slightly different but the Hessian scaling relations are similar. Additionally, we empirically find that the scaling laws for the phase shifter and for the coupler errors are roughly the same (shown in Fig. \ref{fig:error}), which as shown earlier in Eq. \ref{eqn:nodeerror} can be attributed to the fact that coupling and phase error map to each other \cite{Bandyopadhyay2021HardwarePhotonics}. One notable difference though is the constant phase error for unbalanced architectures in Fig. \ref{fig:error}(d) has an error proportional to $N$ rather than $\sqrt{N}$, and this is explained by the Hessian off-diagonal contributions as shown in Fig. \ref{fig:hessian}(b).

We have now shown how to evaluate the distribution of phase errors in any node within any vector unit implementing a random vector as well as various correlations across elements of the circuit. At a high level, the amount of light present in the ``average device'' (device given average input) is larger for an unbalanced tree as compared to a balanced tree. As previously indicated in Eq. \ref{eqn:powererror}, the power in each waveguide segment of an ``average device'' is also a measure of how sensitive that part of the circuit is and how it effectively couples its error to other parts of the circuit. This is not surprising qualitatively, but it is useful to also quantify the sensitivity of optical circuits based on powers expected in the circuit, and our Hessian formalism accomplishes exactly that. The proportion of inputs that are connected to a waveguide segment gives the statistics of the phase sensitivities of an individual node, and the unbalanced tree tends to ``load'' all of its nodes with the maximum possible inputs. later on, specifically Fig. \ref{fig:hessianstats}, we specifically plot the Hessian statistics for balanced and unbalanced trees and observe they obey the appropriate beta statistics predicted by Eq. \ref{eqn:powererror}.

\section{Mode decomposition networks} \label{sec:apps}

We now construct the unitary device operators $U$ to perform matrix-vector products $\bm{y} = U \bm{x}$ by cascading and/or interleaving vector units in various forms, which establishes the many applications of programmable photonic circuits discussed in this work. There are many useful implementations of $U$, some universal and some not universal (i.e., not all unitary matrices can be programmed).

\subsection{Binary tree cascade}

\begin{figure*}
    \centering
    \includegraphics[width=\linewidth]{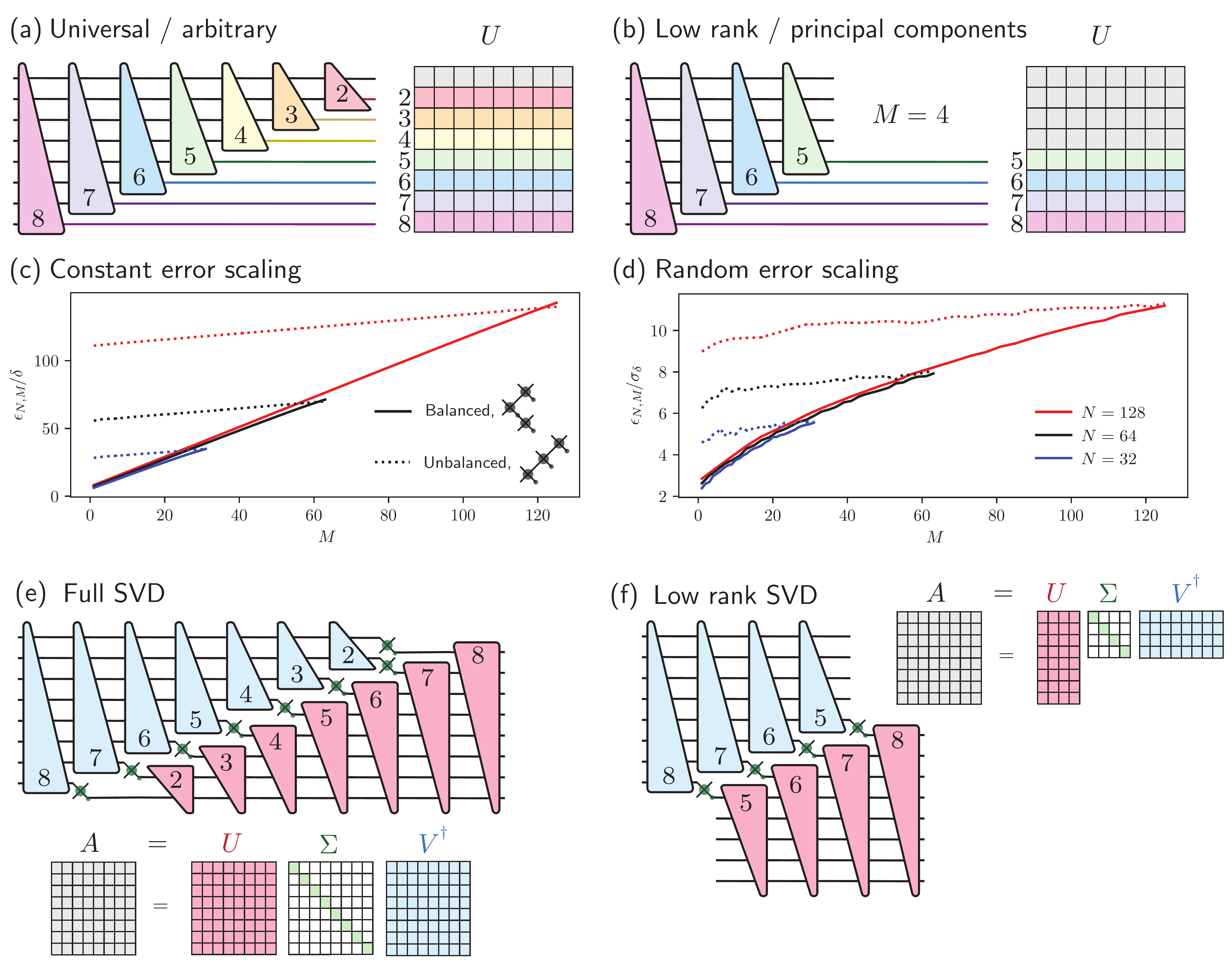}
    \caption{(a) Universal binary tree cascade consists of a sequence of analyzers that are programmed (or equivalently, mathematically computed) via self-configuration to implement rows of $U$ in order $8, 7, 6 \ldots 2$. (b) The SVD architecture implements any arbitrary unitary operator by connecting universal architectures on either end of a set of MZI attenuators with a free input and output. (c) When the error is a constant factor, the balanced architecture wins for for small $N$, but the margin decreases until it is roughly the same at $M = 3N / 4$. (d) When the error is random, the balanced architecture wins for small $N$, but the margin decreases until it is roughly the same at $M = N$. For (c, d), we compare the component error to the overall error and find that constant error is significantly larger than random error due to higher sensitivity to biased error (e.g. circuit bandwidth) versus random error (e.g. fabrication error). (e) Full singular-value decomposition consists of two universal cascades facing each other wih MZI attenuators in between. (f) Low-rank singular-value decomposition consists of two low-rank cascades facing each other with MZI attenuators connecting the relevant dimensions whose bases can be defined arbitrarily.}
    \label{fig:network}
\end{figure*}

For this section, we define a new ``multimode" error function that generalizes the ``single mode" error of Eq. \ref{eqn:vectorerror} to $M \leq N$ orthogonal basis vectors here denoted as $\bm{x}_m$, i.e. not necessarily $N$ ``full rank'' basis vectors. Specifically, we define $\epsilon_{N, M}$ as follows:
\begin{equation}
    \epsilon_{N, M}^2 = \sum_{k = 1}^M\frac{\|\bm{x}_m - \hat{\bm{x}}_m\|^2}{M} = \frac{\mathcal{R}(2 - 2\mathrm{tr}(U_M^\dagger \hat{U}_M))}{M},
\end{equation}
assuming $U_M$ is a $N \times M$ matrix with $M$ orthogonal basis vectors (so $U_M^\dagger U_M$ is $M \times M$). If $M = N$, $U$ is a square matrix and is unitary. The case $M = 1$ degenerates to our previous definition for single vectors in Eq. \ref{eqn:vectorerror}.

In a ``binary tree cascade,'' we use $M$ vector units to construct a set of $M \leq N$ normalized and mutually orthogonal basis vectors that form $U_M$ (which becomes square and unitary when $M = N$). Self-configuring networks in the default orientation defined in Eq. \ref{eqn:binarytree} can be cascaded to form universal unitary architectures, namely architectures that can be used to define any unitary matrix. This follows from the proofs in Refs. \citenum{Reck1994ExperimentalOperator, Miller2013Self-configuringInvited, Miller2013Self-aligningCoupler}, but we will describe a simplified proof here. A \textit{cascade} of self-configuring networks can be defined as:
\begin{equation} \label{eqn:cascade}
    \begin{aligned}
    U_{N, M} &\equiv \prod_{n = 1}^M R_n^{[N]}(\overline{\bm{u}}_n), \\
    R_n^{[N]} &= \begin{cases}
         R_N & n = N \\
         \begin{bmatrix}
                R_{n} & O \\
                O & I_{N - n}
        \end{bmatrix} & \mathrm{otherwise}
        \end{cases}
    \end{aligned}
\end{equation}
where $\bm{u}_n$ is the $n$th column of $U$, $\overline{\bm{u}}_n$ is the $n$th column of $U$ after passing through the first $n$ layers, and $I_{N - n}$ is an identity matrix of size $N - n$. Finally, $R_n^{[N]}$ means, in words, a self-configuring layer matrix representation over the first $n$ inputs of an $N$-waveguide rail system. More explicitly we define $\overline{\bm{u}}_m$ using the recurrence relation:

\begin{equation}
    \overline{\bm{u}}_m = \left(\prod_{n = 1}^{m - 1} R_n^{[N]}(\overline{\bm{u}}_n) \right) \bm{u}_m  \\
\end{equation}

Due to orthogonality, we always have that the last $m$ elements of $\overline{\bm{u}}_m$ are zero, since $\bm{u}_m$, a row in $U_N$ is orthogonal to all other rows in $U$, including the preceding $\bm{u}_1, \bm{u}_2, \ldots \bm{u}_{m - 1}$.

Now, we rewrite Eq. \ref{eqn:cascade} in terms of the actual parameters of the physical system that must be programmed:
\begin{equation}
    U_{N, M}(\mathbf{s}, \bm{\phi}) = \prod_{n = 1}^K D_n^{[N]}(\mathbf{s}_n, \bm{\phi}_n), \\
\end{equation}
where we have applied the self configuration process in Eqs. \ref{eqn:phisc}, \ref{eqn:ssc}, and \ref{eqn:selfconfiguringproof} for each layer in order from left-to-right starting from $R_N \to D_N$ to $R_2 \to D_2$.

While there exist architectures for any binary tree, some vector units can be more compactly cascaded than others. For instance, the unbalanced tree (diagonal line), can be cascaded to form a triangular architecture that is $2N - 3$ nodes deep. However, a balanced binary tree cannot be packed compactly and requires up to $N \log N$ photonic layers to implement. Therefore, for larger $N$, balanced binary tree cascades should ideally be used when $M \ll N$ is sufficient to solve some problem, with the key benefit being that $N$ can now be much larger than what would typically be used for a rectangular or triangular locally interacting network. In such a case, we would have a total of just $M \log N$ layers; an interesting case that warrants further investigation would be a binary tree with $M = N / \log N$ units, which would be $N$ layers deep, equal to the optical depth of a rectangular universal network with a tradeoff of fewer degrees of freedom \cite{Clements2016AnInterferometers}.

\subsection{Singular value decomposition networks}

A singular value decomposition (SVD) network, as first proposed in Ref. \cite{Miller2013Self-configuringInvited} and shown in Fig. \ref{fig:network}(e), is capable of performing any arbitrary complex linear operation (i.e., not just unitary), and is actually a specific case of the cosine-sine (CS) photonic mesh (discussed in the Appendix). The SVD architecture provides the necessary $2N^2$ degrees of freedom by decomposing a matrix in the form $A = U \Sigma V^\dagger$ where $A$ is $N \times M$ matrix, $U, V$ are $N \times N$ and $M \times M$ unitary matrices and $\Sigma$ is $N \times M$ matrix with singular values along the diagonal, i.e. $\Sigma_{nn} = \sigma_n \in \mathbb{C}$ and $\Sigma_{nm} = 0, n \neq m$. The $\Sigma$ matrix is represented as an array of $\min (M, N)$ coupling matrices $X(\theta)$ placed between the universal architectures implementing $U$ and $V^\dagger$, where each element of the coupling matrix act as an effective attenuator (tunable loss element) (with an additional phase shifter if require to implement a complex coupling element).

This feature of these networks is particularly important as we consider using wide or low-rank rather than square matrix multiplication techniques for matrix acceleration architectures. In particular, if $ M \ll N$, then a cascade of $M$ balanced tree architectures for inputs of size $N$ and a small unitary $M \times M$ architecture may be a prudent strategy for scaling up such architectures and reducing the optical depth (number of devices light has to pass) for a lower-loss and lower-systematic error device. This architecture is shown in Fig. \ref{fig:network}(f).

In signal processing applications, there are many cases where a small number of principal components are necessary to characterize some system, which can benefit from the architecture of Fig. \ref{fig:network}(f). This is certainly the case in mode conversion and telecommunications applications \cite{Miller2013EstablishingAutomatically, Annoni2017UnscramblingModes}. However, principal components analysis (PCA) is also a commonly employed tactic for reducing the dimensionality of a problem, first popularized in a 1991 paper by Turk and Pentland \cite{Turk1991EigenfacesRecognition}. Our scalability arguments suggest that future work in assessing the value of low-rank photonic architectures in principal components analysis could very possibly lead to a new paradigm for photonic network-based computing and analog signal processing that is significantly more robust and scalable to large numbers of input modes $N$.

As a final note, many of the error tolerance scaling properties for low-rank matrices in Figs. \ref{fig:error} and \ref{fig:network} apply to SVD architectures. Because random binary tree cascades also encode random unitary matrices, it is straightforward to assign $\mathbf{s}, \bm{\phi}$ to the appropriate beta distribution that depends on the number of inputs in the  cascade meshes corresponding to $V^\dagger, U$ respectively.

\section{Discussion and Conclusion}

A major theme of this work is that wide, balanced architectures are more robust than deep, unbalanced architectures by a factor of roughly $N / \log N$. This should be evident from the fact that light must pass through more components in deeper architectures, resulting in a larger sensitivity. This general trend can be explained by the theory of error tolerance provided in Sec. \ref{sec:error} which explores the effects of both uncorrelated and correlated errors on the overall photonic device performance. Our analysis in Sec. \ref{sec:error} indicates that much of the error is concentrated in waveguides that propagate more optical power and in Sec. \ref{sec:results}, it is revealed that in a balanced tree, those are the waveguides routed closer to the root of the tree. Defining shorter (i.e., less phase sensitive) and more robust components near the root of the tree is therefore paramount in scaling photonic technologies.

Equipped with this background, it is now important to discuss some specific applications where our theory may be considered, such as analog computing and sensing when the rank (number of supported optical modes) is low, i.e. $M \ll N$. The focus in this paper has been in matrix rank (Fig. \ref{fig:network}) and error tolerance (Figs. \ref{fig:error} and \ref{fig:network}), which are considered along with other application-specific criteria such as footprint, power consumption, loss, and speed.

In transceiver applications such as sensing, LIDAR, and telecommunications \cite{Bogaerts2020ProgrammableCircuits, Annoni2017UnscramblingModes}, the main factor is the speed of phase modulators as well as the scalability in the mesh needed for optical phased array communications. While speed can be a crucial factor in wireless data transfer nodes, high resolution optical phased arrays require large numbers of emitters (say, a million outputs or a $1024 \times 1024$ array), and our error analysis in this paper is highly relevant to reaching devices of such scale in photonics. In particular, as proven throughout this paper, our analysis can help elucidate the optical sensitivities of various parts of a optical phased array circuit. Such circuits are typically a balanced tree design (which may follow an H-tree fractal) with errors (sensitivities) in various waveguide segments in the passive splitter circuit leading up to the all-important final output phase shift array. Thankfully we have shown that the overall sensitivity of a balanced photonic network scales with $\log N$, making the prospects for scaling up to a million phase-controlled emitters more feasible given $\log N = 20$ layers. 

Separately, we have mentioned in the Introduction how photonic computing (e.g., matrix-vector multiplication) tasks can suffer from low error tolerance at sufficiently high circuit sizes, e.g. for machine learning \cite{Shen2017DeepCircuits, Pai2022ExperimentallyNetworks} and cryptography or blockchain \cite{Pai2022ExperimentalCryptocurrency}. Computing tasks are evaluated in terms of OPS (operations per second), which are limited by detection limits like integration time for sufficiently low signal-to-noise ratio and modulation switching speed limits for setting up inputs into the feedforward mesh. It is predicted that we might be able to achieve petaops level efficiencies using photonic mesh circuits of sufficient circuit size $N > 64$ \cite{Pai2022ExperimentalCryptocurrency, Nahmias2020PhotonicNetworks}. 

An alternative approach for photonic computing inspired by our theory is to embrace the ``wide'' over ``deep'' architectures. If we engineer an architecture to minimize the number of layers light has to propagate through, we can reduce both the overall photonic loss as well as the overall accumulated error. This concept is introduced as part of our ``splay architecture'' framework in the Appendix, which has a larger footprint but is significantly more error tolerant and possibly a less lossy scheme. This same theme is obeyed by Hadamard and FFT-like photonic mesh networks, also discussed in the Appendix, are nested binary trees that implement a subset of unitary space but use nonlocal connections to efficiently couple large numbers of modes ($N$) with fewer degrees of freedom ($N \log N$ vs $N^2$). Such efficient unitary representations have already been suggested for machine learning applications and also implement FFT and permutation operations, both of which have wide-ranging applications.

Another key advantage of wide, balanced architectures is speed. As discussed in the Appendix, fast calibration (self-configuration) is also important for any applications (including computing and possibly sensing) where interrupting device execution to re-calibrate the system is important. Self-configuration for instance is much faster in balanced trees and balanced tree cascades versus unbalanced trees and unbalanced tree cascades. In conjunction with higher tolerance to error, this means balanced trees need to be calibrated less often, and when they do need to be calibrated, the calibration is likely faster as well.

In conclusion, we find that low-rank ``wide'' photonic computing offers high error tolerance and bandwidth as well as more efficient device operation and calibration which can serve a large range of signal processing applications. Our introduction of binary tree cascades and splay networks (and analysis using our theoretical framework of binary tree architectures) is a key step to realizing higher-scale photonic circuits.

\section*{Acknowledgements}

We would like to thank Nathnael Abebe, Annie Kroo, Dirk Englund, Ryan Hamerly, Srikrishna Vadlamani, and Saumil Bandyopadhyay for helpful discussions. Additionally, we would also like to acknowledge funding from Air Force Office of Scientific Research (AFOSR) grants FA9550-17-1-0002 in collaboration with UT Austin and FA9550-18-1-0186 through which we share a close collaboration with UC Davis under Dr. Ben Yoo.

\section*{Data and software}
All software for running the simulations and Hessian calculations are available via Simphox \cite{Pai2022Simphox:Library}.

\appendix

\section{Self-configuration of a node}
As mentioned previously, a node guides light from its left ports, which we assume to be of the form $\bm{x} = (x_1, x_2)$ into either of its right ports $\bm{y} = (y_1, y_2)$. In other words, if $\bm{y} = T(s, \phi) \bm{x}$, we should be able to find $s, \phi$ such that $y_1 = 0$.

We first minimize $y_1$ with respect to $\phi$:
\begin{equation}\label{eqn:phisc}
\begin{aligned}
y_1 &= e^{i \phi} \sqrt{1 - s} x_{1} +  \sqrt{s} x_{2}\\
|y_1|^2 &= (1 - s) |x_1|^2 + s |x_2|^2 \\& \hspace{-0.1cm} + 2 \sqrt{s (1 - s)}|x_1||x_2| \mathrm{Re}(e^{i\mathrm{arg}(x_1)}e^{-i\mathrm{arg}(x_2)}e^{i\phi'})\\
\phi &:= \displaystyle{\min_{\phi' \in [0, 2\pi)} |y_1|^2} = -\arg\left(\frac{x_1}{x_2}\right),
\end{aligned}
\end{equation}
where $\arg(\cdot)$ refers to measuring the angle or phase of the quantity.

We then minimize with respect to $s$:
\begin{equation}\label{eqn:ssc}
\begin{aligned}
|y_1|^2 &= \left(\sqrt{1 - s} |x_1| - \sqrt{s} |x_2|\right)^2 \stackrel{?}{=} 0\\
s &= \frac{|x_1|^2}{|x_1|^2 + |x_2|^2} = \cos^2 \frac{\theta}{2} \\
\theta &= 2\arccos{\sqrt{s}}.
\end{aligned}
\end{equation}

In many practical cases, such as machine learning inference \cite{Shen2017DeepCircuits}, a lookup table or calibration curve generated by a phase calibration can be sufficient and quite stable over long periods of time, \cite{Taballione20188x8Waveguides} and these generally require more explicit error correction \cite{Bandyopadhyay2021HardwarePhotonics}. However, self-configuration is particularly useful for error correction and cases where the network needs to change often in response to external environmental cues, e.g. sensing and dynamic training of the network which requires no additional calculation beyond the input signal. Crucially, this definition differs from some previous work in this area \cite{Hamerly2021AccurateInterferometers, Pai2020ParallelNetwork}, which obey a different definition of configuration where additional calculations need to be done off-chip. This aspect of ``no additional calculation'' is an important requirement of self-configuration as defined here, which renders rectangular architectures \cite{Clements2016AnInterferometers} non-self-configurable despite their low device depth that confers clear advantages for machine learning inference \cite{Shen2017DeepCircuits} and quantum computing \cite{Arrazola2021QuantumChip} applications. 

\section{Error modeling}

The self-configuring networks in this paper automatically implement error-corrected values for $s, \theta$ because self-configuration is an inherently model-free programming approach \cite{Miller2013Self-configuringInvited, Hamerly2021InfinitelyInterferometers}. The model-free error correction approach of self-configuring networks can also be extended to any feedforward network \cite{Pai2020ParallelNetwork} via a procedure called ``parallel nullification'' or in this case ``parallel error correction'' without any need for calibration. This is a particularly important consideration for rectangular networks which are not self-configurable but offer advantages in loss variation balancing and optical depth.

In any case, error-corrected values are still needed to simulate the optimal performance after hardware correction, and thus we still consider the explicit phase values required to program error-corrected hardware. Assuming that the directional coupler or MMI is symmetric and the error in gap and/or waveguide width is also correspondingly symmetric, the error correction can be simplified from full coupled waveguide theory \cite{Huang1994Coupled-modeOverview} by a modified expression for $s$ as follows (following a similar calculation as Ref. \citenum{Bandyopadhyay2021HardwarePhotonics}, but not separating the $\delta$ error terms out):
\begin{equation}
    \begin{aligned}
        T_\mathrm{MZI}(\theta, \phi) &= \begin{bmatrix}
      C_r & iS_r \\
      iS_r & C_r
\end{bmatrix} \begin{bmatrix}
      e^{i \theta} & 0 \\
      0 & 1
\end{bmatrix} \begin{bmatrix}
      C_\ell & iS_\ell \\
      iS_\ell & C_\ell
\end{bmatrix} \begin{bmatrix}
      e^{i \phi} & 0 \\
      0 & 1
\end{bmatrix} \\
        s &= C_\ell^2S_r^2 + C_r^2S_\ell^2 + 2 C_\ell S_\ell C_rS_r \cos \theta \\
        \theta(s) &= \arccos\left(\frac{s - C_\ell^2S_r^2 - C_r^2S_\ell^2}{2 C_\ell S_\ell C_rS_r}\right)\\
        \phi &= -\arg\left(-\frac{x_1 \cdot (S_rS_\ell + C_rC_\ell e^{i \theta})}{x_2\cdot i(C_rS_\ell + S_rC_\ell e^{i \theta})}\right),
    \end{aligned}
\end{equation}
where $\ell, r$ refer to left and right beamsplitters, and $C_j = \cos\left(\frac{\pi}{4} + \delta_j\right), S_j = \sin\left(\frac{\pi}{4} + \delta_j\right)$ are the matrix elements of the beamsplitter. Note $\delta_\ell, \delta_r$ are the left and right phase-parametrized beamsplitter errors as $\alpha, \beta$ is defined in Ref. \citenum{Bandyopadhyay2021HardwarePhotonics}. With some trigonometric identities, it can be shown that:

\begin{equation}
    \sin^2(\delta_\ell -\delta_r)  \leq s \leq \cos^2(\delta_\ell +\delta_r)
\end{equation}

Note that now, in the case of an imperfect splitter, perfect (or near-perfect) self-configuration requires a true two-parameter optimization over $(\theta, \phi)$, owing to the ``tearing'' transformation of the Bloch or Riemann sphere given hardware error \cite{Miller2015PerfectComponents, Bandyopadhyay2021HardwarePhotonics, Hamerly2021InfinitelyInterferometers}. Self-configuration also considers a broader set of implementations of nodes including MZI nodes with asymmetric broadband splitters \cite{Cabanillas2019ExperimentalCouplers, Lu2015ComparisonPlatforms} that do not have the same phase error correction scheme used in Ref. \citenum{Bandyopadhyay2021HardwarePhotonics}, but for simplicity in modelling, we will only consider symmetric directional coupler-based splitters.

In the case of ``correlated errors'' where the two couplers have identical split error, we enforce $\delta = \delta_\ell = \delta_r$, $C = C_\ell = C_r$ and $S = S_\ell = S_r$:
\begin{equation} \label{eqn:nodeerror}
    \begin{aligned}
        s &= 2C^2S^2(1 + \cos \theta) = 4C^2S^2 \cos^2 \frac{\theta}{2} \\
        \theta(s) &= 2\arccos \frac{\sqrt{s}}{2CS} := 2\arccos \sqrt{\hat s}\\
        \phi &= -\arg\left(-\frac{x_1 \cdot (S^2 + C^2 e^{i \theta})}{x_2\cdot i(CS + SC e^{i \theta})}\right)
    \end{aligned}
\end{equation}
which reduces to the form of Eq. \ref{eqn:selfconfiguringproof} only when $C^2 = S^2 = 0.5$ (perfect 50/50 splitting).

As is evident from Eq. \ref{eqn:nodeerror}, splitting errors to the individual (passive) splitters of MZI nodes, while likely correlated in fabrication processes, can result in an upper limit $s \leq s_{\mathrm{max}} = 4C^2S^2 < 1$ \cite{Pai2019MatrixDevices}. This means that when $s > s_{\mathrm{max}}$ (the ``forbidden region''), the error-corrected MZI is programmed to its limit $s = s_{\mathrm{max}}, \theta = 0$ \cite{Bandyopadhyay2021HardwarePhotonics}. Since a large proportion of useful matrices require achieving cross state, an alternate MZI+Crossing architecture (adding a crossing element at the input) has been proposed such that $s \geq s_{\mathrm{min}} = 1 - 4C^2S^2$ \cite{Hamerly2021InfinitelyInterferometers}. In this case, when we desire $s < s_{\mathrm{min}}$, then $\theta$ is programmed such that $s = s_{\mathrm{min}}, \theta = 0$. This addition is mainly helpful for locally interacting meshes with MZI nodes such as triangular and rectangular meshes due to theoretical considerations which we discuss in the context of nonlocally interacting meshes in the main text.

Another type of error is dispersion error, which is useful to consider when multiple wavelengths are sent into the photonic network to perform some computation in parallel or in sensing applications that require larger bandwidth. We model the dispersion based on a given wavelength $\lambda \neq \lambda_c$, for a center wavelength, e.g. $\lambda_c = 1.55$ $\mu$m. When making a dispersion model, we find an expression for $\delta(\lambda)$, where $\delta$ is the beat phase error in the MMI or directional coupler parametrizing the splitting amplitudes $C^2, S^2$ defined as above. Similarly, there is an expression for the phase error $\theta(\lambda) = \theta(\lambda_c) + \delta_\theta(\lambda)$. All of these calculations can be done using a mode solver, which gives both $\delta_\theta(\lambda)$ and $\delta(\lambda)$, substituted into Eq. \ref{eqn:mzinode}.

\section{Correlated errors in binary trees}

In this section, we prove off-diagonal Hessian terms for binary tree structures, as claimed in Eq. \ref{eqn:hessianbt}. Unlike before, where we just had to consider a single phase shifter's sensitivity, we are interested here in how two independent phase shifters affect each other.

The key point to realize is that phase shifters that affect each other will correspond to Hessian nonzero terms equal to the power in the descendant phase shifter (belonging to nodes deeper in the recursion of Eq. \ref{eqn:binarytree}  closer to the leaves of the tree). To get the final Hessian matrix element, the relative power in the descendant phase shifter is multiplied by a factor of 1 or 2 depending on whether it is a $\theta$ or $\phi$ single-mode phase shifter. We will now prove this statement by explicitly evaluating the second-order derivatives of the error function with respect to the descendant phase shift.

First, to simplify our problem, we implement the ``clipping'' trick of Fig. \ref{fig:sensitivityhessian}(a) where we evaluate the error of the vector along the same column as descendant phase shifter $\eta$, which may be either in the internal arm or external arm of some MZI. As proven in Eq. \ref{eqn:powererror}, this is equal to the overall error assuming all other phase shifters are perfect. In other words, define $\epsilon^2(\delta\eta', \delta\eta)$:
\begin{equation}
    \begin{aligned}
        \epsilon^2(\delta\eta', \delta\eta) &= 2 - 2 \mathcal{R}(\bm{y}^\dagger \hat{\bm{y}}) \\
        &= 2 - 2 \mathcal{R}(\bm{y}_\eta^\dagger \hat{\bm{y}}_\eta)
    \end{aligned}
\end{equation}
where as before, we ``clip'' the architecture at the same point as the descendant phase shifter giving the vector $\bm{y}_\eta$ and only evaluate the error here as it is equivalent to the overall error. 

We also assume that the individual $\eta$ phase shifters corresponds to some number $y_\eta$ which is an element of the vector $\bm{y}_\eta$. Note that the power through phase shifter $\eta$ in the overall vector unit is simply $p_\eta = |y_\eta|^2$. The second derivatives only depend on this single $y_\eta$ term. We are now ready to evaluate the off-diagonal Hessian terms which in the main text we have claimed obey $\mathcal{H}_{\theta \to \eta} = p_\eta$ and $\mathcal{H}_{\phi \to \eta} = 2 p_\eta$.

Considering these two cases where $\eta$ is a descendant phase shifter of $\eta' = \theta_n$ or $\eta' = \phi_n$, the Hessian off-diagonal terms $\mathcal{H}_{\eta\eta'}$ in a binary tree is given direct second-derivative evaluation by:
\begin{equation} \label{eqn:hessianproof}
    \begin{aligned}
        \mathcal{H}_{\theta\eta} &= \frac{\partial^2 \epsilon^2}{\partial \eta \partial \theta_n} = 2 p_\eta \left(\frac{i}{2} e^{i \frac{\delta\theta_n}{2}} \right) \cdot i e^{i \delta \eta} + \cdots \Bigg|_{\bm{\Delta} = 0} = p_\eta \\
        \mathcal{H}_{\phi\eta} &= \frac{\partial^2 \epsilon^2}{\partial \eta\partial \phi_n} = 2 p_\eta \left(i e^{i \delta \phi_n} \right) \cdot i e^{i \delta \eta} \Bigg|_{\bm{\Delta} = 0} = 2 p_\eta,
    \end{aligned}
\end{equation}
where the vector of phase errors $\bm{\Delta}$ is defined as in Eq. \ref{eqn:hessian} and the $\cdots$ indicates terms that evaluate to zero when computing the real part at $\bm{\Delta} = 0$. To further clarify this last point, we take the derivative of $\cos \frac{\theta_n}{2} e^{i \frac{\theta_n}{2}}$ or $\sin \frac{\theta_n}{2} e^{i \frac{\theta_n}{2}}$ depending on the location of the descendant phase shifter $\eta$ in the tree. Only the $e^{i \frac{\theta_n}{2}}$ derivative in the product rule contributes because of the $i / 2$ leading term in the evaluated derivative. The other term (the $\cdots$ in Eq. \ref{eqn:hessianproof}) ends up being completely imaginary when $\delta \theta_n = 0$ after being multiplied by $i e^{i \delta \eta}$ term from taking the derivative with respect to $\eta$.

\section{Hessian evaluation and statistics}

In this section we show how the Hessian statistical distributions of Fig. \ref{fig:hessianstats} can be evaluated and how the relevant statistics can be determined. First, using the expressions of Eq. \ref{eqn:hessianbt}, we arrive at the power formulas for any vector unit are shown in Tbl. \ref{tbl:hessian} using Eq. \ref{eqn:hessianpower}, which are just scalar magnitudes of 0, 1, 0.5, and 2 multiplied by beta distributed powers $p_n$ and $p_n s_n$. Thus, using this table we can verify the simulated distributions of Hessian magnitudes shown in Fig. \ref{fig:hessianstats}. The resulting distribution values are from 0 to 2 along the $y$-axis, with the maximum value given by the scalar factor multiplied by the power.

\begin{table}[h]
\begin{tabular}{l|l|l|l}
$n' \to n$ & same node $n = n'$ & top tree $n$ & bottom tree $n$ \\ \hline
$\theta \to \theta$                                                          & $p_n$     & $p_n / 2$                   & $p_n / 2$                    \\\hline
$\theta \to \phi$                                                            & $p_ns_n$        & $p_ns_n$               & $p_ns_n$                \\\hline
$\phi \to \theta$                                                            & $p_ns_n$        & $p_n$                   & $0$                      \\\hline
$\phi \to \phi$                                                              & $2p_ns_n$ & $2p_ns_n$                      & $0$                     
\end{tabular}
\caption{The Hessian magnitudes for any vector unit (generator configuration) based on powers entering the node ($p_n$) or exiting the top output ($p_n s_n$). The top tree corresponds to $n \in \mathcal{T}_{n'}$ and the bottom tree corresponds to $n \in \mathcal{B}_{n'}$.} \label{tbl:hessian}
\end{table}

\begin{figure*}
    \centering
    \includegraphics[width=\linewidth]{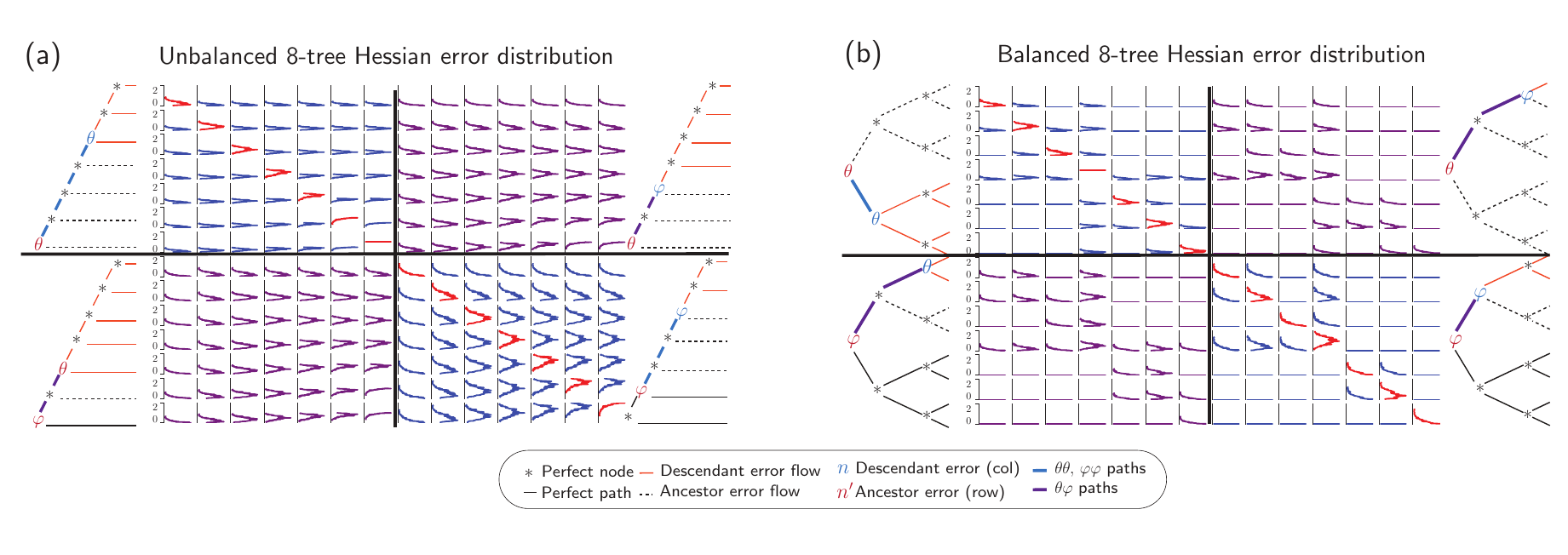}
    \caption{Beta distributed powers of Hessians for unbalanced (a) and balanced (b) trees based on Eq. \ref{eqn:hessianpower}, which map to the same matrix elements of Fig. \ref{fig:hessian}(a) and (b). Note that as labelled on the left, all $y$-axes range from 0 to 2. As in Fig. \ref{fig:hessian}, we include example diagrams of the correlation links between nodes, signifying how error in various network parameters affect each other. These diagrams are lined up with the various rows of the matrix so it is possible to reason out why certain Hessian terms are zero or why some are larger than others.}
    \label{fig:hessianstats}
\end{figure*}

\section{Comparison with other networks}

The binary tree cascade can benefit from comparisons to other types of architectures in terms of depth.

\begin{table*}[]
    \centering
    \begin{tabular}{l|c|l|l|c|c}
    \textbf{Architecture}    & \textbf{Functionality}     & \textbf{Depth} & \textbf{DoF}   & \textbf{Notation}                                & \textbf{Key feature} \\
    \hline
    Balanced tree \cite{Miller2013Self-aligningCoupler, Miller2020AnalyzingNetworks}           & $N$-vector  & $\log N$  &    $2N - 1$   & $D_N$       & Self-configurable, broadband                         \\
    Unbalanced tree \cite{Miller2013Self-aligningCoupler, Miller2020AnalyzingNetworks}            & $N$-vector  & $N$  &       $2N - 1 $    & $D_N$       & Self-configurable, narrowband                         \\
    Rectangular (Clements) \cite{Clements2016AnInterferometers}   & $N$-unitary & $N$  &    $N^2$       & $\Pi_N$     & Universal, loss-balanced, low-depth                          \\
    Triangular (Reck) \cite{Reck1994ExperimentalOperator, Miller2013Self-configuringInvited}        & Any $N$-unitary & $2N - 3$  &  $N^2$     & $\Lambda_N$ & Universal, self-configurable                         \\
    Balanced tree cascade   & $M \ll N$ basis vectors & $M (\log N - 1)$ & $NM$  & $U_{N, M}$    & Self-configurable, broadband                         \\
    Butterfly \cite{Flamini2017BenchmarkingProcessing}               & $N$-FFT, $N$-vector     & $\log N$ &   $N \log N$     & $F_N$       & Self-configurable, 1D/2D FFT, broadband \\
    Benes (double-butterfly) & $N$-FFT/$(N, 2)$-cascade   & $2 \log N$ &  $2 N \log N$  & $B_N$       & Permutation, 1D/2D conv, broadband     \\
    Cosine-sine \cite{Mottonen2004QuantumGates, Basani2022AInterferometers}  & $N$-unitary  & $N$  &    $N^2$    & $A_N$      & Universal, ($N / 2$)-SVD\\
    Splay \cite{Mottonen2004QuantumGates}  & $N$-unitary  & $3\log N$  &    $2N^2$    & $W_N$      & Any complex matrix, low-depth
    \end{tabular}
    \caption{Summary of photonic architectures, including both vector units and matrix units.}
    \label{tab:summary}
\end{table*}

\begin{figure*}
    \centering
    \includegraphics[width=\linewidth]{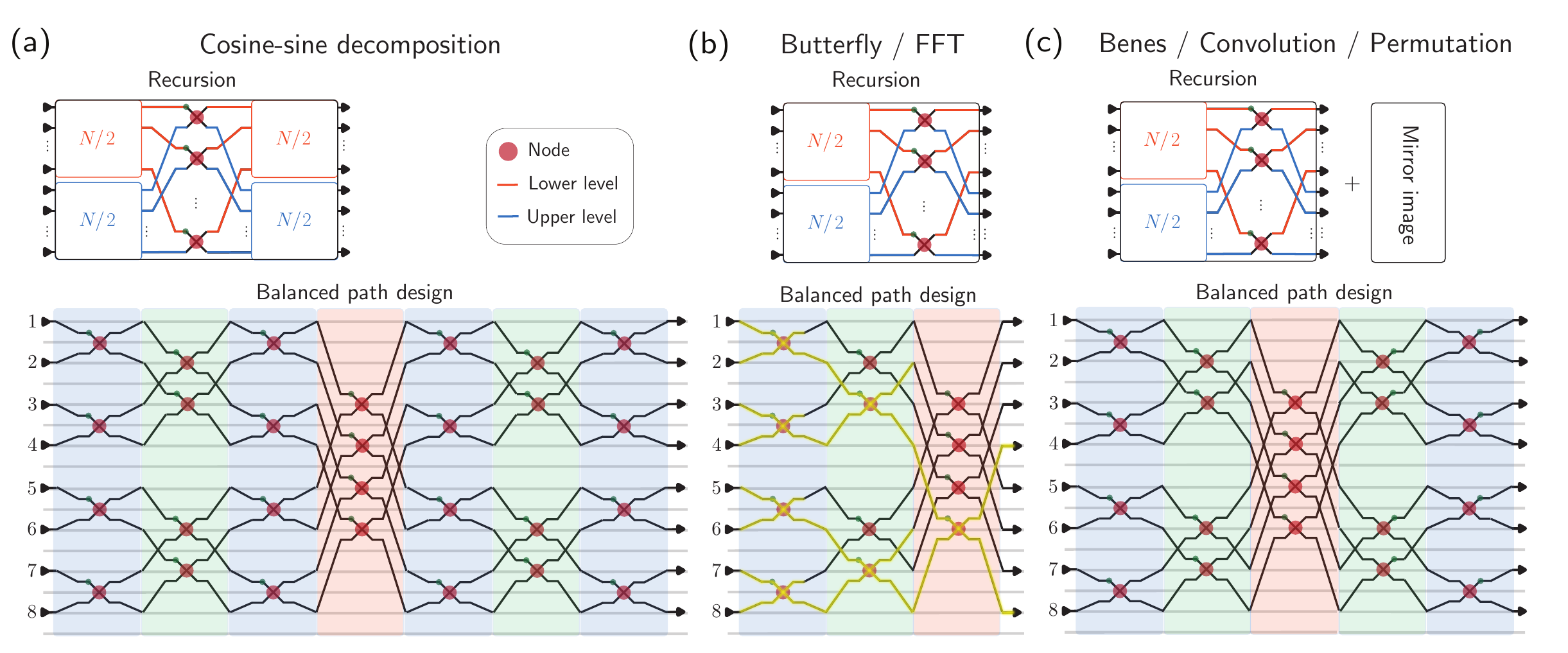}
    \caption{(a) The universal cosine-sine decomposition architecture. (b) The non-universal FFT butterfly architecture is simply a nested binary tree as seen by the recursion and thus has the same error scaling as a balanced binary tree. An example binary tree within the nested structure is highlighted in yellow. (c) The non-universal Benes architecture. Note: an additional orange (largest stride) layer may be required in case an arbitrary $N$-dimensional 1D convolution is desired.}
    \label{fig:matrixunit}
\end{figure*}

\subsection{Butterfly network}

In practical designs of ``balanced'' binary tree architectures, there are many potential issues to consider including routing nonlocal interactions and balancing optical loss across the circuit, and especially across layers. Additionally, there is a necessity to improve the compute density by implementing as many degrees of freedom in the allotted architecture as possible within a given footprint. While not universal, butterfly networks address some of these challenges \cite{Flamini2017BenchmarkingProcessing, Jing2017TunableRNNs}, and as found in Ref. \cite{Flamini2017BenchmarkingProcessing} have $\log N$ error scaling.

Balanced binary tree vector units are actually ``subgraphs'' of butterfly (or FFT) architectures shown in Fig. \ref{fig:matrixunit}(b). The butterfly architecture is an architecture of $L = \log N$ columns which interferes waveguides nonlocally at intervals of $2^\ell$ for $\ell = 1$ to $L$, which is more-or-less a nested binary tree architecture. Specificially, butterfly architectures implement $N / 2$ nested binary tree architectures in the most compact form possible, using $N / 2$ MZIs in each column (the most possible MZIs in such columns), which is indicated in the case of one such binary tree path in orange in the bottom panel of Fig. \ref{fig:matrixunit}(b). Such architectures therefore make the best use of the provided area so that each column of the optical network has the same number of nodes, and therefore the most degrees of freedom in that area. This can be defined recursively along the lines of Eq. \ref{eqn:binarytree} as follows:

\begin{equation} \label{eqn:butterfly}
    \begin{aligned}
        F_N(\mathbf{s}_N,\bm{\phi}_N) &= \Sigma_{N}(\mathbf{s}, \bm{\phi})\begin{bmatrix}
                F_{N / 2} & 0 \\
                0 & F_{N / 2}
        \end{bmatrix} \\
        \mathbf{s}_N &= [\mathbf{s}, \mathbf{s}_{N / 2}, \mathbf{s}_{N / 2}] \\
        \bm{\phi}_N &= [\bm{\phi}, \bm{\phi}_{N / 2}, \bm{\phi}_{N / 2}] \\
        \Sigma_{N}(\mathbf{s}, \bm{\phi}) &:= \prod_{n = 1}^{N / 2} T_{n, N / 2 + n}^{[N]}(s_n, \phi_n)
    \end{aligned}
\end{equation}

Comparing Eq. \ref{eqn:butterfly} to our earlier recursive definition Eq. \ref{eqn:binarytree}, the only difference in the recursive stem is the use of $T^{[N]}_{n, N / 2}$ instead of $\Sigma_N$, which is simply multiplying a column of $N / 2$ MZIs. Interestingly, this directly suggests that we actually have a nested binary tree which consists of the maximum number of root nodes given $N$, i.e. $N / 2$. This loss-balanced representation is particularly convenient for designing cascaded binary trees. Owing to the fact that butterfly networks are nested binary trees, cascaded balanced trees can be also achieved by cascading subnetworks (subgraphs) of FFT-style or butterfly photonic networks, generally with $N = 2^L$ for integer number of columns $L$. Similar architectures have been proposed and evaluated for photonic loss, robustness and other characteristics \cite{Flamini2017BenchmarkingProcessing, Fang2019DesignImprecisions}. It is reasonable to conclude that the statistical modelling (and thus the error more generally) for a butterfly network is identical to that of the binary tree above, in particular if we consider each of the individual vectors of the matrix alone.

Connecting two butterfly architectures back-to-back, for example as shown in Fig \ref{fig:matrixunit}(c) for $(N, K) = (8, 2)$, forms a ``Benes network,'' which is an architecture typically used in telecommunications capable of routing any $N \times N$ permutation. Interestingly, the Benes network can also be modified to also allow a convolution if attenuators and/or phase shifts are placed into a Benes network, since a convolution can be written in terms of a Fourier transform, elementwise multiply and inverse Fourier transform. Critically, we now have a low-depth architecture which can perform rank 2 matrix multiplication ($N \times 2$ SVD architecture using a binary tree subgraph), any permutation matrix, an FFT matrix, and convolutions. Thus, simply doubling the layers in a butterfly network opens the door to a host of new and useful computation without requiring universality.

As for photonic waveguide crossing routing, which is the biggest hurdle to realizing architectures such as the butterfly architecture that have nonlocally interacting waveguides, we suggest the use of a two-photonic-layer approach to avoid excessive high-loss crossings \cite{Chiles2017Multi-planarLoss}. Some CMOS foundries that support photonic integration might provide the option to implement an \textit{escalator}, which transfers light from a lower silicon layer to higher silicon nitride layers, with generally low loss (under 0.05 dB). If the waveguide turns up (goes to a rail assigned a lower index that the current rail) then we use an escalator to route those waveguides over any crossing waveguide that turns down and then de-escalate back down to the silicon layer for input into the next column of nodes. A key implementation detail is the need for integrated path length matching or dispersion compensation for each MZI, which may require using multiple silicon nitride layers or tunable dispersion compensation \cite{Bandyopadhyay2021HardwarePhotonics}. The alternative, using planar waveguide crossings, is likely not scalable due to the large number of required crossings in each layer (up to $N / 2$) and the non-negligible 0.1 dB loss per crossing \cite{Hamerly2021InfinitelyInterferometers}.

There are two methods for self-configuring a butterfly architecture implementing $U_N \in \mathrm{U}(N)$, where $N = 2^L$ for some integer (optical depth) $L$. One method involves parallel nullification of the vertical layers of the butterfly architecture requiring the input of just $L = \log N$ vectors \cite{Pai2020ParallelNetwork}. The other method involves tuning a unitary operator based on the first $N / 2$ columns of the matrix itself using the self-configuration approach in Eq. \ref{eqn:selfconfiguringproof}.

At each step of the algorithm, we send in photonic vectors for each column of the matrix $U$ and we perform the standard binary tree nullification routine \cite{Miller2013Self-aligningCoupler} until all light gathers at the appropriate input (indexed by the column vector index). Starting from the second vector of the nullification procedure, there will be MZIs in the light path that are already calibrated during the configuration process. However, there will always also be uncalibrated MZIs in the light path until $N / 2$ vectors of the matrix have been shined in.

\subsection{Cosine-sine decomposition matrix unit}

\begin{figure*}
    \centering
    \includegraphics[width=\textwidth]{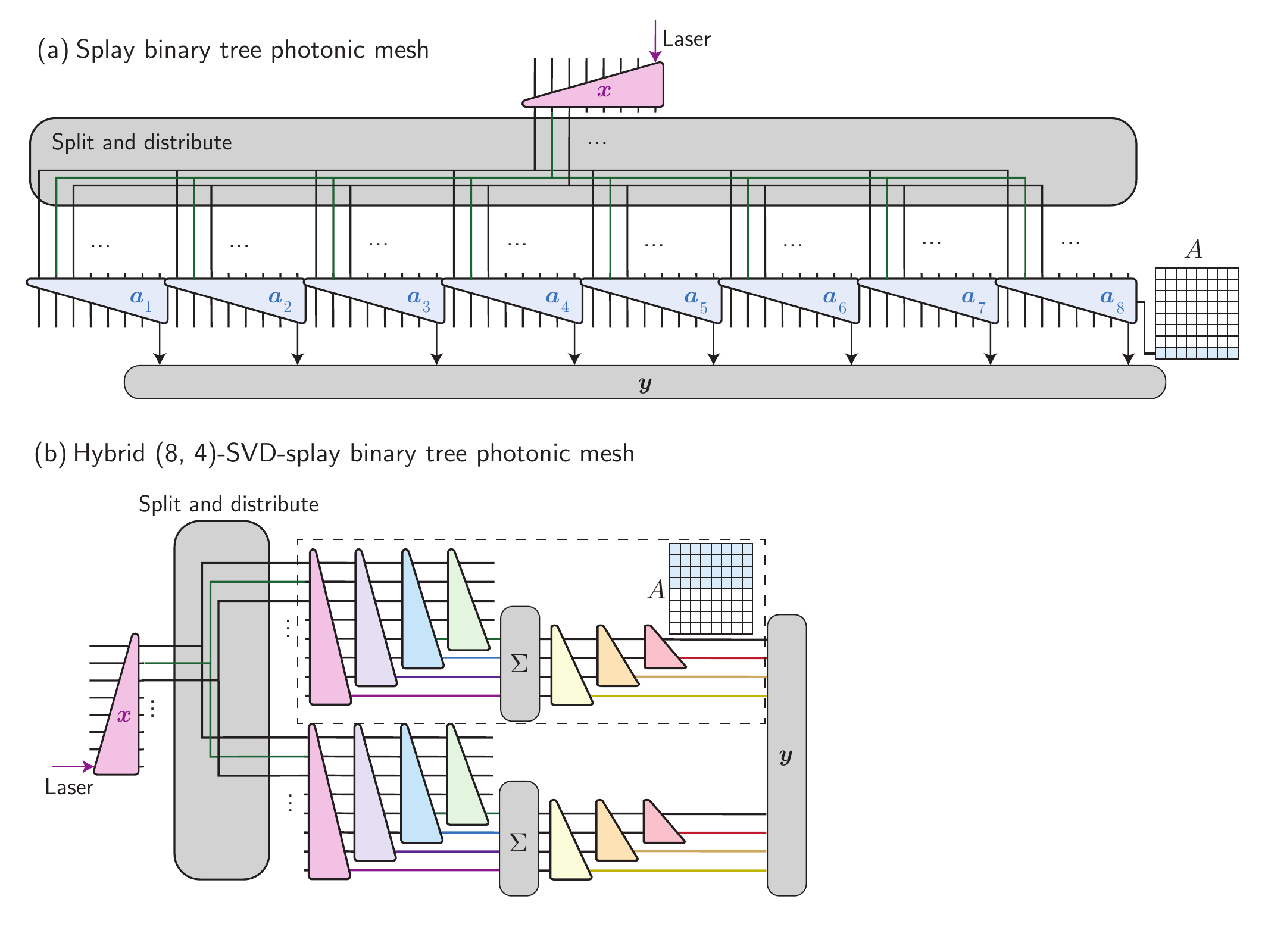}
    \caption{(a) The simplest proposal of a splay architecture for matrix-vector multiplication proceeds by directly routing the $N$-element input vector into $N$ vector units, here shown for $N = 8$. A split-and-distribute permutation layer is required for this effective ``copying'' procedure. (b) To reduce the loss by minimizing the number of splits of the input laser light, we can combine the SVD architecture at the expense of depth of the hybrid cascade and splay architectures, here shown for $M = 4, N = 8$ requiring just one additional copy of the original input (in general $N / M$).}
    \label{fig:splay}
\end{figure*}

We briefly propose a new universal matrix unit shown in Fig. \ref{fig:matrixunit}(a) based on the cosine-sine decomposition (CSD) \cite{Basani2022AInterferometers}, an architecture typically reserved for universal quantum computation \cite{Mottonen2004QuantumGates}. Like our other binary tree-inspired networks, the CSD architecture also provides nonlocal connections. The key detail is to realize that any unitary matrix (and, in fact, any matrix more generally) can be decomposed in the following form:

\begin{equation} \label{eqn:cosinesine}
    \begin{aligned}
        A_N(\bm{\theta}_N,\bm{\phi}_N) &= \begin{bmatrix}
        M_{11} & M_{12} \\
        M_{21} & M_{22}
        \end{bmatrix} = \begin{bmatrix} 
        L_0 S(\bm{\theta}) R_0 & L_0 C(\bm{\theta}) R_1 \\
        L_1 C(\bm{\theta}) R_0 & -L_1 S(\bm{\theta}) R_1
        \end{bmatrix} \\
        &= \begin{bmatrix} 
L_0 & O \\
O & L_1 
\end{bmatrix} \begin{bmatrix}
S(\bm{\theta}_M) & C(\bm{\theta}_M) \\
C(\bm{\theta}_M) & -S(\bm{\theta}_M)
\end{bmatrix} \begin{bmatrix} 
R_0 & O \\
O & R_1 
\end{bmatrix}
        \\ &= \begin{bmatrix}
                A^{[11]}_{N / 2} & 0 \\
                0 & A^{[21]}_{N / 2}
        \end{bmatrix} \Sigma_N(\mathbf{s}, \bm{\phi})\begin{bmatrix}
                A^{[12]}_{N / 2} & 0 \\
                0 & A^{[22]}_{N / 2}
        \end{bmatrix} \\
        \mathbf{s}_N &= [\mathbf{s}, \mathbf{s}^{[11]}, \mathbf{s}^{[12]}, \mathbf{s}^{[21]}, \mathbf{s}^{[22]}] \\
        \bm{\phi}_N &= [\bm{\phi}, \bm{\phi}^{[11]}, \bm{\phi}^{[12]}, \bm{\phi}^{[21]}, \bm{\phi}^{[22]}] \\
        \Sigma_N(\mathbf{s}, \bm{\phi}) &:= \prod_{n = 1}^{N / 2} T_{n, N / 2 + n}^{[N]}(s_n, \phi_n)
    \end{aligned}
\end{equation}

As suggested by Eq. \ref{eqn:cosinesine}, one method to perform CS decomposition is to partition $U$ into four non-unitary submatrices of size $\frac{N}{2} \times \frac{N}{2}$: $A_{11}, A_{21}, A_{12}, A_{22}$. First, one performs SVD on $M_{11}$ to generate $L_0 S(\bm{\theta}_N) R_0$ where $S(\bm{\theta}_N)$ is nonnegative. The only remaining submatrices to find are then $L_1, R_1$, which can be found by running QR decomposition on $A_{12}^\dagger L_0$ and $M_{21} R_0^\dagger$, which give $R_1^\dagger C(\bm{\theta}_N)$ and $L_1 C(\bm{\theta}_N)$ respectively. Again, we ensure in both cases that $C(\bm{\theta}_N)$ is nonnegative. Ultimately this nonnegativity assumption gives us our $\bm{\theta} \in \left[0, \pi\right]^{M}$ constraint.

Note that as referenced in the main text, the SVD architecture is a specific case of the CS decomposition. From the perspective of this CS decomposition, the SVD technically embeds an $N \times N$ arbitrary matrix in a $2N \times 2N$ unitary space. In the SVD architecture, we stop the recursion of Eq. \ref{eqn:cosinesine} after a single iteration. To see the equivalence with a specific case of the CS photonic mesh, note with the SVD architecture we are effectively embedding a circuit in a $2N \times 2N$ space, so we consider $A_{2N}$. Accordingly, in Eq. \ref{eqn:cosinesine}, we maintain $A^{[11]}_{N}$ and $A^{[12]}_{N}$ but prune $A^{[21]}_{N / 2}$ and $A^{[22]}_{N / 2}$ by setting them to an identity matrix $I$, which gives us a resulting matrix $A_N = A^{[11]}_{N} S(\bm{\theta}) A^{[12]}_{N} := U \Sigma V^\dagger$.

\subsection{Splay architecture}

The increased robustness of ``wide'' balanced over ``deep'' unbalanced architectures in our paper motivate a new ``splay'' architecture based on balanced tree meshes capable of low-loss and highly error-tolerant arbitrary matrix multiplication shown in Fig. \ref{fig:splay}. Instead of a cascade mesh, we might also have a $1 \times N^2$ balanced tree network which can \textit{also} achieve $O(N^2)$ OPS, where the depth and error both scale with $\log N^2 = 2 \log N$, a dramatic improvement over the $N$ depth and error scaling for triangular and rectangular architectures. This computation can be partitioned across $N$ $1 \times N$ trees to form a matrix-vector product of the same depth. A single $1 \times N$ tree is responsible for the input generation and that input is copied and fed into $N$ partitioned $1 \times N$ analyzer trees which compute vector-vector products in parallel.

There are three problems with the splay matrix-vector multiply approach: (1) large footprint, (2) output of $1 / N$ the power of a lossless cascade approach and (3) large passive split-and-permute to replicate the input vector and feed appropriately across the $N$ trees. The $1 / N$ factor comes from the fact that two random settings of the input and analyzer vector setting should result in $1 / N$ of the light to leave the root waveguide. As discussed in the Appendix, such a permutation network can leverage escalator technologies to avoid losses due to waveguide crossings.

The main consideration here is the tradeoff between the component loss $L_{\mathrm{comp}}$ affecting the circuit depth loss $L_{\mathrm{depth}} = D L_{\mathrm{comp}}$ (where $D$ is the depth generally of order $N$ or $M \log N$ for unbalanced and balanced respectively) and the dropout loss affected by the $1 / N$ factor (on average) $L_{\mathrm{drop}} = N / M$. The loss problem is actually not as bad once we consider the loss of individual components in the photonic network. For instance, a $64 \times 64$ matrix multiply can be performed in $D = 18$ photonic layers (including the splitter network and permutations) which actually saves $46$ device layers, reducing both loss and error significantly. The loss reduction happens in devices that incur a loss of $0.5$ to $1$ dB per node possibly 1 to 2 orders of magnitude and balancing the $L_{\mathrm{drop}}$ average loss factor. For this loss-limited case, it may actually make sense to settle for a compromise between a cascade (deep) and splay (wide) architecture as in Fig. \ref{fig:splay}(b). Both the loss and error tolerance are related by the rank $M$. If the error tolerance is allowed to be reduced by a factor of $M = 10$, for instance, the loss due to drop ports can also be reduced by an order of magnitude using $M = 10$ binary tree cascade.

\section{Tree coordinates introduction} \label{sec:intro}
In this mathematically-oriented supplement, we present a Bayesian framework for discrete linear transformations and number systems based on the Dirichlet distribution. In the main paper, as a concrete practical application, we showed that arbitrary linear transformations can be systematically programmed on linear optical devices and error scaling relations may be derived based on a \textit{tree coordinate system}, as first suggested in Ref. \cite{Miller2013Self-aligningCoupler}. 

However, our framework can be framed in an even more general context. Using our coordinate system, we develop a statistical unitary matrix model based on gamma and Dirichlet distributions that generalizes the seminal result of Hurwitz \cite{Hurwitz1897UberIntegration} in 1897 and its translation to physical linear optical computing platforms by Reck and Miller \cite{Miller2013Self-configuringInvited, Reck1994ExperimentalOperator} over a century later. We apply our framework to fault-tolerant, high-bandwidth linear optical architectures, which have been shown to have applications in linear optical quantum computing \cite{Carolan2015UniversalOptics}, energy efficient machine learning \cite{Shen2017DeepCircuits} and communications \cite{Annoni2017UnscramblingModes}.

Our decision tree graphical framework is simultaneously of use to multidimensional coordinate systems, parametrizations of unitary matrices, and design of fault-tolerant linear optical devices. The outline for this supplement proceeds as follows:
\begin{enumerate}
    \item In Section \ref{sec:treecoordinatestats}, we explain why decision tree models obey Gamma and Dirichlet distribution statistics.
    \item In Section \ref{sec:rotationstats}, we show how any $M$-dimensional rotation operator can be modelled by a decision tree with $M$ leaves.
    \item In Section \ref{sec:unitarystats}, we show that any $N \times N$ unitary matrix $U$ can be modelled by decision trees containing $1, 2, \ldots N$ leaves. If the decision tree model parameters for each graph obey the appropriate Dirichlet statistics, we arrive at the Haar measure of the unitary group.
\end{enumerate}

\section{Tree coordinate statistics} \label{sec:treecoordinatestats}

The tree coordinate system uses decision tree statistics to model multidimensional vectors. Decision tree statistics are generally used to model resource allocation strategies \cite{Dennis1996AProblems}.

For this section, we will consider this resource to be a string of length $Y_M$ (taking up the interval $[0, Y_M]$ on the number line). The ultimate goal is to find optimal strategies to split the string into $M$ pieces (which ultimately informs linear optical network designs discuss in the main text). The length of the cut strings represents how the total resource is allocated ($\boldsymbol{y} \in \mathbb{R}_{\geq 0}^M$), while the length of overall string represents the total resource available ($Y_M = \boldsymbol{1} \cdot \boldsymbol{y} \in \mathbb{R}_{\geq 0}$). 

\subsection{Tree coordinates}
At its core level, the tree coordinate system is generally a model for random complex vectors. Given a complex vector $\boldsymbol{v}_N \in \mathbb{C}^N$, the formula for each element is $v_n = a_n + i b_n$, where $a_n$ is the real part and $b_n$ is the imaginary part. In discrete linear optics, each vector element can be represented by measurable quantities: the power $y_n = |v_n|^2$ (denoted as the vector $\boldsymbol{y} \in \mathbb{R}_{\geq 0}^N$) and the relative phase $\varphi_n = \angle(v_n)$ (denoted as the vector $\boldsymbol{\varphi} \in [0, 2\pi)^N$) of a propagating mode in the $n$th single-mode waveguide. This phasor representation leads to a more intuitive representation of the statistics of coordinates in $N$-dimensional Euclidean space.

Reconfigurable beamsplitter trees \cite{Miller2013Self-aligningCoupler, Harris2017QuantumProcessor} can be fabricated on a photonic platform to guide light arbitrarily from a single waveguide to $N$ waveguides. In this paper, we represent the single input to the root node with power $Y_N := \boldsymbol{1} \cdot \boldsymbol{y} = \|\boldsymbol{v}_N\|^2$ (assuming a lossless optical system) using the $N$th standard basis vector, i.e. $\boldsymbol{v}_{\mathrm{in}} = \sqrt{Y_N} \boldsymbol{e}_{N}$. The operator implemented by the device (represented by tree graph $\mathcal{G}_N$) is capable of generating any $\boldsymbol{v}_N$ using the arbitrary unitary operator $R_{\mathcal{G}_N}(\boldsymbol{u}_N)$, where $\boldsymbol{u}_N = \boldsymbol{v}_N / \sqrt{Y_N}$. The operator $R_{\mathcal{G}_N}(\boldsymbol{u}_N)$ can be thought of as a ``complex rotation" computed entirely in the analog domain that is independent of the total power (or squared vector norm) $Y_N$. The device implements $ \boldsymbol{v}_N = R^\dagger_{\mathcal{G}_N}(\boldsymbol{u}_N)\sqrt{Y_N} \boldsymbol{e}_N$ as light propagates from the input port to the output ports. 

\subsection{Gamma and Dirichlet distributions}

Assume the string length $Y_K \sim \mathrm{Gam}(A)$ is a gamma-distributed random variable. If we make $K - 1$ simultaneous cuts in the string, we obtain a set of string lengths $\boldsymbol{y} \in \mathbb{R}_{\geq 0}^K$, a vector of $K$ positive real numbers. By virtue of the additive property of gamma-distributed variables, we require $\boldsymbol{y} \sim \mathrm{Gam}(\boldsymbol{\alpha})$ (iid $y_k \sim \mathrm{Gam}(\alpha_k)$), where $\mathbf{1} \cdot \boldsymbol{\alpha} = A$. Intuitively, the $\alpha_k$ define how long each string piece is on average, so $\boldsymbol{\alpha}$ are constants that represent the cut strategy. This concept can also be thought of in reverse; given cut strings of lengths $\boldsymbol{y} \sim \mathrm{Gam}(\boldsymbol{\alpha})$, we can glue the strings end-to-end to achieve a master string of length $Y_K \sim \mathrm{Gam}(A)$. In optical systems, we can consider this master string to be analogous to total power, and the cuts to be the allocation of that power to different optical paths.

Now define $\boldsymbol{x} = \boldsymbol{y} / Y_K$, so that each element $x_k$ represents the fractional string length of the $k$th piece. Then $\boldsymbol{x}$ follows a Dirichlet distribution parametrized by $\boldsymbol{\alpha}$, i.e. $\boldsymbol{x} \sim \mathrm{Dir}(\boldsymbol{\alpha})$. The proof of this relationship is a standard result in statistics provided explicitly in Appendix \ref{sec:gammatodirichlet} for convenience.

Given $\boldsymbol{\alpha}$, we define the probability distribution functions for $\boldsymbol{x}, \boldsymbol{y}$ to be
\begin{equation} \label{eqn:gammadirichletdist}
    \begin{aligned}
    \mathcal{P}_\Gamma(\boldsymbol{y}; \boldsymbol{\alpha}) &:= \prod_{k = 1}^K \frac{y_k^{\alpha_k - 1}e^{-y_k}}{\Gamma(\alpha_k)}\\
    \mathcal{P}_\mathrm{D}(\boldsymbol{x}; \boldsymbol{\alpha}) &:= \prod_{k = 1}^K \frac{x_k^{\alpha_k - 1}}{\mathrm{D}(\boldsymbol{\alpha})}\\
    \mathrm{D}(\boldsymbol{\alpha}) &:= \frac{ \Gamma(A)}{\prod_{k = 1}^K \Gamma(\alpha_k)},
    \end{aligned}
\end{equation}
where for integer values of $\alpha$, $\Gamma(\alpha) = (\alpha - 1)!$, and in general, $\Gamma(\alpha) = \int_0^\infty y^{\alpha - 1}e^{-y} dy$, the normalization constant for the Gamma distribution.

\begin{figure}
    \centering
    \includegraphics[width=0.48\textwidth]{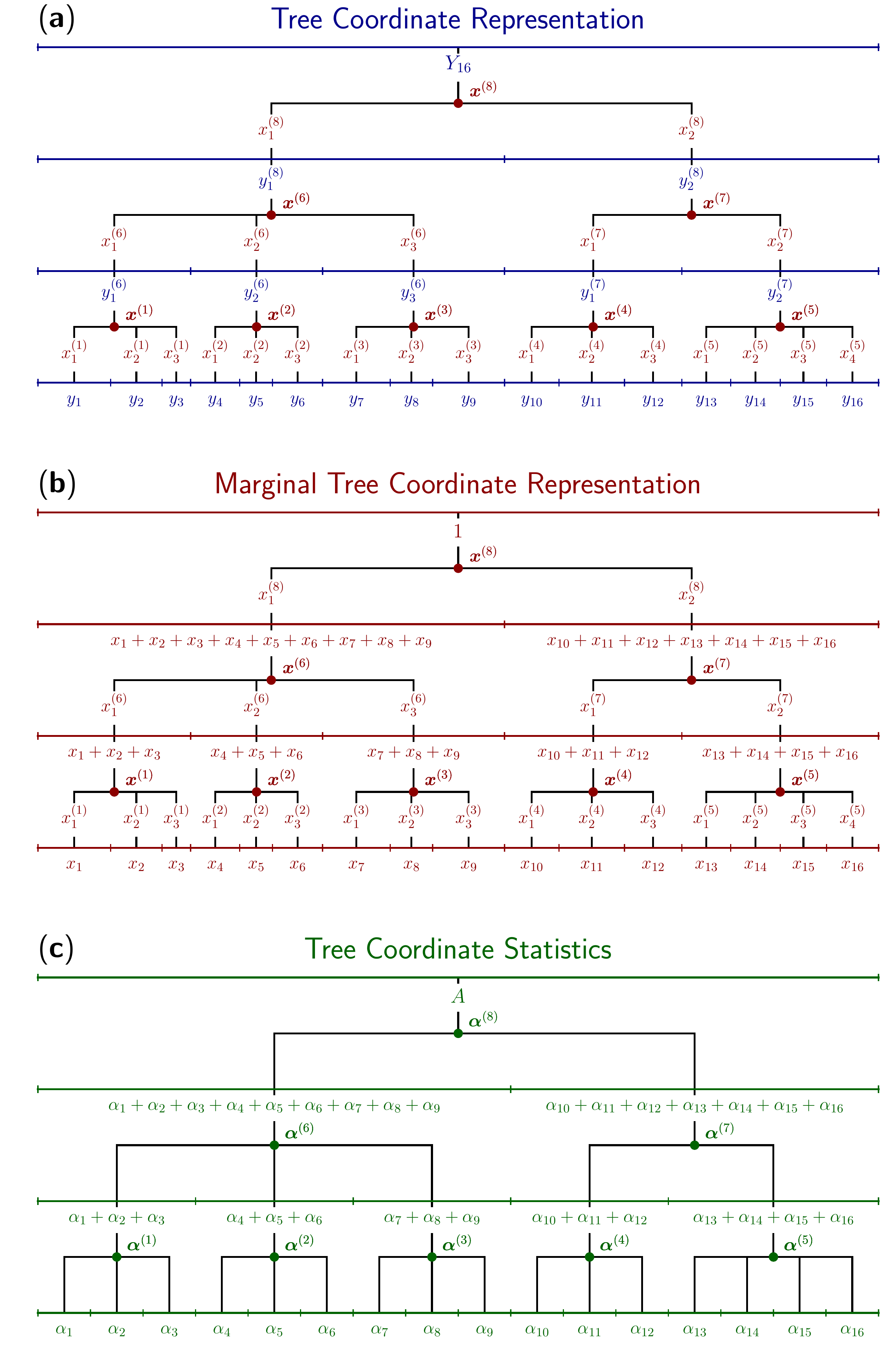}
    \caption{We plot various representations of tree coordinate system in terms of string cutting for a given $\mathcal{G}$ and $M = 16$ (ignoring $\boldsymbol{\varphi}$). (a) The tree coordinates $\boldsymbol{x}^{(j)}, Y_{16}$ for Gamma basis $\boldsymbol{y} \sim \mathrm{Gam}(\boldsymbol{\alpha})$ for $\boldsymbol{\alpha} = 10 \cdot \boldsymbol{1}$. (b) The tree coordinates $\boldsymbol{x}^{(j)}$ after marginalizing out $Y_{16}$ for Dirichlet basis $\boldsymbol{x} = \boldsymbol{y} / Y_{16} \sim \mathrm{Dir}(\boldsymbol{\alpha})$. (c) The propagation of the Dirichlet statistical parameters $\boldsymbol{\alpha}$ into node Dirichlet parameters $\boldsymbol{\alpha}^{(j)}$. Using these statistics, we find that, in panels (a) and (b), $\boldsymbol{y}^{(j)} \sim \mathrm{Gam}(\boldsymbol{\alpha}^{(j)})$ and $\boldsymbol{x}^{(j)} \sim \mathrm{Dir}(\boldsymbol{\alpha}^{(j)})$.}
    \label{fig:treecoordinates}
\end{figure}

\subsection{Dirichlet tree representation}

We refer to any group of $K$ cuts as a \textit{cut event}. In our string-cutting problem, we require $M - 1$ total cuts to get our $M$ pieces. As defined previously, $\boldsymbol{y} \in \mathbb{R}_{\geq 0}^M$ represents the length of the pieces and $\boldsymbol{x} = \boldsymbol{y} / Y_M$ represents the fractional length of the pieces. In general, we can have $J \leq M - 1$ cut events, where if the $j$th cut event involves $K_j$ simultaneous cuts, then $\sum_{j = 1}^J (K_j - 1) = M - 1$ must hold.

This rule matches the convenient property of a tree graph consisting of $J$ decision nodes, where each node is a cut event that maps a single input edge to $K_j$ output edges. In the extreme cases, we can cut the string once at a time ($M - 1$ cut events) or we can make all $M - 1$ cuts simultaneously ($1$ cut event). As an example, we show a tree graph of $J = 8$ nodes and $\{K_j\} = (3, 3, 3, 3, 4, 3, 2, 2)$ in Figure \ref{fig:treecoordinates}.


Any tree graph $\mathcal{G}$ can be represented as the set of connection sets $\{\mathcal{T}_1, \mathcal{T}_2, \ldots \mathcal{T}_J\}$. Each connection set $\mathcal{T}_j$ consists of $K_j$ indices $r \leq M + J$ linking node $j$ to other nodes ($r \leq J$) or leaves ($r > J$). For $\mathcal{G}$, the leaf subset $\mathcal{S}_r \subset \{1, 2, \ldots M\}$ can be calculated for each node or leaf:
\begin{equation} \label{eqn:leafspanset}
    \begin{aligned}
        \mathcal{S}_r &= \begin{dcases}
        \bigcup_{r' \in \mathcal{T}_r} \mathcal{S}_{r'} & r \leq J\\
        \{r - J\}& r > J
        \end{dcases}
    \end{aligned}
\end{equation}

We define $\boldsymbol{y}^{(j)} \in \mathbb{R}_{\geq 0}^{K_j}$ as the string lengths exiting node $j$. We invoke the property that the sum of gamma-distributed variables is also gamma distributed to find the length of the $k$th string cut by node $j$:
\begin{equation}\label{eqn:noderandomvars}
\begin{aligned}
    \alpha_k^{(j)} &:= \sum_{m \in \mathcal{S}_{\mathcal{T}_j[k]}} \alpha_m\\
    y_k^{(j)} &= \sum_{m \in \mathcal{S}_{\mathcal{T}_j[k]}} y_m \sim \mathrm{Gam}(\alpha_k^{(j)})
\end{aligned}
\end{equation}
where $\mathcal{T}_j[k]$ represents $r$ corresponding to edge $k$ in the connection set $\mathcal{T}_j$.

Given Equation \ref{eqn:noderandomvars}, we find $\boldsymbol{x}^{(j)}$, the fractional lengths of the strings cut by node $j$, satisfies the Dirichlet distribution needed for each node $j$, i.e.
\begin{equation} \label{eqn:unitdirichlet}
    \begin{aligned}
    \boldsymbol{y}^{(j)} \sim \mathrm{Gam}(\boldsymbol{\alpha}^{(j)}) &\to
    \boldsymbol{x}^{(j)} \sim \mathrm{Dir}(\boldsymbol{\alpha}^{(j)}).
    \end{aligned}
\end{equation}

Regardless of the graph structure, we have shown how the node statistical parameters $\boldsymbol{\alpha}^{(j)}$ can be defined such that any final cut strategy $\boldsymbol{\alpha}$ can be achieved.

Note that $\boldsymbol{x}$ can be written in terms of the node parameters $\boldsymbol{x}^{(j)}$ in $\mathcal{G}$. For the $m$th fractional cut length $x_m$, we define a unique \textit{path} $\mathcal{E}_{m, \mathcal{G}}$ as the set of node-edge pairs $(j, k)$ in the path from the root node of $\mathcal{G}$:
\begin{equation} \label{eqn:dirichlettotree}
    x_m = \prod\limits_{(j, k) \in \mathcal{E}_{m, \mathcal{G}}} x_k^{(j)},
\end{equation}
which we then substitute to find each complex vector element $v_m = \sqrt{Y_M x_m}e^{-i\varphi_m}$. For example, in Figure \ref{fig:treecoordinates}, we have that $x_6 = x_1^{(8)}x_2^{(6)}x_3^{(2)}$ and $|v_6| = \sqrt{y_6} = \sqrt{Y_{16}x_6} = \sqrt{Y_{16} x_1^{(8)}x_2^{(6)}x_3^{(2)}}$, which can be seen by following the path from the root node to leaf $m$. Note that the tree graphical structure ensures that there is exactly one path $\mathcal{E}_{m, \mathcal{G}}$ to leaf $m$.

\subsection{Applications}

The key analogy between string cutting and the resource allocation applications we have mentioned is that each cut to the string represents a component or branch operation in a physical resource allocation system (e.g., MZI in an interferometer tree). The lengths of the string represent quantities (e.g. light intensity, electrical current, volumetric flow rate, probability current) being allocated throughout the network.

Decision trees gives a straightforward way to think about what $\boldsymbol{x}^{(j)}, \boldsymbol{y}$ mean in terms of real applications: $\boldsymbol{y}$ represents how the resource ends up being allocated whereas $\boldsymbol{x}^{(j)}$ represents how each node $j$ has to split up the resource entering the node to achieve the final desired $\boldsymbol{y}$. We now provide a mathematical application of this idea to multidimensional rotations.

\section{Statistics of a rotation} \label{sec:rotationstats}

We have presented the statistics of string cutting and its general use in tree coordinate systems. Multidimensional rotations can also be statistically described as a resource allocation problem.

\subsection{Gamma and Dirichlet basis}

For the standard normal complex vector $\boldsymbol{v}_K \in \mathbb{C}^K$, there are two basis representations that we consider:
\begin{enumerate}
    \item the \textit{Gamma basis} $(\boldsymbol{y}, \boldsymbol{\varphi})$ where $y_k = |v_k|^2$ and $\varphi_k = \angle(v_k)$ for all $k$, i.e. $\boldsymbol{v}_K = \sqrt{\boldsymbol{y}} e^{i \boldsymbol{\varphi}}$.
    \item the \textit{Dirichlet basis} $(\boldsymbol{x}, \boldsymbol{\varphi}, Y_K)$ where $\boldsymbol{x} = \boldsymbol{y} / Y_K$ are fractional powers and $Y_K = \|v_k\|^2 = \sum_{k = 1}^K y_k$ is the radius or normalization factor.
\end{enumerate}

The Gamma and Dirichlet measures are defined as:
\begin{equation} \label{eqn:gammadirichletbasis}
    \begin{aligned}
    \mathcal{P}(\boldsymbol{v}_K) \mathrm{d}\boldsymbol{v}_K &:= \mathcal{P}_\Gamma(\boldsymbol{y}; \boldsymbol{\alpha}) \mathrm{d}\boldsymbol{y} \frac{\mathrm{d}\boldsymbol{\varphi}}{(2\pi)^K} \\&= \mathcal{P}_\mathrm{D, \Gamma}(\boldsymbol{x}, Y_K; \boldsymbol{\alpha}) \mathrm{d}\boldsymbol{x}\mathrm{d}Y_K \frac{\mathrm{d}\boldsymbol{\varphi}}{(2\pi)^K},
    \end{aligned}
\end{equation}
which follows naturally from the proof in Appendix \ref{sec:gammatodirichlet}, with support $\sum_{k=1}^K x_k = 1$. This definition, ignoring the addition of uniform-random phases $\boldsymbol{\varphi}$, follows decision tree statistics.

Consider the standard complex normal vector $\boldsymbol{v}_K \in \mathbb{C}^K$, where we require $a_k, b_k \sim \mathcal{N}(0, 0.5)$ for all $k \leq K$ where $v_k = a_k + ib_k$. Then it is straightforward to show that the Gamma basis for $\boldsymbol{v}_K$ is parametrized by $\boldsymbol{\alpha} = \mathbf{1}$.
\begin{equation}
    \begin{aligned} \label{eqn:jointgamma1}
    \mathcal{P}_{\mathcal{N}}(\boldsymbol{v}_K) \mathrm{d}\boldsymbol{v}_K &:= \prod_{k=1}^K \frac{e^{-a_k^2}e^{-b_k^2}}{\pi} \mathrm{d}a_k \mathrm{d}b_k \\
    &= \prod_{k=1}^K e^{-y_k} \mathrm{d}y_k \frac{\mathrm{d}\varphi_k}{2\pi}\\
    &= \mathcal{P}_\Gamma(\boldsymbol{y}; \mathbf{1}) \mathrm{d}\boldsymbol{y} \frac{\mathrm{d}\boldsymbol{\varphi}}{(2\pi)^K},
    \end{aligned}
\end{equation}
where we use the fact that the determinant of the Jacobian $\det \mathcal{J}_{(a_k, b_k)}^{(y_k, \varphi_k)} = 1/ 2$.

\subsection{Graphical rotation operator}

Given any $\boldsymbol{v}_M \in \mathbb{C}^M$ and $\boldsymbol{u}_M = \boldsymbol{v}_M / Y_M$, our goal is to find $R_{M, \mathcal{G}}(\boldsymbol{u}_M)$ (where $\mathcal{G}$ is a tree graph with $M$ leaves) such that:
\begin{equation} \label{eqn:rotationoperator}
    R_{M, \mathcal{G}}(\boldsymbol{u}_M)\boldsymbol{v}_M = O_{M, \mathcal{G}}(\boldsymbol{x}) D_M(-\boldsymbol{\varphi}) \boldsymbol{v}_M = \sqrt{Y_M}\boldsymbol{e}_M,
\end{equation}
where $\boldsymbol{e}_M$ is the standard Euclidean basis vector and $D_M$ represents a diagonal unitary of phases. We define the \textit{unit Dirichlet basis} as $(\boldsymbol{x}, \boldsymbol{\varphi})$, where $\boldsymbol{u}_M = \sqrt{\boldsymbol{x}}e^{i\boldsymbol{\varphi}}$. For convenience, we also define the Dirichlet basis rotation operator $R_M(\boldsymbol{u}_M) = O_M(\boldsymbol{x}) D_M(-\boldsymbol{\varphi})$.

As demonstrated in Equation \ref{eqn:rotationoperator}, a general rotation operatorin a unitary operator that can be constructed in two steps:
\begin{enumerate}
    \item \textit{Absolute value operator}: $D_M(-\boldsymbol{\varphi})$ is a diagonal unitary that removes the phases stored in $\boldsymbol{v}_M$, i.e. $D_M(-\boldsymbol{\varphi}) \boldsymbol{v}_M = |\boldsymbol{v}_M| = \sqrt{\boldsymbol{y}}$.
    \item \textit{Dirichlet tree operator}: $O_{M, \mathcal{G}}(\boldsymbol{x})$ is an orthogonal operator modelled by $\mathcal{G}$ that depends on the Dirichlet basis $\boldsymbol{x}$. Each node in the tree graph implements the \textit{Dirichlet node operator} $O_{K_j}(\boldsymbol{x}^{(j)})$, where $\boldsymbol{x}^{(j)} \sim \mathrm{Dir}(\boldsymbol{\alpha}^{(j)})$. The only requirement is $O_{K_j}(\boldsymbol{x}^{(j)}) \sqrt{\boldsymbol{x}^{(j)}} = \boldsymbol{e}_{K_j}$, and that the operator function itself is not factorizable  
\end{enumerate}
In summary, we remove phases so we are left with positive real numbers. We then explicitly construct the Dirichlet tree operator $O_{M, \mathcal{G}}(\boldsymbol{x})$ from Dirichlet node operators $O_{K_j}(\boldsymbol{x}^{(j)})$ using Lemma \ref{lm:mdirichlettree} of Appendix \ref{sec:dirichlettreeoperator}.

Note that for real rotations, we restrict $\boldsymbol{\varphi}$ to take values of only $0$ or $\pi$ (i.e., $e^{i\varphi_m} = \pm 1$). In this case, $\boldsymbol{\varphi}$ are no longer degrees of freedom stored in $R_{M, \mathcal{G}}(\boldsymbol{u}_M)$. This ultimately allows us to parametrize either real rotations or orthogonal matrices discussed further in Appendix \ref{sec:orthogonal}.

\subsection{Graphical coordinate systems}

We have found a tree coordinate system to represent any complex vector. These results are summarized in Table \ref{tbl:summary} and Figure \ref{fig:treecoordinates}.

\begin{table}[h] 
\makegapedcells
\begin{tabular}{p{0.1\textwidth}|p{0.24\textwidth}|p{0.125\textwidth}}
\textbf{Basis} & \textbf{Parameters} & \textbf{Distribution} \\ \hline
Euclidean \newline $\boldsymbol{a}, \boldsymbol{b}$  & $\boldsymbol{v}_M = \boldsymbol{a} + i \boldsymbol{b}$ & $a_m \sim \mathcal{N}\left(0, \frac{1}{2}\right)$ \newline $b_m \sim \mathcal{N}\left(0, \frac{1}{2}\right)$\\\hline
Gamma \newline $(\boldsymbol{y}, \boldsymbol{\varphi})$ & $\boldsymbol{v}_M = \sqrt{\boldsymbol{y}}e^{i\boldsymbol{\varphi}}$ & $y_m \sim \mathrm{Gam}(1)$\newline$\varphi_m \sim \mathcal{U}(0, 2\pi)$\\\hline
Dirichlet \newline $(\boldsymbol{x}, Y_M, \boldsymbol{\varphi})$ & $\boldsymbol{v}_M = \sqrt{Y_M\boldsymbol{x}}e^{i\boldsymbol{\varphi}}$\newline \newline $\boldsymbol{v}_M = \sqrt{Y_M} R^\dagger_{M}(\boldsymbol{x}, \boldsymbol{\varphi})\boldsymbol{e}_M$ & $\boldsymbol{x} \sim \mathrm{Dir}(\boldsymbol{1})$\newline $Y_M \sim \mathrm{Gam}(M)$\newline$\varphi_m \sim \mathcal{U}(0, 2\pi)$ \\\hline
Tree, $\mathcal{G}$ \newline $(\boldsymbol{x}^{(j)}, Y_M, \boldsymbol{\varphi})$ &  $\boldsymbol{v}_M = \sqrt{Y_M} R^\dagger_{M, \mathcal{G}}(\boldsymbol{x}, \boldsymbol{\varphi})\boldsymbol{e}_M$ \newline\newline  $v_m = \sqrt{Y_M\displaystyle\prod\limits_{(j, k) \in \mathcal{E}_{m, \mathcal{G}}} x_k^{(j)}} e^{i\varphi_m}$ & $\boldsymbol{x}^{(j)} \sim \mathrm{Dir}(\boldsymbol{\alpha}^{(j)})$\newline$Y_M \sim \mathrm{Gam}(M)$\newline$\varphi_m \sim \mathcal{U}(0, 2\pi)$
\end{tabular}
\caption{The tree coordinate system and corresponding distributions for an iid complex standard normal vector $\boldsymbol{v}_M$.}
\label{tbl:summary}
\end{table}

We therefore find that the Gamma basis, Dirichlet basis, and all tree bases are all equally valid ways to represent a complex vector. It is possible to reparameterize a Dirichlet basis into any tree basis (and vice versa) using Equation \ref{eqn:dirichlettotree}. Therefore, there are an exponential number of tree coordinate systems that represent a random complex vector, and each tree coordinate system obeys a different set of statistics based on the Dirichlet distribution.

A specific case of the tree coordinate system (binary tree) accounts for all possible Euler angle representations of a multidimensional rotation, which is relevant for canonical linear optical architectures since Euler angles correspond to phase shifts in linear optical devices \cite{Reck1994ExperimentalOperator}.

\section{Statistics of a unitary matrix} \label{sec:unitarystats}

The statistics of a random unitary matrix corresponds closely to the rotation statistics we have just described. This is because any unitary matrix $U_N$ of size $N$ can be constructed by multiplying general rotation operators in Hilbert spaces of size $1, 2 \ldots N$ \cite{Reck1994ExperimentalOperator}.

\subsection{Unitary construction}

Consider the graphical rotation operators $R_{1, \mathcal{G}_1}, R_{2, \mathcal{G}_2}, \ldots R_{N, \mathcal{G}_N}$ in an $N$-dimensional basis (notated as $R_{n, \mathcal{G}_n}^{[N]}$, an $n \times n$ block in the first $n$ rows of an $N \times N$ identity matrix).
\begin{equation} \label{eqn:unitary}
    U_N = \prod_{n = 1}^N R_{n, \mathcal{G}_n}^{[N]}(\widetilde{\boldsymbol{u}}_n)
\end{equation}
where the $\widetilde{\boldsymbol{u}}_n$ are recursively defined in terms of $\boldsymbol{u}_n$ (columns of $U_N$) as:
\begin{equation} \label{eqn:unitaryalg}
    \widetilde{\boldsymbol{u}}_n = \left(\prod_{n' = n + 1}^{N} R_{n', \mathcal{G}_{n'}}^{[N]}(\widetilde{\boldsymbol{u}}_{n'})\right) \cdot \boldsymbol{u}_n,
\end{equation}
where we note that $\widetilde{\boldsymbol{u}}_N = \boldsymbol{u}_N$.

\begin{figure}
    \centering
    \includegraphics[width=0.48\textwidth]{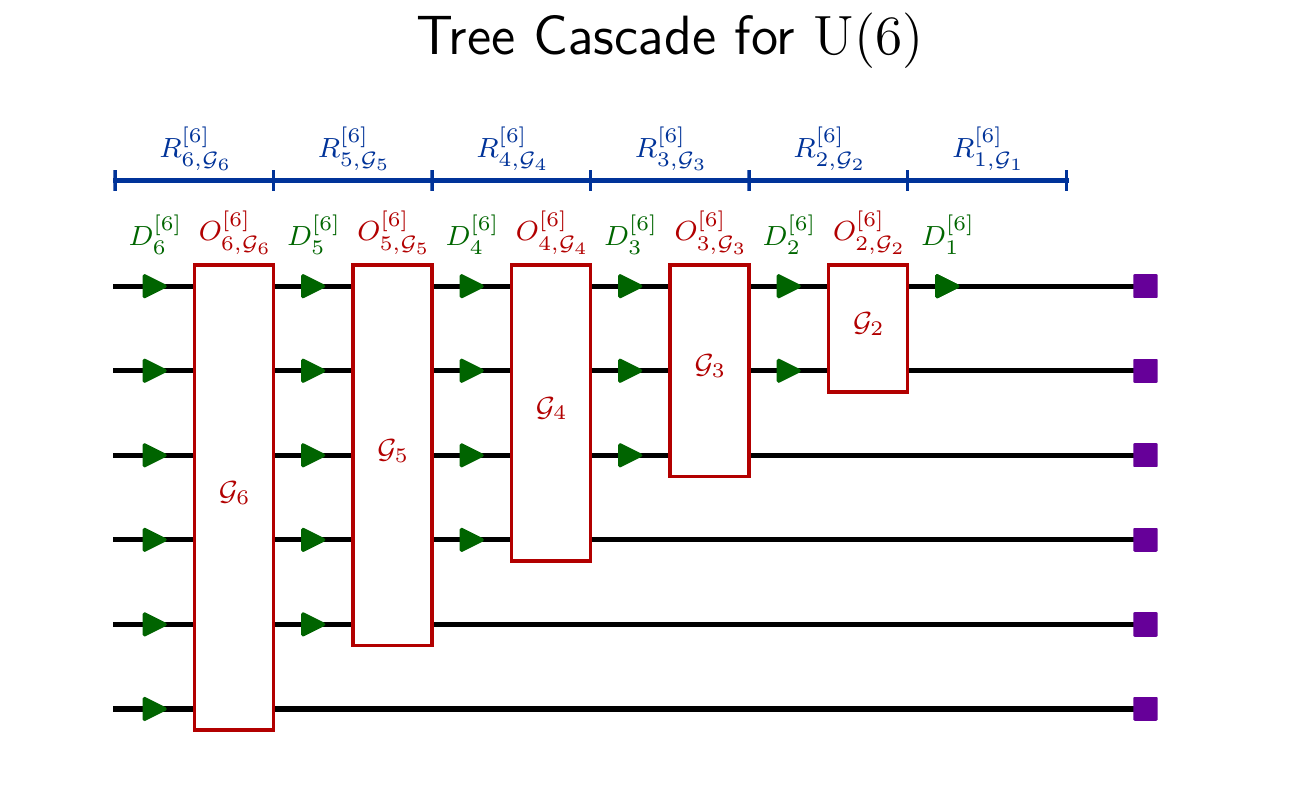}
    \caption{Graphical representation of the rotation operator cascade of Equation \ref{eqn:unitary}. Each rotation operator is represented in terms of two stages: absolute value operator (green) and Dirichlet tree operator (red). Note that $O_{1, \mathcal{G}_1}  = 1$ always.}
    \label{fig:unitary}
\end{figure}

We depict Equations \ref{eqn:unitary} in Figure \ref{fig:unitary} and \ref{eqn:unitaryalg} is more-or-less depicted in the main text. Each graph $\mathcal{G}_n$ is any unit Dirichlet tree basis model as in Fig. \ref{fig:treecoordinates}(b) and parameterizes $R^{[N]}_{n, \mathcal{G}_n}$. The graphs parametrize the overall unitary operator and corresponding statistics. In the case that each $\widetilde{\boldsymbol{u}}_n$ follows a unit Dirichlet basis parametrized by $\boldsymbol{\alpha}_n = \boldsymbol{1}$, we have a random unitary matrix which defines the Haar measure \cite{Russell2017DirectMatrices}.

\subsection{Haar measure of $\mathrm{U}(N)$}
Cascading probabilistic graphs $\mathcal{G}_n$ together forms a unitary operator $U_N \in \mathrm{U}(N)$ as shown in Equation \ref{eqn:unitary}. We can write the rotation measure for $\mathrm{d}\mu(R_{n, \mathcal{G}_n})$ by multiplying the Dirichlet PDFs $P_\mathrm{D}(\boldsymbol{x}_n^{(j)}; \boldsymbol{\alpha}_n^{(j)})$ of all nodes $j \leq J$ in $\mathcal{G}_n$ \cite{Russell2017DirectMatrices}:
\begin{equation}\label{eqn:haarrotation}
    \begin{aligned}
    \mathrm{d}\mu(R_{n, \mathcal{G}_n}) &= \frac{\mathrm{d}\boldsymbol{\varphi}_n}{(2\pi)^{n}} \prod_{j=1}^{J_n} P_\mathrm{D}(\boldsymbol{x}_n^{(j)}; \boldsymbol{\alpha}_n^{(j)}) \mathrm{d}\boldsymbol{x}_n^{(j)},
    \end{aligned}
\end{equation}
with support $x_{n, K_{nj}}^{(j)} = 1 - \sum_{k = 1}^{K_{nj}-1} x^{(j)}_{n, k}$ for each node $j$. Each of the $\boldsymbol{\alpha}_n^{(j)}$ represent the number of leaves in the graph $\mathcal{G}_n$ spanned by the $K_{nj}$ edges exiting node $j$. Intuitively, the $\boldsymbol{\alpha}_n^{(j)}$ are weights representing the total amount of resource that needs to be sent into each edge of the graph to ultimately achieve an approximately equal power distribution, i.e. $\boldsymbol{y}_n \sim \mathrm{Gam}(\boldsymbol{1})$ as required for a random normal vector $\boldsymbol{v}_n \in \mathbb{C}^n$. 

Due to the unitary construction of Equation \ref{eqn:unitaryalg}, we can define the most general parametrization of the unitary Haar measure as a product of the measures in Equation \ref{eqn:haarrotation} as is done in Ref. \cite{Russell2017DirectMatrices}:
\begin{equation}\label{eqn:haarmeasure}
    \begin{aligned}
    \mathrm{d}\mu(U_N) &= \prod_{n=1}^{N} \mathrm{d}\mu(R_{n, \mathcal{G}_n}) \\
    &= \prod_{n=1}^{N} \frac{\mathrm{d}\boldsymbol{\varphi}_n}{(2\pi)^{n}} \prod_{j=1}^{J_n} P_\mathrm{D}(\boldsymbol{x}_n^{(j)}; \boldsymbol{\alpha}_n^{(j)})  \mathrm{d}\boldsymbol{x}_n^{(j)}
    \end{aligned}
\end{equation}

\subsection{Spherical coordinate parametrizations}
While we have determined the general parametrization of a Haar measure in terms of Dirichlet tree probabilistic graphs, it is also important to consider the implications of our model for  spherical coordinate systems that parametrize physically realizable linear optical devices.

We first note that general spherical coordinates are a specific case of Dirichlet tree where $J_n = n - 1$ and $K_{nj} = 2$ for all $n, j$. Since $K_{nj} = 2$, we can represent each node $j \leq J_n$ by a single transmissivity parameter $t_{nj}$ (where $\boldsymbol{x}_n^{(j)} = (t_{nj}, r_{nj})$ and $r_{nj} = 1 - t_{nj}$ is the reflectivity). 

Define a vector of angles $\boldsymbol{\theta}_n \in [0, \pi]^{n - 1}$ where $\boldsymbol{\theta}_n = 2 \arccos{\sqrt{\boldsymbol{t}_n}}$. Based on this definition, we find that each node can be represented by:
\begin{equation} \label{eqn:mzi}
    \begin{aligned}
    O_2(\theta) &= \begin{bmatrix}\sin \frac{\theta}{2} & \cos \frac{\theta}{2} \\ \cos \frac{\theta}{2} & -\sin \frac{\theta}{2} \\
    \end{bmatrix}\\
    O_2(t) &= \begin{bmatrix}\sqrt{1 - t} & \sqrt{t} \\ \sqrt{t} & -\sqrt{1 - t} \\
    \end{bmatrix},
    \end{aligned}
\end{equation}
where $\theta \in [0, \pi]$, $t = \cos^2(\theta / 2) \in [0, 1]$.

This leads to a Haar measure of the unitary group modelled by the binary tree sequence $\{\mathcal{G}_1, \mathcal{G}_2, \ldots \mathcal{G}_N\}$. For ease of notation, we let $\boldsymbol{\alpha}^{(j)}_n = (\alpha_{nj}, \beta_{nj})$ parametrize the beta (Dirichlet for $K_j = 2$) distributions for each node.
\begin{equation} \label{eqn:binaryhaarmeasure}
    \begin{aligned}
    \mathrm{d}\mu(U_N) &= \prod_{n=1}^{N} \mathrm{d}\mu(R_{n, \mathcal{G}_n}) \\
    &= \prod_{n=1}^{N} \frac{\mathrm{d}\boldsymbol{\varphi}_n}{(2\pi)^n} \prod_{j=1}^{n - 1} \mathcal{P}_{\mathrm{B}}(t_{nj}; \alpha_{nj}, \beta_{nj}) \mathrm{d}t_{nj}\\
    &= \prod_{n=1}^{N} \frac{\mathrm{d}\boldsymbol{\varphi}_n}{(2\pi)^n} \prod_{j=1}^{n - 1} \mathcal{P}_{\mathrm{B}, \theta}(\theta_{nj}; \alpha_{nj}, \beta_{nj}) \mathrm{d}\theta_{nj}\\
    &:= \prod_{n=1}^{N} \frac{\mathrm{d}\boldsymbol{\varphi}_n}{(2\pi)^n} \prod_{j=1}^{n - 1} \mathrm{d}\xi_{nj},
    \end{aligned}
\end{equation}
where
\begin{align*}
    \mathcal{P}_\mathrm{B}(t; \alpha, \beta) &= \frac{t^{\alpha - 1} (1 - t)^{\beta - 1}}{\mathrm{B}(\alpha, \beta)}\\
    \mathcal{P}_{\mathrm{B}, \theta}(\theta; \alpha, \beta) &= \frac{\left(\cos \frac{\theta}{2}\right)^{2\alpha - 1} \left(\sin \frac{\theta}{2}\right)^{2\beta - 1}}{\pi \mathrm{B}(\alpha, \beta)} \\
    \xi_{nj} &= I_{t_{nj}}(\alpha_{nj}, \beta_{nj}).
\end{align*}

In Equation \ref{eqn:binaryhaarmeasure}, we introduce the Haar phase $\xi_{nj} \in [0, 1]$, which when uniformly distributed yields the Haar measure. The Haar phase is related to the transmissivity $t_{nj}$ by the incomplete regularized beta function $I_{t_{nj}}$, the CDF of the beta distribution. Note that for canonical spherical coordinates \cite{Russell2017DirectMatrices}, which has the most unbalanced beta-distributed transmissivities (all $\alpha_{nj} = j, \beta_{nj} = 1$), the Haar phase has the simple expression
\begin{equation}\label{eqn:unbalancedhaarphase}
    \begin{aligned}
    \xi_{nj} &= I_{t_{nj}}(j, 1) = t_{nj}^{j},
    \end{aligned}
\end{equation}
which is the special case of our framework for the Haar phase of triangular and rectangular architectures proven in Ref. \cite{Pai2019MatrixDevices}.


\section{Tree coordinates appendix}

\subsection{Gamma basis to Dirichlet basis proof} \label{sec:gammatodirichlet}

\begin{lemma}
For a Gamma basis $\boldsymbol{y}$ parametrized by $\boldsymbol{\alpha}$, the corresponding Dirichlet basis has Dirichlet parameters $\boldsymbol{x} := (x_1, x_2, \ldots x_{K - 1}) \sim \mathrm{Dir}(\boldsymbol{\alpha})$ and total magnitude $Y_K \sim \mathrm{Gam}(A)$, where $A = \sum_{k=1}^{K}\alpha_k$.
\end{lemma}
\begin{proof}
The joint distribution for $\boldsymbol{y}$ (by definition) behaves as follows:
\begin{equation} \label{eqn:jointgamma}
    \begin{aligned}
    \mathcal{P}_\Gamma(\boldsymbol{y}; \boldsymbol{\alpha}) &= \prod_{k = 1}^K \frac{e^{-y_k}y_k^{\alpha_k - 1}}{\Gamma(\alpha_k)}
    \end{aligned}
\end{equation}

Our Jacobian determinant has the form
\begin{equation}
\begin{aligned}
    \mathrm{det} \mathcal{J}^{\boldsymbol{y}}_{(\boldsymbol{x}, Y_K)} &:= \left|\frac{\partial \boldsymbol{y}}{\partial(\boldsymbol{x}, Y_K)}\right|\\&= \begin{vmatrix} Y_K&  0 &  \ldots &  0 &  x_1 \\
0& Y_K& \ldots& 0& x_2 \\ \vdots& \vdots& \ddots & \vdots &\vdots \\0 & 0&\ldots& Y_K &x_{K - 1} \\ -Y_K & -Y_K &\ldots &-Y_K& x_K \end{vmatrix} \\
&= \begin{vmatrix} Y_K&  0 &  \ldots &  0 &  x_1 \\
0& Y_K& \ldots& 0& x_2 \\ \vdots& \vdots& \ddots & \vdots &\vdots \\0 & 0&\ldots& Y_K &x_{K - 1} \\ 0 & 0 &\ldots &0& 1 \end{vmatrix}\\
&= \det \widetilde{\mathcal{J}} = \prod_k \widetilde{\mathcal{J}}_{kk} = Y_K^{K - 1}
\end{aligned}
\end{equation}
where in the second step we add all of the rows to the final row of the determinant and in the final step, we use the fact that the determinant of an upper triangular matrix is always the same as the product of the diagonal elements of that matrix.

Now, we can write the joint distribution for the new basis by applying the change-of-basis theorem (assuming the constraint $x_K = 1 - \sum_{k = 1}^{K - 1} x_k$):
\begin{equation}
    \begin{aligned}
    \mathcal{P}_{\mathrm{D}, \Gamma}(\boldsymbol{x}, Y_K; \boldsymbol{\alpha}) &= \det \mathcal{J}^{\mathrm{D}}_{\Gamma} e^{-Y_K} \prod_{k = 1}^K \frac{(Y_K x_k)^{\alpha_k - 1}}{\Gamma(\alpha_k)} \\
    &= Y_K^{K - 1} e^{-Y_K} \prod_{k = 1}^K \frac{(Y_K x_k)^{\alpha_k - 1}}{\Gamma(\alpha_k)} \\
    &= Y_K^{A - 1} e^{-Y_K}\prod_{k = 1}^K \frac{ x_k^{\alpha_k - 1}}{\Gamma(\alpha_k)}\\
    &= \Gamma(A) \mathcal{P}_{\Gamma}(Y_K; A) \prod_{k = 1}^K \frac{ x_k^{\alpha_k - 1}}{\Gamma(\alpha_k)}\\
    \mathcal{P}_{\mathrm{D}}(\boldsymbol{x}; \boldsymbol{\alpha}) &= \int \mathcal{P}_{\mathrm{D}, \Gamma}(\boldsymbol{x}, Y_K; \boldsymbol{\alpha}) \mathrm{d}Y_K \\
    &= \Gamma(A) \prod_{k = 1}^K \frac{ x_k^{\alpha_k - 1}}{\Gamma(\alpha_k)}\\
    &= \frac{\prod_{k = 1}^K x_k^{\alpha_k - 1}}{\mathrm{D}(\boldsymbol{\alpha})}.
    \end{aligned}
\end{equation}

Therefore, we have that $Y_K$ follows a gamma distribution, and we integrate out $Y_K$ to get the marginal PDF $\mathcal{P}_\mathrm{D}(\boldsymbol{x}; \boldsymbol{\alpha})$, which is the Dirichlet distribution for $\boldsymbol{x}$ where $\mathrm{D}(\boldsymbol{\alpha})$ is the normalization constant.

\end{proof}

\subsection{Haar measure of $\mathrm{SO}(N)$} \label{sec:orthogonal}

Our theory extends to parametrizing orthogonal matrices using general Dirichlet trees. This parametrization might be useful for designing compact optical devices designed to parametrize arbitrary orthogonal operators (not necessarily all arbitrary \textit{unitary} operators).

We begin by showing that a real Gaussian-distributed vector has a different Gamma basis than the complex Gaussian-distributed vector.

\begin{lemma}
An iid Gaussian-distributed real vector $\boldsymbol{v}$ has a Gamma basis $\boldsymbol{y}$ with $\boldsymbol{\alpha} = \frac{1}{2}\cdot \boldsymbol{1}$.
\end{lemma}

\begin{proof}
We perform a change of basis:
\begin{equation}
    \begin{aligned} \label{eqn:jointgammareal}
    \mathcal{P}_{\mathcal{N}}(\boldsymbol{v}_K) \mathrm{d}\boldsymbol{v}_K &:\propto \prod_{k=1}^K e^{-v_k^2} \mathrm{d}v_k \\
    &\propto \prod_{k=1}^K \frac{1}{ \sqrt{y_k}} e^{-y_k} \mathrm{d}y_k
    \end{aligned}
\end{equation}
Therefore, we have that the Gamma basis $y_k \sim \mathrm{Gam}\left(\frac{\alpha_k}{2}\right)$, i.e. $\boldsymbol{\alpha} = \frac{1}{2}\cdot \boldsymbol{1}$.
\end{proof}

We find that the orthogonal Haar measure can thus be determined using the same procedure as outlined in the main text for the unitary Haar measure. One key difference is that the phase shifts $\boldsymbol{\varphi} \in \{0, \pi\}^{M}$ (multiply by either $1$ or $-1$) and thus are not continuous degrees of freedom. The other difference is that the input $\boldsymbol{\alpha}$ all have values of $1 / 2$ rather than $1$ as was the case for the complex random Gaussian vector. The Haar measure for any $N \times N$ orthogonal matrix $U_N \in \mathrm{SO}(N)$ in our decision tree framework can be written as:
\begin{equation} \label{eqn:generalorthohaarmeasure}
    \begin{aligned}
    \mathrm{d}\mu(U_N) &= \prod_{n=1}^{N} \mathrm{d}\mu(R_{n, \mathcal{G}_n}^{[N]}) \\
    &= \prod_{n=1}^N \prod_{j=1}^{J_n} \mathcal{P}_{\mathrm{D}}(\boldsymbol{x}^{(j)}_n, \boldsymbol{\alpha}^{(j)}_n) \mathrm{d}\boldsymbol{x}^{(j)}_n
    \end{aligned}
\end{equation}
where we have the set of possible $\boldsymbol{x}_n^{(j)}$ (i.e., the support) must satisfy $x_{n, K_{nj}}^{(j)} = 1 - \sum_{k = 1}^{K_{nj}-1} x^{(j)}_{n,k}$ for all $j, n$.

The special case of binary tree networks is:
\begin{equation} \label{eqn:binaryorthohaarmeasure}
    \begin{aligned}
    \mathrm{d}\mu(U_N) &= \prod_{n=1}^{N} \prod_{n'=1}^{n - 1} \mathcal{P}_\mathrm{B}(t_{nj}; \alpha_{nj}, \beta_{nj}) \mathrm{d}t_{nj}\\
    &= \prod_{n=1}^{N} \prod_{n'=1}^{n - 1} \mathcal{P}_{\mathrm{B}, \theta}(\theta_{nj}; \alpha_{nj}, \beta_{nj}) \mathrm{d}\theta_{nj} \\
    &= \prod_{n=1}^{N} \prod_{n'=1}^{n - 1} \mathrm{d}\xi_{nj},
    \end{aligned}
\end{equation}
where
\begin{align*}
    \mathcal{P}_\mathrm{B}(t; \alpha, \beta) &= \frac{t^{\alpha - 1} (1 - t)^{\beta - 1}}{\mathrm{B}(\alpha, \beta)}\\
    \mathcal{P}_{\mathrm{B}, \theta}(\theta; \alpha, \beta) &= \frac{\left(\cos \frac{\theta}{2}\right)^{\alpha - 1} \left(\sin \frac{\theta}{2}\right)^{\beta - 1}}{\mathrm{B}(\alpha, \beta)} \\
    \xi_{nj} &= I_{t_{nj}}\left(\alpha_{nj}, \beta_{nj}\right).
\end{align*}
where $\theta_{nj} \in [0, \pi]$ and $\alpha, \beta$ are the number of outputs spanned by each subtree of the node. The tree coordinate system for real numbers in Equation \ref{eqn:binaryorthohaarmeasure} additionally matches the result for hyperspherical harmonics in Ref. \cite{Nikiforov1991ClassicalVariable}.

\section{Dirichlet node operator} \label{sec:irwin}
In this section, we explicitly define one possible definition for the Dirichlet node operator $O_K(\boldsymbol{y})$ in terms of the Gamma basis as defined in Section \ref{sec:treecoordinatestats}. We note $y_k = |v_k|^2$ for any complex vector $\boldsymbol{v}$ and the sum of $y_k$ is $Y_K$. Applying Gram-Schmidt orthogonalization to the vector $\boldsymbol{y}$ gives the matrix representation \cite{J.O.Irwin1942OnWeighting} (where $O_K = [O_{jk}]$):

\begin{equation} \label{eqn:irwingamma}
    \begin{aligned}
        O_{jk}(\boldsymbol{y}) &= \begin{cases}
        \sqrt{\frac{y_j}{Y_K}} & k = K\\
        \sqrt{\frac{y_jy_k}{\sum_{i=1}^{j-1} y_i \sum_{i=1}^{j} y_i}} & j < k < K\\
        \frac{-\sum_{i=1}^{j} y_i}{\sqrt{\sum_{i=1}^{j-1} y_i \sum_{i=1}^{j} y_i}} & 1 < j = k < K\\
        0 & k < j < K
        \end{cases}
    \end{aligned}
\end{equation}

Equation \ref{eqn:irwingamma} can also be written in terms of the unit Dirichlet basis defined in Section \ref{sec:treecoordinatestats} as 
\begin{equation} \label{eqn:irwindirichlet}
    \begin{aligned}
        O_{jk}(\boldsymbol{x}) &= \begin{cases}
        \sqrt{x_k} & k = K\\
        \sqrt{\frac{x_jx_k}{\sum_{i=1}^{j-1} x_i \sum_{i=1}^{j} x_i}} & j < k < K\\
        \frac{-\sum_{i=1}^{j} x_i}{\sqrt{\sum_{i=1}^{j-1} x_i \sum_{i=1}^{j} x_i}} & 1 < j = k < K\\
        0 & k < j < K
        \end{cases}
    \end{aligned}
\end{equation}

Any network component implementing $R_K$ would be parametrized by $\boldsymbol{x}$, which is independent of $Y_K$, and follows the form of Eq. \ref{eqn:irwindirichlet}. 

Note that this representation is \textit{not} necessarily unique. There may be many equivalent representations of $O_{jk}(\boldsymbol{x})$ in the Dirichlet basis that do not assume the form of Eq. \ref{eqn:irwindirichlet}. Empirically, we find that the absolute values of the matrix in Eq. \ref{eqn:irwindirichlet} is exactly equal to those of unbalanced binary tree representation in the main text, but a proof of this equivalence is not explicitly provided here.

\subsection{Dirichlet tree operator construction} \label{sec:dirichlettreeoperator}

In this section, we define a recursive protocol for defining any Dirichlet tree operator, which is a generalization of the definition in the main text.

\begin{definition}
Assume we are given a set of $K$ Dirichlet tree operators $O_{M_1, \mathcal{G}_1}, O_{M_2, \mathcal{G}_2}, \ldots O_{M_K, \mathcal{G}_K}$ programmed respectively to vectors $\boldsymbol{x}_{M_1}, \boldsymbol{x}_{M_2}, \ldots \boldsymbol{x}_{M_K}$. This set can be \textit{connected} to form a $M$-Dirichlet tree operator $O_{M, \mathcal{G}}$ programmed to $\boldsymbol{x}_M = [\boldsymbol{x}_{M_1}, \boldsymbol{x}_{M_2}, \ldots \boldsymbol{x}_{M_K}]$, where $M = \sum_{k = 1}^K M_k$. In the base case, if $M_k = 1$, the tree $\mathcal{G}_k$ is an empty set suggesting there are no further connections to other subgraphs, and we define $O_{1, \{\}} := O_1 = 1$. For the sake of the recursive definition, we define $\mathcal{G} = \{\mathcal{G}_1, \mathcal{G}_2, \ldots \mathcal{G}_K\}$.

We then define the overall connected ``OMG'' operator $O_{M, \mathcal{G}}(\boldsymbol{x}_M)$ recursively as:
\begin{equation}\label{eqn:connection}
    \begin{aligned}
    T_{M_1, \ldots M_K} &:= P_{M, K} O_{K}^{[M]}(\overline{\boldsymbol{x}}_K) \left(\prod_{k = 1}^K P_{k, M_k}\right)\\
    O_{\mathcal{G}_1, \ldots \mathcal{G}_K} &:= \begin{bmatrix}
        O_{M_1, \mathcal{G}_1} & 0 & \cdots & 0\\
        0 & O_{M_2, \mathcal{G}_2} & \cdots & 0\\
        \vdots & \vdots & \ddots & \vdots\\
        0 & 0 & \cdots & O_{M_K, \mathcal{G}_K}
    \end{bmatrix} \\
    O_{M, \mathcal{G}} &:= T_{M_1, \ldots M_K} (\widetilde{\boldsymbol{x}}_M) O_{\mathcal{G}_1, \ldots \mathcal{G}_K}(\boldsymbol{x}_M),
    \end{aligned}
\end{equation}
where $\widetilde{\boldsymbol{v}} := O_{\mathcal{G}_1, \ldots \mathcal{G}_K}(\boldsymbol{v}) \boldsymbol{v}$ is the vector propagated to , $O_{K}^{[M]}$ is a $K$-Dirichlet node embedded in the first $K$ dimensions of $M$-dimensional space, $\overline{\boldsymbol{x}}_K = [\overline{x}_1, \overline{x}_2, \ldots \overline{x}_K]$ where $\overline{x}_k = \sqrt{\sum_{m = 1}^{M_k}{x_{k, m}}}$, and $P_{i, j}$ is a permutation matrix switching indices $i, j$. As is shown in Lemma \ref{lm:mdirichlettree}, we have $O_{M, \mathcal{G}}(\boldsymbol{x}_M) \sqrt{\boldsymbol{x}_M} = \boldsymbol{e}_M$.
\end{definition}

\begin{lemma} \label{lm:mdirichlettree}
The product of a Dirichlet tree operator and its vector always has the form $O_{M, \mathcal{G}}(\boldsymbol{x}_M) \sqrt{\boldsymbol{x}_M} = \boldsymbol{e}_M$.
\end{lemma}

\begin{proof}
Assume we have some vector $\boldsymbol{v}_{M} \in \mathbb{C}^M$ and some graphical representation $\mathcal{G}$ defined as above.

We then perform the following calculation as defined in Equation \ref{eqn:connection}:

\begin{equation} \label{eqn:mprojectiontreeconnection}
\begin{aligned}
    O_{M; 1, M_1} \boldsymbol{v}_{M} &= \|\boldsymbol{v}_{M; 1, M_1}\| \boldsymbol{e}_{M_1} + \boldsymbol{v}_{M; M_1 + 1, M} \\&= \widetilde{\boldsymbol{v}}_{M; M_1, M} \\
    O_M \boldsymbol{v}_M &= O_{M; M_1, M} P_{m_2, M_1} \widetilde{\boldsymbol{v}}_{M; M_1, M} \\&= \sqrt{Y_M} \boldsymbol{e}_M,
\end{aligned}
\end{equation}
where we get $\sqrt{Y_M} = \|\boldsymbol{v}_{M}\|$ since all operators in Equation \ref{eqn:mprojectiontreeconnection} are unitary. Note $P_{m_2, M_1}$ does not affect the first $M_1 - 1$ elements of $\widetilde{\boldsymbol{v}}_{M; M_1, M}$, which should be zero.
\end{proof}

\bibliography{binarytree}

\begin{thebibliography}{47}%
\makeatletter
\providecommand \@ifxundefined [1]{%
 \@ifx{#1\undefined}
}%
\providecommand \@ifnum [1]{%
 \ifnum #1\expandafter \@firstoftwo
 \else \expandafter \@secondoftwo
 \fi
}%
\providecommand \@ifx [1]{%
 \ifx #1\expandafter \@firstoftwo
 \else \expandafter \@secondoftwo
 \fi
}%
\providecommand \natexlab [1]{#1}%
\providecommand \enquote  [1]{``#1''}%
\providecommand \bibnamefont  [1]{#1}%
\providecommand \bibfnamefont [1]{#1}%
\providecommand \citenamefont [1]{#1}%
\providecommand \href@noop [0]{\@secondoftwo}%
\providecommand \href [0]{\begingroup \@sanitize@url \@href}%
\providecommand \@href[1]{\@@startlink{#1}\@@href}%
\providecommand \@@href[1]{\endgroup#1\@@endlink}%
\providecommand \@sanitize@url [0]{\catcode `\\12\catcode `\$12\catcode
  `\&12\catcode `\#12\catcode `\^12\catcode `\_12\catcode `\%12\relax}%
\providecommand \@@startlink[1]{}%
\providecommand \@@endlink[0]{}%
\providecommand \url  [0]{\begingroup\@sanitize@url \@url }%
\providecommand \@url [1]{\endgroup\@href {#1}{\urlprefix }}%
\providecommand \urlprefix  [0]{URL }%
\providecommand \Eprint [0]{\href }%
\providecommand \doibase [0]{http://dx.doi.org/}%
\providecommand \selectlanguage [0]{\@gobble}%
\providecommand \bibinfo  [0]{\@secondoftwo}%
\providecommand \bibfield  [0]{\@secondoftwo}%
\providecommand \translation [1]{[#1]}%
\providecommand \BibitemOpen [0]{}%
\providecommand \bibitemStop [0]{}%
\providecommand \bibitemNoStop [0]{.\EOS\space}%
\providecommand \EOS [0]{\spacefactor3000\relax}%
\providecommand \BibitemShut  [1]{\csname bibitem#1\endcsname}%
\let\auto@bib@innerbib\@empty
\bibitem [{\citenamefont {Bogaerts}\ \emph {et~al.}(2020)\citenamefont
  {Bogaerts}, \citenamefont {P{\'{e}}rez}, \citenamefont {Capmany},
  \citenamefont {Miller}, \citenamefont {Poon}, \citenamefont {Englund},
  \citenamefont {Morichetti},\ and\ \citenamefont
  {Melloni}}]{Bogaerts2020ProgrammableCircuits}%
  \BibitemOpen
  \bibfield  {author} {\bibinfo {author} {\bibfnamefont {Wim}\ \bibnamefont
  {Bogaerts}}, \bibinfo {author} {\bibfnamefont {Daniel}\ \bibnamefont
  {P{\'{e}}rez}}, \bibinfo {author} {\bibfnamefont {José}\ \bibnamefont
  {Capmany}}, \bibinfo {author} {\bibfnamefont {David~A.B.}\ \bibnamefont
  {Miller}}, \bibinfo {author} {\bibfnamefont {Joyce}\ \bibnamefont {Poon}},
  \bibinfo {author} {\bibfnamefont {Dirk}\ \bibnamefont {Englund}}, \bibinfo
  {author} {\bibfnamefont {Francesco}\ \bibnamefont {Morichetti}}, \ and\
  \bibinfo {author} {\bibfnamefont {Andrea}\ \bibnamefont {Melloni}},\
  }\bibfield  {title} {\enquote {\bibinfo {title} {{Programmable photonic
  circuits}},}\ }\href {\doibase 10.1038/s41586-020-2764-0} {\bibfield
  {journal} {\bibinfo  {journal} {Nature}\ }\textbf {\bibinfo {volume} {586}},\
  \bibinfo {pages} {207--216} (\bibinfo {year} {2020})}\BibitemShut {NoStop}%
\bibitem [{\citenamefont
  {Miller}(2013{\natexlab{a}})}]{Miller2013Self-aligningCoupler}%
  \BibitemOpen
  \bibfield  {author} {\bibinfo {author} {\bibfnamefont {David A.~B.}\
  \bibnamefont {Miller}},\ }\bibfield  {title} {\enquote {\bibinfo {title}
  {{Self-aligning universal beam coupler}},}\ }\href {\doibase
  10.1364/OE.21.006360} {\bibfield  {journal} {\bibinfo  {journal} {Optics
  Express}\ }\textbf {\bibinfo {volume} {21}},\ \bibinfo {pages} {6360}
  (\bibinfo {year} {2013}{\natexlab{a}})}\BibitemShut {NoStop}%
\bibitem [{\citenamefont {Miller}(2020)}]{Miller2020AnalyzingNetworks}%
  \BibitemOpen
  \bibfield  {author} {\bibinfo {author} {\bibfnamefont {David A.~B.}\
  \bibnamefont {Miller}},\ }\bibfield  {title} {\enquote {\bibinfo {title}
  {{Analyzing and generating multimode optical fields using self-configuring
  networks}},}\ }\href {\doibase 10.1364/optica.391592} {\bibfield  {journal}
  {\bibinfo  {journal} {Optica}\ }\textbf {\bibinfo {volume} {7}},\ \bibinfo
  {pages} {794} (\bibinfo {year} {2020})}\BibitemShut {NoStop}%
\bibitem [{\citenamefont {Wang}\ \emph
  {et~al.}(2018{\natexlab{a}})\citenamefont {Wang}, \citenamefont {Zhang},
  \citenamefont {Stern}, \citenamefont {Lipson},\ and\ \citenamefont
  {Lon{\v{c}}ar}}]{Wang2018NanophotonicModulatorsb}%
  \BibitemOpen
  \bibfield  {author} {\bibinfo {author} {\bibfnamefont {Cheng}\ \bibnamefont
  {Wang}}, \bibinfo {author} {\bibfnamefont {Mian}\ \bibnamefont {Zhang}},
  \bibinfo {author} {\bibfnamefont {Brian}\ \bibnamefont {Stern}}, \bibinfo
  {author} {\bibfnamefont {Michal}\ \bibnamefont {Lipson}}, \ and\ \bibinfo
  {author} {\bibfnamefont {Marko}\ \bibnamefont {Lon{\v{c}}ar}},\ }\bibfield
  {title} {\enquote {\bibinfo {title} {{Nanophotonic lithium niobate
  electro-optic modulators}},}\ }\href {\doibase 10.1364/OE.26.001547}
  {\bibfield  {journal} {\bibinfo  {journal} {Optics Express}\ }\textbf
  {\bibinfo {volume} {26}},\ \bibinfo {pages} {1547} (\bibinfo {year}
  {2018}{\natexlab{a}})}\BibitemShut {NoStop}%
\bibitem [{\citenamefont {Harris}\ \emph {et~al.}(2014)\citenamefont {Harris},
  \citenamefont {Ma}, \citenamefont {Mower}, \citenamefont {Baehr-Jones},
  \citenamefont {Englund}, \citenamefont {Hochberg},\ and\ \citenamefont
  {Galland}}]{Harris2014EfficientSilicon}%
  \BibitemOpen
  \bibfield  {author} {\bibinfo {author} {\bibfnamefont {Nicholas~C.}\
  \bibnamefont {Harris}}, \bibinfo {author} {\bibfnamefont {Yangjin}\
  \bibnamefont {Ma}}, \bibinfo {author} {\bibfnamefont {Jacob}\ \bibnamefont
  {Mower}}, \bibinfo {author} {\bibfnamefont {Tom}\ \bibnamefont
  {Baehr-Jones}}, \bibinfo {author} {\bibfnamefont {Dirk}\ \bibnamefont
  {Englund}}, \bibinfo {author} {\bibfnamefont {Michael}\ \bibnamefont
  {Hochberg}}, \ and\ \bibinfo {author} {\bibfnamefont {Christophe}\
  \bibnamefont {Galland}},\ }\bibfield  {title} {\enquote {\bibinfo {title}
  {{Efficient, compact and low loss thermo-optic phase shifter in silicon}},}\
  }\href {\doibase 10.1364/OE.22.010487} {\bibfield  {journal} {\bibinfo
  {journal} {Optics Express}\ }\textbf {\bibinfo {volume} {22}},\ \bibinfo
  {pages} {10487} (\bibinfo {year} {2014})}\BibitemShut {NoStop}%
\bibitem [{\citenamefont {Errando-Herranz}\ \emph {et~al.}(2020)\citenamefont
  {Errando-Herranz}, \citenamefont {Takabayashi}, \citenamefont {Edinger},
  \citenamefont {Sattari}, \citenamefont {Gylfason},\ and\ \citenamefont
  {Quack}}]{Errando-Herranz2020MEMSCircuits}%
  \BibitemOpen
  \bibfield  {author} {\bibinfo {author} {\bibfnamefont {Carlos}\ \bibnamefont
  {Errando-Herranz}}, \bibinfo {author} {\bibfnamefont {Alain~Yuji}\
  \bibnamefont {Takabayashi}}, \bibinfo {author} {\bibfnamefont {Pierre}\
  \bibnamefont {Edinger}}, \bibinfo {author} {\bibfnamefont {Hamed}\
  \bibnamefont {Sattari}}, \bibinfo {author} {\bibfnamefont {Kristinn~B.}\
  \bibnamefont {Gylfason}}, \ and\ \bibinfo {author} {\bibfnamefont {Niels}\
  \bibnamefont {Quack}},\ }\bibfield  {title} {\enquote {\bibinfo {title}
  {{MEMS for Photonic Integrated Circuits}},}\ }\href {\doibase
  10.1109/JSTQE.2019.2943384} {\bibfield  {journal} {\bibinfo  {journal} {IEEE
  Journal of Selected Topics in Quantum Electronics}\ }\textbf {\bibinfo
  {volume} {26}} (\bibinfo {year} {2020}),\
  10.1109/JSTQE.2019.2943384}\BibitemShut {NoStop}%
\bibitem [{\citenamefont {Edinger}\ \emph {et~al.}(2020)\citenamefont
  {Edinger}, \citenamefont {Errando-Herranz}, \citenamefont {Yuji~Takabayashi},
  \citenamefont {Sattari}, \citenamefont {Quack}, \citenamefont {Verheyen},
  \citenamefont {Bogaerts},\ and\ \citenamefont
  {Gylfason}}]{Edinger2020CompactPhotonics}%
  \BibitemOpen
  \bibfield  {author} {\bibinfo {author} {\bibfnamefont {Pierre}\ \bibnamefont
  {Edinger}}, \bibinfo {author} {\bibfnamefont {Carlos}\ \bibnamefont
  {Errando-Herranz}}, \bibinfo {author} {\bibfnamefont {Alain}\ \bibnamefont
  {Yuji~Takabayashi}}, \bibinfo {author} {\bibfnamefont {Hamed}\ \bibnamefont
  {Sattari}}, \bibinfo {author} {\bibfnamefont {Niels}\ \bibnamefont {Quack}},
  \bibinfo {author} {\bibfnamefont {Peter}\ \bibnamefont {Verheyen}}, \bibinfo
  {author} {\bibfnamefont {Wim}\ \bibnamefont {Bogaerts}}, \ and\ \bibinfo
  {author} {\bibfnamefont {Kristinn~B}\ \bibnamefont {Gylfason}},\ }\href
  {http://hdl.handle.net/1854/LU-8673216} {\emph {\bibinfo {title} {Conference
  on Lasers and Electro-Optics (CLEO 2020)}}},\ \bibinfo {type} {Tech. Rep.}\
  (\bibinfo {year} {2020})\BibitemShut {NoStop}%
\bibitem [{\citenamefont {Wuttig}\ \emph {et~al.}(2017)\citenamefont {Wuttig},
  \citenamefont {Bhaskaran},\ and\ \citenamefont
  {Taubner}}]{Wuttig2017Phase-changeApplications}%
  \BibitemOpen
  \bibfield  {author} {\bibinfo {author} {\bibfnamefont {M.}~\bibnamefont
  {Wuttig}}, \bibinfo {author} {\bibfnamefont {H.}~\bibnamefont {Bhaskaran}}, \
  and\ \bibinfo {author} {\bibfnamefont {T.}~\bibnamefont {Taubner}},\ }\href
  {\doibase 10.1038/nphoton.2017.126} {\enquote {\bibinfo {title}
  {{Phase-change materials for non-volatile photonic applications}},}\ }
  (\bibinfo {year} {2017})\BibitemShut {NoStop}%
\bibitem [{\citenamefont {Wang}\ \emph
  {et~al.}(2018{\natexlab{b}})\citenamefont {Wang}, \citenamefont {Zhang},
  \citenamefont {Stern}, \citenamefont {Lipson},\ and\ \citenamefont
  {Lon{\v{c}}ar}}]{Wang2018NanophotonicModulators}%
  \BibitemOpen
  \bibfield  {author} {\bibinfo {author} {\bibfnamefont {Cheng}\ \bibnamefont
  {Wang}}, \bibinfo {author} {\bibfnamefont {Mian}\ \bibnamefont {Zhang}},
  \bibinfo {author} {\bibfnamefont {Brian}\ \bibnamefont {Stern}}, \bibinfo
  {author} {\bibfnamefont {Michal}\ \bibnamefont {Lipson}}, \ and\ \bibinfo
  {author} {\bibfnamefont {Marko}\ \bibnamefont {Lon{\v{c}}ar}},\ }\bibfield
  {title} {\enquote {\bibinfo {title} {{Nanophotonic lithium niobate
  electro-optic modulators}},}\ }\href {\doibase 10.1364/OE.26.001547}
  {\bibfield  {journal} {\bibinfo  {journal} {Optics Express}\ }\textbf
  {\bibinfo {volume} {26}},\ \bibinfo {pages} {1547} (\bibinfo {year}
  {2018}{\natexlab{b}})}\BibitemShut {NoStop}%
\bibitem [{\citenamefont {Reck}\ \emph {et~al.}(1994)\citenamefont {Reck},
  \citenamefont {Zeilinger}, \citenamefont {Bernstein},\ and\ \citenamefont
  {Bertani}}]{Reck1994ExperimentalOperator}%
  \BibitemOpen
  \bibfield  {author} {\bibinfo {author} {\bibfnamefont {Michael}\ \bibnamefont
  {Reck}}, \bibinfo {author} {\bibfnamefont {Anton}\ \bibnamefont {Zeilinger}},
  \bibinfo {author} {\bibfnamefont {Herbert~J.}\ \bibnamefont {Bernstein}}, \
  and\ \bibinfo {author} {\bibfnamefont {Philip}\ \bibnamefont {Bertani}},\
  }\bibfield  {title} {\enquote {\bibinfo {title} {{Experimental realization of
  any discrete unitary operator}},}\ }\href {\doibase
  10.1103/PhysRevLett.73.58} {\bibfield  {journal} {\bibinfo  {journal}
  {Physical Review Letters}\ }\textbf {\bibinfo {volume} {73}},\ \bibinfo
  {pages} {58--61} (\bibinfo {year} {1994})}\BibitemShut {NoStop}%
\bibitem [{\citenamefont
  {Miller}(2013{\natexlab{b}})}]{Miller2013Self-configuringInvited}%
  \BibitemOpen
  \bibfield  {author} {\bibinfo {author} {\bibfnamefont {David A.~B.}\
  \bibnamefont {Miller}},\ }\bibfield  {title} {\enquote {\bibinfo {title}
  {{Self-configuring universal linear optical component [Invited]}},}\ }\href
  {\doibase 10.1364/PRJ.1.000001} {\bibfield  {journal} {\bibinfo  {journal}
  {Photonics Research}\ }\textbf {\bibinfo {volume} {1}},\ \bibinfo {pages} {1}
  (\bibinfo {year} {2013}{\natexlab{b}})}\BibitemShut {NoStop}%
\bibitem [{\citenamefont {Hurwitz}(1897)}]{Hurwitz1897UberIntegration}%
  \BibitemOpen
  \bibfield  {author} {\bibinfo {author} {\bibfnamefont {Adolf}\ \bibnamefont
  {Hurwitz}},\ }\bibfield  {title} {\enquote {\bibinfo {title} {{{\"{u}}ber die
  Erzeugung der Invarianten durch Integration}},}\ }\href
  {https://eudml.org/doc/58378} {\bibfield  {journal} {\bibinfo  {journal}
  {Nachrichten von der Gesellschaft der Wissenschaften zu G{\"{o}}ttingen,
  Mathematisch-Physikalische Klasse}\ ,\ \bibinfo {pages} {71--72}} (\bibinfo
  {year} {1897})}\BibitemShut {NoStop}%
\bibitem [{\citenamefont {Pai}\ \emph {et~al.}(2022{\natexlab{a}})\citenamefont
  {Pai}, \citenamefont {Sun}, \citenamefont {Hughes}, \citenamefont {Park},
  \citenamefont {Bartlett}, \citenamefont {Williamson}, \citenamefont {Minkov},
  \citenamefont {Milanizadeh}, \citenamefont {Abebe}, \citenamefont
  {Morichetti}, \citenamefont {Melloni}, \citenamefont {Fan}, \citenamefont
  {Solgaard},\ and\ \citenamefont {Miller}}]{Pai2022ExperimentallyNetworks}%
  \BibitemOpen
  \bibfield  {author} {\bibinfo {author} {\bibfnamefont {Sunil}\ \bibnamefont
  {Pai}}, \bibinfo {author} {\bibfnamefont {Zhanghao}\ \bibnamefont {Sun}},
  \bibinfo {author} {\bibfnamefont {Tyler~W.}\ \bibnamefont {Hughes}}, \bibinfo
  {author} {\bibfnamefont {Taewon}\ \bibnamefont {Park}}, \bibinfo {author}
  {\bibfnamefont {Ben}\ \bibnamefont {Bartlett}}, \bibinfo {author}
  {\bibfnamefont {Ian A.~D.}\ \bibnamefont {Williamson}}, \bibinfo {author}
  {\bibfnamefont {Momchil}\ \bibnamefont {Minkov}}, \bibinfo {author}
  {\bibfnamefont {Maziyar}\ \bibnamefont {Milanizadeh}}, \bibinfo {author}
  {\bibfnamefont {Nathnael}\ \bibnamefont {Abebe}}, \bibinfo {author}
  {\bibfnamefont {Francesco}\ \bibnamefont {Morichetti}}, \bibinfo {author}
  {\bibfnamefont {Andrea}\ \bibnamefont {Melloni}}, \bibinfo {author}
  {\bibfnamefont {Shanhui}\ \bibnamefont {Fan}}, \bibinfo {author}
  {\bibfnamefont {Olav}\ \bibnamefont {Solgaard}}, \ and\ \bibinfo {author}
  {\bibfnamefont {David A.~B.}\ \bibnamefont {Miller}},\ }\bibfield  {title}
  {\enquote {\bibinfo {title} {{Experimentally realized in situ backpropagation
  for deep learning in nanophotonic neural networks}},}\ }\href
  {http://arxiv.org/abs/2205.08501} {\bibfield  {journal} {\bibinfo  {journal}
  {arXiv preprint}\ } (\bibinfo {year} {2022}{\natexlab{a}})}\BibitemShut
  {NoStop}%
\bibitem [{\citenamefont {Pai}\ \emph {et~al.}(2022{\natexlab{b}})\citenamefont
  {Pai}, \citenamefont {Park}, \citenamefont {Ball}, \citenamefont {Penkovsky},
  \citenamefont {Milanizadeh}, \citenamefont {Dubrovsky}, \citenamefont
  {Abebe}, \citenamefont {Morichetti}, \citenamefont {Melloni}, \citenamefont
  {Fan}, \citenamefont {Solgaard},\ and\ \citenamefont
  {Miller}}]{Pai2022ExperimentalCryptocurrency}%
  \BibitemOpen
  \bibfield  {author} {\bibinfo {author} {\bibfnamefont {Sunil}\ \bibnamefont
  {Pai}}, \bibinfo {author} {\bibfnamefont {Taewon}\ \bibnamefont {Park}},
  \bibinfo {author} {\bibfnamefont {Marshall}\ \bibnamefont {Ball}}, \bibinfo
  {author} {\bibfnamefont {Bogdan}\ \bibnamefont {Penkovsky}}, \bibinfo
  {author} {\bibfnamefont {Maziyar}\ \bibnamefont {Milanizadeh}}, \bibinfo
  {author} {\bibfnamefont {Michael}\ \bibnamefont {Dubrovsky}}, \bibinfo
  {author} {\bibfnamefont {Nathnael}\ \bibnamefont {Abebe}}, \bibinfo {author}
  {\bibfnamefont {Francesco}\ \bibnamefont {Morichetti}}, \bibinfo {author}
  {\bibfnamefont {Andrea}\ \bibnamefont {Melloni}}, \bibinfo {author}
  {\bibfnamefont {Shanhui}\ \bibnamefont {Fan}}, \bibinfo {author}
  {\bibfnamefont {Olav}\ \bibnamefont {Solgaard}}, \ and\ \bibinfo {author}
  {\bibfnamefont {David A.~B.}\ \bibnamefont {Miller}},\ }\bibfield  {title}
  {\enquote {\bibinfo {title} {{Experimental evaluation of digitally-verifiable
  photonic computing for blockchain and cryptocurrency}},}\ }\href
  {http://arxiv.org/abs/2205.08512} {\bibfield  {journal} {\bibinfo  {journal}
  {arXiv preprint}\ } (\bibinfo {year} {2022}{\natexlab{b}})}\BibitemShut
  {NoStop}%
\bibitem [{\citenamefont {Clements}\ \emph {et~al.}(2016)\citenamefont
  {Clements}, \citenamefont {Humphreys}, \citenamefont {Metcalf}, \citenamefont
  {Kolthammer},\ and\ \citenamefont
  {Walmsley}}]{Clements2016AnInterferometers}%
  \BibitemOpen
  \bibfield  {author} {\bibinfo {author} {\bibfnamefont {William~R.}\
  \bibnamefont {Clements}}, \bibinfo {author} {\bibfnamefont {Peter~C.}\
  \bibnamefont {Humphreys}}, \bibinfo {author} {\bibfnamefont {Benjamin~J.}\
  \bibnamefont {Metcalf}}, \bibinfo {author} {\bibfnamefont {W.~Steven}\
  \bibnamefont {Kolthammer}}, \ and\ \bibinfo {author} {\bibfnamefont {Ian~A.}\
  \bibnamefont {Walmsley}},\ }\bibfield  {title} {\enquote {\bibinfo {title}
  {{An Optimal Design for Universal Multiport Interferometers}},}\ }\href
  {\doibase 10.1364/OPTICA.3.001460} {\bibfield  {journal} {\bibinfo  {journal}
  {Optica}\ ,\ \bibinfo {pages} {1--8}} (\bibinfo {year} {2016})}\BibitemShut
  {NoStop}%
\bibitem [{\citenamefont {Hamerly}\ \emph
  {et~al.}(2021{\natexlab{a}})\citenamefont {Hamerly}, \citenamefont
  {Bandyopadhyay},\ and\ \citenamefont
  {Englund}}]{Hamerly2021InfinitelyInterferometers}%
  \BibitemOpen
  \bibfield  {author} {\bibinfo {author} {\bibfnamefont {Ryan}\ \bibnamefont
  {Hamerly}}, \bibinfo {author} {\bibfnamefont {Saumil}\ \bibnamefont
  {Bandyopadhyay}}, \ and\ \bibinfo {author} {\bibfnamefont {Dirk}\
  \bibnamefont {Englund}},\ }\bibfield  {title} {\enquote {\bibinfo {title}
  {{Infinitely Scalable Multiport Interferometers}},}\ }\href@noop {}
  {\bibfield  {journal} {\bibinfo  {journal} {arXiv preprint}\ } (\bibinfo
  {year} {2021}{\natexlab{a}})}\BibitemShut {NoStop}%
\bibitem [{\citenamefont {Chiles}\ \emph {et~al.}(2017)\citenamefont {Chiles},
  \citenamefont {Buckley}, \citenamefont {Nader}, \citenamefont {Nam},
  \citenamefont {Mirin},\ and\ \citenamefont
  {Shainline}}]{Chiles2017Multi-planarLoss}%
  \BibitemOpen
  \bibfield  {author} {\bibinfo {author} {\bibfnamefont {Jeff}\ \bibnamefont
  {Chiles}}, \bibinfo {author} {\bibfnamefont {Sonia}\ \bibnamefont {Buckley}},
  \bibinfo {author} {\bibfnamefont {Nima}\ \bibnamefont {Nader}}, \bibinfo
  {author} {\bibfnamefont {Sae~Woo}\ \bibnamefont {Nam}}, \bibinfo {author}
  {\bibfnamefont {Richard~P.}\ \bibnamefont {Mirin}}, \ and\ \bibinfo {author}
  {\bibfnamefont {Jeffrey~M.}\ \bibnamefont {Shainline}},\ }\bibfield  {title}
  {\enquote {\bibinfo {title} {{Multi-planar amorphous silicon photonics with
  compact interplanar couplers, cross talk mitigation, and low crossing
  loss}},}\ }\href {\doibase 10.1063/1.5000384} {\bibfield  {journal} {\bibinfo
   {journal} {APL Photonics}\ }\textbf {\bibinfo {volume} {2}},\ \bibinfo
  {pages} {116101} (\bibinfo {year} {2017})}\BibitemShut {NoStop}%
\bibitem [{\citenamefont {Araujo}\ \emph {et~al.}(2021)\citenamefont {Araujo},
  \citenamefont {Park}, \citenamefont {Petruccione},\ and\ \citenamefont
  {da~Silva}}]{Araujo2021APreparation}%
  \BibitemOpen
  \bibfield  {author} {\bibinfo {author} {\bibfnamefont {Israel~F.}\
  \bibnamefont {Araujo}}, \bibinfo {author} {\bibfnamefont {Daniel~K.}\
  \bibnamefont {Park}}, \bibinfo {author} {\bibfnamefont {Francesco}\
  \bibnamefont {Petruccione}}, \ and\ \bibinfo {author} {\bibfnamefont
  {Adenilton~J.}\ \bibnamefont {da~Silva}},\ }\bibfield  {title} {\enquote
  {\bibinfo {title} {{A divide-and-conquer algorithm for quantum state
  preparation}},}\ }\href {\doibase 10.1038/s41598-021-85474-1} {\bibfield
  {journal} {\bibinfo  {journal} {Scientific Reports}\ }\textbf {\bibinfo
  {volume} {11}},\ \bibinfo {pages} {6329} (\bibinfo {year}
  {2021})}\BibitemShut {NoStop}%
\bibitem [{\citenamefont {Flamini}\ \emph {et~al.}(2017)\citenamefont
  {Flamini}, \citenamefont {Spagnolo}, \citenamefont {Viggianiello},
  \citenamefont {Crespi}, \citenamefont {Osellame},\ and\ \citenamefont
  {Sciarrino}}]{Flamini2017BenchmarkingProcessing}%
  \BibitemOpen
  \bibfield  {author} {\bibinfo {author} {\bibfnamefont {Fulvio}\ \bibnamefont
  {Flamini}}, \bibinfo {author} {\bibfnamefont {Nicolò}\ \bibnamefont
  {Spagnolo}}, \bibinfo {author} {\bibfnamefont {Niko}\ \bibnamefont
  {Viggianiello}}, \bibinfo {author} {\bibfnamefont {Andrea}\ \bibnamefont
  {Crespi}}, \bibinfo {author} {\bibfnamefont {Roberto}\ \bibnamefont
  {Osellame}}, \ and\ \bibinfo {author} {\bibfnamefont {Fabio}\ \bibnamefont
  {Sciarrino}},\ }\bibfield  {title} {\enquote {\bibinfo {title} {{Benchmarking
  integrated linear-optical architectures for quantum information
  processing}},}\ }\href {\doibase 10.1038/s41598-017-15174-2} {\bibfield
  {journal} {\bibinfo  {journal} {Scientific Reports}\ }\textbf {\bibinfo
  {volume} {7}},\ \bibinfo {pages} {15133} (\bibinfo {year}
  {2017})}\BibitemShut {NoStop}%
\bibitem [{\citenamefont {M{\"{o}}tt{\"{o}}nen}\ \emph
  {et~al.}(2004)\citenamefont {M{\"{o}}tt{\"{o}}nen}, \citenamefont
  {Vartiainen}, \citenamefont {Bergholm},\ and\ \citenamefont
  {Salomaa}}]{Mottonen2004QuantumGates}%
  \BibitemOpen
  \bibfield  {author} {\bibinfo {author} {\bibfnamefont {Mikko}\ \bibnamefont
  {M{\"{o}}tt{\"{o}}nen}}, \bibinfo {author} {\bibfnamefont {Juha~J.}\
  \bibnamefont {Vartiainen}}, \bibinfo {author} {\bibfnamefont {Ville}\
  \bibnamefont {Bergholm}}, \ and\ \bibinfo {author} {\bibfnamefont
  {Martti~M.}\ \bibnamefont {Salomaa}},\ }\bibfield  {title} {\enquote
  {\bibinfo {title} {{Quantum Circuits for General Multiqubit Gates}},}\ }\href
  {\doibase 10.1103/PhysRevLett.93.130502} {\bibfield  {journal} {\bibinfo
  {journal} {Physical Review Letters}\ }\textbf {\bibinfo {volume} {93}},\
  \bibinfo {pages} {130502} (\bibinfo {year} {2004})}\BibitemShut {NoStop}%
\bibitem [{\citenamefont {Bandyopadhyay}\ \emph {et~al.}(2021)\citenamefont
  {Bandyopadhyay}, \citenamefont {Hamerly},\ and\ \citenamefont
  {Englund}}]{Bandyopadhyay2021HardwarePhotonics}%
  \BibitemOpen
  \bibfield  {author} {\bibinfo {author} {\bibfnamefont {Saumil}\ \bibnamefont
  {Bandyopadhyay}}, \bibinfo {author} {\bibfnamefont {Ryan}\ \bibnamefont
  {Hamerly}}, \ and\ \bibinfo {author} {\bibfnamefont {Dirk}\ \bibnamefont
  {Englund}},\ }\bibfield  {title} {\enquote {\bibinfo {title} {{Hardware error
  correction for programmable photonics}},}\ }\href {\doibase
  10.1364/optica.424052} {\bibfield  {journal} {\bibinfo  {journal} {Optica}\
  }\textbf {\bibinfo {volume} {8}},\ \bibinfo {pages} {1247} (\bibinfo {year}
  {2021})}\BibitemShut {NoStop}%
\bibitem [{\citenamefont {Hamerly}\ \emph
  {et~al.}(2021{\natexlab{b}})\citenamefont {Hamerly}, \citenamefont
  {Bandyopadhyay},\ and\ \citenamefont
  {Englund}}]{Hamerly2021AccurateInterferometers}%
  \BibitemOpen
  \bibfield  {author} {\bibinfo {author} {\bibfnamefont {Ryan}\ \bibnamefont
  {Hamerly}}, \bibinfo {author} {\bibfnamefont {Saumil}\ \bibnamefont
  {Bandyopadhyay}}, \ and\ \bibinfo {author} {\bibfnamefont {Dirk}\
  \bibnamefont {Englund}},\ }\bibfield  {title} {\enquote {\bibinfo {title}
  {{Accurate Self-Configuration of Rectangular Multiport Interferometers}},}\
  }\href {http://arxiv.org/abs/2106.03249} {\bibfield  {journal} {\bibinfo
  {journal} {arXiv preprint}\ } (\bibinfo {year}
  {2021}{\natexlab{b}})}\BibitemShut {NoStop}%
\bibitem [{\citenamefont {Pai}\ \emph {et~al.}(2020)\citenamefont {Pai},
  \citenamefont {Williamson}, \citenamefont {Hughes}, \citenamefont {Minkov},
  \citenamefont {Solgaard}, \citenamefont {Fan},\ and\ \citenamefont
  {Miller}}]{Pai2020ParallelNetwork}%
  \BibitemOpen
  \bibfield  {author} {\bibinfo {author} {\bibfnamefont {Sunil}\ \bibnamefont
  {Pai}}, \bibinfo {author} {\bibfnamefont {Ian~A.D.}\ \bibnamefont
  {Williamson}}, \bibinfo {author} {\bibfnamefont {Tyler~W.}\ \bibnamefont
  {Hughes}}, \bibinfo {author} {\bibfnamefont {Momchil}\ \bibnamefont
  {Minkov}}, \bibinfo {author} {\bibfnamefont {Olav}\ \bibnamefont {Solgaard}},
  \bibinfo {author} {\bibfnamefont {Shanhui}\ \bibnamefont {Fan}}, \ and\
  \bibinfo {author} {\bibfnamefont {David~A.B.}\ \bibnamefont {Miller}},\
  }\bibfield  {title} {\enquote {\bibinfo {title} {{Parallel Programming of an
  Arbitrary Feedforward Photonic Network}},}\ }\href {\doibase
  10.1109/JSTQE.2020.2997849} {\bibfield  {journal} {\bibinfo  {journal} {IEEE
  Journal of Selected Topics in Quantum Electronics}\ }\textbf {\bibinfo
  {volume} {26}} (\bibinfo {year} {2020}),\
  10.1109/JSTQE.2020.2997849}\BibitemShut {NoStop}%
\bibitem [{\citenamefont {Shen}\ \emph {et~al.}(2017)\citenamefont {Shen},
  \citenamefont {Harris}, \citenamefont {Skirlo}, \citenamefont {Prabhu},
  \citenamefont {Baehr-Jones}, \citenamefont {Hochberg}, \citenamefont {Sun},
  \citenamefont {Zhao}, \citenamefont {Larochelle}, \citenamefont {Englund},\
  and\ \citenamefont {Solja{\v{c}}i{\'{c}}}}]{Shen2017DeepCircuits}%
  \BibitemOpen
  \bibfield  {author} {\bibinfo {author} {\bibfnamefont {Yichen}\ \bibnamefont
  {Shen}}, \bibinfo {author} {\bibfnamefont {Nicholas~C.}\ \bibnamefont
  {Harris}}, \bibinfo {author} {\bibfnamefont {Scott}\ \bibnamefont {Skirlo}},
  \bibinfo {author} {\bibfnamefont {Mihika}\ \bibnamefont {Prabhu}}, \bibinfo
  {author} {\bibfnamefont {Tom}\ \bibnamefont {Baehr-Jones}}, \bibinfo {author}
  {\bibfnamefont {Michael}\ \bibnamefont {Hochberg}}, \bibinfo {author}
  {\bibfnamefont {Xin}\ \bibnamefont {Sun}}, \bibinfo {author} {\bibfnamefont
  {Shijie}\ \bibnamefont {Zhao}}, \bibinfo {author} {\bibfnamefont {Hugo}\
  \bibnamefont {Larochelle}}, \bibinfo {author} {\bibfnamefont {Dirk}\
  \bibnamefont {Englund}}, \ and\ \bibinfo {author} {\bibfnamefont {Marin}\
  \bibnamefont {Solja{\v{c}}i{\'{c}}}},\ }\bibfield  {title} {\enquote
  {\bibinfo {title} {{Deep learning with coherent nanophotonic circuits}},}\
  }\href {\doibase 10.1038/nphoton.2017.93} {\bibfield  {journal} {\bibinfo
  {journal} {Nature Photonics}\ }\textbf {\bibinfo {volume} {11}},\ \bibinfo
  {pages} {441--446} (\bibinfo {year} {2017})}\BibitemShut {NoStop}%
\bibitem [{\citenamefont {Annoni}\ \emph {et~al.}(2017)\citenamefont {Annoni},
  \citenamefont {Guglielmi}, \citenamefont {Carminati}, \citenamefont
  {Ferrari}, \citenamefont {Sampietro}, \citenamefont {Miller}, \citenamefont
  {Melloni},\ and\ \citenamefont {Morichetti}}]{Annoni2017UnscramblingModes}%
  \BibitemOpen
  \bibfield  {author} {\bibinfo {author} {\bibfnamefont {Andrea}\ \bibnamefont
  {Annoni}}, \bibinfo {author} {\bibfnamefont {Emanuele}\ \bibnamefont
  {Guglielmi}}, \bibinfo {author} {\bibfnamefont {Marco}\ \bibnamefont
  {Carminati}}, \bibinfo {author} {\bibfnamefont {Giorgio}\ \bibnamefont
  {Ferrari}}, \bibinfo {author} {\bibfnamefont {Marco}\ \bibnamefont
  {Sampietro}}, \bibinfo {author} {\bibfnamefont {David~Ab}\ \bibnamefont
  {Miller}}, \bibinfo {author} {\bibfnamefont {Andrea}\ \bibnamefont
  {Melloni}}, \ and\ \bibinfo {author} {\bibfnamefont {Francesco}\ \bibnamefont
  {Morichetti}},\ }\bibfield  {title} {\enquote {\bibinfo {title}
  {{Unscrambling light - Automatically undoing strong mixing between modes}},}\
  }\href {\doibase 10.1038/lsa.2017.110} {\bibfield  {journal} {\bibinfo
  {journal} {Light: Science and Applications}\ }\textbf {\bibinfo {volume} {6}}
  (\bibinfo {year} {2017}),\ 10.1038/lsa.2017.110}\BibitemShut {NoStop}%
\bibitem [{\citenamefont {Harris}\ \emph {et~al.}(2018)\citenamefont {Harris},
  \citenamefont {Carolan}, \citenamefont {Bunandar}, \citenamefont {Prabhu},
  \citenamefont {Hochberg}, \citenamefont {Baehr-Jones}, \citenamefont {Fanto},
  \citenamefont {Smith}, \citenamefont {Tison}, \citenamefont {Alsing},\ and\
  \citenamefont {Englund}}]{Harris2018LinearProcessors}%
  \BibitemOpen
  \bibfield  {author} {\bibinfo {author} {\bibfnamefont {Nicholas~C.}\
  \bibnamefont {Harris}}, \bibinfo {author} {\bibfnamefont {Jacques}\
  \bibnamefont {Carolan}}, \bibinfo {author} {\bibfnamefont {Darius}\
  \bibnamefont {Bunandar}}, \bibinfo {author} {\bibfnamefont {Mihika}\
  \bibnamefont {Prabhu}}, \bibinfo {author} {\bibfnamefont {Michael}\
  \bibnamefont {Hochberg}}, \bibinfo {author} {\bibfnamefont {Tom}\
  \bibnamefont {Baehr-Jones}}, \bibinfo {author} {\bibfnamefont {Michael~L.}\
  \bibnamefont {Fanto}}, \bibinfo {author} {\bibfnamefont {A.~Matthew}\
  \bibnamefont {Smith}}, \bibinfo {author} {\bibfnamefont {Christopher~C.}\
  \bibnamefont {Tison}}, \bibinfo {author} {\bibfnamefont {Paul~M.}\
  \bibnamefont {Alsing}}, \ and\ \bibinfo {author} {\bibfnamefont {Dirk}\
  \bibnamefont {Englund}},\ }\bibfield  {title} {\enquote {\bibinfo {title}
  {{Linear programmable nanophotonic processors}},}\ }\href {\doibase
  10.1364/OPTICA.5.001623} {\bibfield  {journal} {\bibinfo  {journal} {Optica}\
  }\textbf {\bibinfo {volume} {5}},\ \bibinfo {pages} {1623} (\bibinfo {year}
  {2018})}\BibitemShut {NoStop}%
\bibitem [{\citenamefont {Taballione}\ \emph {et~al.}(2018)\citenamefont
  {Taballione}, \citenamefont {Wolterink}, \citenamefont {Lugani},
  \citenamefont {Eckstein}, \citenamefont {Bell}, \citenamefont {Grootjans},
  \citenamefont {Visscher}, \citenamefont {Renema}, \citenamefont {Geskus},
  \citenamefont {Roeloffzen}, \citenamefont {Walmsley}, \citenamefont
  {Pinkse},\ and\ \citenamefont {Boller}}]{Taballione20188x8Waveguides}%
  \BibitemOpen
  \bibfield  {author} {\bibinfo {author} {\bibfnamefont {Caterina}\
  \bibnamefont {Taballione}}, \bibinfo {author} {\bibfnamefont {Tom A.~W.}\
  \bibnamefont {Wolterink}}, \bibinfo {author} {\bibfnamefont {Jasleen}\
  \bibnamefont {Lugani}}, \bibinfo {author} {\bibfnamefont {Andreas}\
  \bibnamefont {Eckstein}}, \bibinfo {author} {\bibfnamefont {Bryn~A.}\
  \bibnamefont {Bell}}, \bibinfo {author} {\bibfnamefont {Robert}\ \bibnamefont
  {Grootjans}}, \bibinfo {author} {\bibfnamefont {Ilka}\ \bibnamefont
  {Visscher}}, \bibinfo {author} {\bibfnamefont {Jelmer~J.}\ \bibnamefont
  {Renema}}, \bibinfo {author} {\bibfnamefont {Dimitri}\ \bibnamefont
  {Geskus}}, \bibinfo {author} {\bibfnamefont {Chris G.~H.}\ \bibnamefont
  {Roeloffzen}}, \bibinfo {author} {\bibfnamefont {Ian~A.}\ \bibnamefont
  {Walmsley}}, \bibinfo {author} {\bibfnamefont {Pepijn W.~H.}\ \bibnamefont
  {Pinkse}}, \ and\ \bibinfo {author} {\bibfnamefont {Klaus-J.}\ \bibnamefont
  {Boller}},\ }\bibfield  {title} {\enquote {\bibinfo {title} {{8x8
  Programmable Quantum Photonic Processor based on Silicon Nitride
  Waveguides}},}\ }in\ \href {\doibase 10.1364/FIO.2018.JTu3A.58} {\emph
  {\bibinfo {booktitle} {Frontiers in Optics / Laser Science}}}\ (\bibinfo
  {publisher} {OSA},\ \bibinfo {address} {Washington, D.C.},\ \bibinfo {year}
  {2018})\ p.\ \bibinfo {pages} {JTu3A.58}\BibitemShut {NoStop}%
\bibitem [{\citenamefont {Huang}(1994)}]{Huang1994Coupled-modeOverview}%
  \BibitemOpen
  \bibfield  {author} {\bibinfo {author} {\bibfnamefont {Wei-ping}\
  \bibnamefont {Huang}},\ }\bibfield  {title} {\enquote {\bibinfo {title}
  {{Coupled-mode theory for optical waveguides : an overview}},}\ }\href@noop
  {} {\bibfield  {journal} {\bibinfo  {journal} {Journal of the Optical Society
  of America A}\ }\textbf {\bibinfo {volume} {11}},\ \bibinfo {pages}
  {963--983} (\bibinfo {year} {1994})}\BibitemShut {NoStop}%
\bibitem [{\citenamefont {Perez}\ \emph {et~al.}(2017)\citenamefont {Perez},
  \citenamefont {Gasulla}, \citenamefont {Capmany},\ and\ \citenamefont
  {Soref}}]{Perez2017HexagonalInterferometers}%
  \BibitemOpen
  \bibfield  {author} {\bibinfo {author} {\bibfnamefont {Daniel}\ \bibnamefont
  {Perez}}, \bibinfo {author} {\bibfnamefont {Ivana}\ \bibnamefont {Gasulla}},
  \bibinfo {author} {\bibfnamefont {Jose}\ \bibnamefont {Capmany}}, \ and\
  \bibinfo {author} {\bibfnamefont {Richard~A.}\ \bibnamefont {Soref}},\
  }\bibfield  {title} {\enquote {\bibinfo {title} {{Hexagonal waveguide mesh
  design for universal multiport interferometers}},}\ }in\ \href {\doibase
  10.1109/IPCon.2016.7831099} {\emph {\bibinfo {booktitle} {2016 IEEE Photonics
  Conference, IPC 2016}}}\ (\bibinfo {year} {2017})\ pp.\ \bibinfo {pages}
  {285--286}\BibitemShut {NoStop}%
\bibitem [{\citenamefont {Pai}(2022)}]{Pai2022Simphox:Library}%
  \BibitemOpen
  \bibfield  {author} {\bibinfo {author} {\bibfnamefont {Sunil}\ \bibnamefont
  {Pai}},\ }\href {https://github.com/fancompute/simphox} {\enquote {\bibinfo
  {title} {{simphox: Another inverse design library}},}\ } (\bibinfo {year}
  {2022})\BibitemShut {NoStop}%
\bibitem [{\citenamefont {Russell}\ \emph {et~al.}(2017)\citenamefont
  {Russell}, \citenamefont {Chakhmakhchyan}, \citenamefont {O'Brien},\ and\
  \citenamefont {Laing}}]{Russell2017DirectMatrices}%
  \BibitemOpen
  \bibfield  {author} {\bibinfo {author} {\bibfnamefont {Nicholas~J.}\
  \bibnamefont {Russell}}, \bibinfo {author} {\bibfnamefont {Levon}\
  \bibnamefont {Chakhmakhchyan}}, \bibinfo {author} {\bibfnamefont {Jeremy~L.}\
  \bibnamefont {O'Brien}}, \ and\ \bibinfo {author} {\bibfnamefont {Anthony}\
  \bibnamefont {Laing}},\ }\bibfield  {title} {\enquote {\bibinfo {title}
  {{Direct dialling of Haar random unitary matrices}},}\ }\href {\doibase
  10.1088/1367-2630/aa60ed} {\bibfield  {journal} {\bibinfo  {journal} {New
  Journal of Physics}\ } (\bibinfo {year} {2017}),\
  10.1088/1367-2630/aa60ed}\BibitemShut {NoStop}%
\bibitem [{\citenamefont
  {Miller}(2013{\natexlab{c}})}]{Miller2013EstablishingAutomatically}%
  \BibitemOpen
  \bibfield  {author} {\bibinfo {author} {\bibfnamefont {David A.~B.}\
  \bibnamefont {Miller}},\ }\bibfield  {title} {\enquote {\bibinfo {title}
  {{Establishing Optimal Wave Communication Channels Automatically}},}\ }\href
  {https://www.osapublishing.org/jlt/abstract.cfm?uri=jlt-31-24-3987}
  {\bibfield  {journal} {\bibinfo  {journal} {Journal of Lightwave Technology,
  Vol. 31, Issue 24, pp. 3987-3994}\ }\textbf {\bibinfo {volume} {31}},\
  \bibinfo {pages} {3987--3994} (\bibinfo {year}
  {2013}{\natexlab{c}})}\BibitemShut {NoStop}%
\bibitem [{\citenamefont {Turk}\ and\ \citenamefont
  {Pentland}(1991)}]{Turk1991EigenfacesRecognition}%
  \BibitemOpen
  \bibfield  {author} {\bibinfo {author} {\bibfnamefont {M.}~\bibnamefont
  {Turk}}\ and\ \bibinfo {author} {\bibfnamefont {A.}~\bibnamefont
  {Pentland}},\ }\bibfield  {title} {\enquote {\bibinfo {title} {{Eigenfaces
  for Recognition}},}\ }\href {\doibase 10.1162/JOCN.1991.3.1.71} {\bibfield
  {journal} {\bibinfo  {journal} {Journal of Cognitive Neuroscience}\ }\textbf
  {\bibinfo {volume} {3}},\ \bibinfo {pages} {71--86} (\bibinfo {year}
  {1991})}\BibitemShut {NoStop}%
\bibitem [{\citenamefont {Nahmias}\ \emph {et~al.}(2020)\citenamefont
  {Nahmias}, \citenamefont {De~Lima}, \citenamefont {Tait}, \citenamefont
  {Peng}, \citenamefont {Shastri},\ and\ \citenamefont
  {Prucnal}}]{Nahmias2020PhotonicNetworks}%
  \BibitemOpen
  \bibfield  {author} {\bibinfo {author} {\bibfnamefont {Mitchell~A.}\
  \bibnamefont {Nahmias}}, \bibinfo {author} {\bibfnamefont {Thomas~Ferreira}\
  \bibnamefont {De~Lima}}, \bibinfo {author} {\bibfnamefont {Alexander~N.}\
  \bibnamefont {Tait}}, \bibinfo {author} {\bibfnamefont {Hsuan~Tung}\
  \bibnamefont {Peng}}, \bibinfo {author} {\bibfnamefont {Bhavin~J.}\
  \bibnamefont {Shastri}}, \ and\ \bibinfo {author} {\bibfnamefont {Paul~R.}\
  \bibnamefont {Prucnal}},\ }\bibfield  {title} {\enquote {\bibinfo {title}
  {{Photonic Multiply-Accumulate Operations for Neural Networks}},}\ }\href
  {\doibase 10.1109/JSTQE.2019.2941485} {\bibfield  {journal} {\bibinfo
  {journal} {IEEE Journal of Selected Topics in Quantum Electronics}\ }\textbf
  {\bibinfo {volume} {26}} (\bibinfo {year} {2020}),\
  10.1109/JSTQE.2019.2941485}\BibitemShut {NoStop}%
\bibitem [{\citenamefont {Arrazola}\ \emph {et~al.}(2021)\citenamefont
  {Arrazola}, \citenamefont {Bergholm}, \citenamefont {Br{\'{a}}dler},
  \citenamefont {Bromley}, \citenamefont {Collins}, \citenamefont {Dhand},
  \citenamefont {Fumagalli}, \citenamefont {Gerrits}, \citenamefont {Goussev},
  \citenamefont {Helt}, \citenamefont {Hundal}, \citenamefont {Isacsson},
  \citenamefont {Israel}, \citenamefont {Izaac}, \citenamefont {Jahangiri},
  \citenamefont {Janik}, \citenamefont {Killoran}, \citenamefont {Kumar},
  \citenamefont {Lavoie}, \citenamefont {Lita}, \citenamefont {Mahler},
  \citenamefont {Menotti}, \citenamefont {Morrison}, \citenamefont {Nam},
  \citenamefont {Neuhaus}, \citenamefont {Qi}, \citenamefont {Quesada},
  \citenamefont {Repingon}, \citenamefont {Sabapathy}, \citenamefont {Schuld},
  \citenamefont {Su}, \citenamefont {Swinarton}, \citenamefont {Sz{\'{a}}va},
  \citenamefont {Tan}, \citenamefont {Tan}, \citenamefont {Vaidya},
  \citenamefont {Vernon}, \citenamefont {Zabaneh},\ and\ \citenamefont
  {Zhang}}]{Arrazola2021QuantumChip}%
  \BibitemOpen
  \bibfield  {author} {\bibinfo {author} {\bibfnamefont {J.~M.}\ \bibnamefont
  {Arrazola}}, \bibinfo {author} {\bibfnamefont {V.}~\bibnamefont {Bergholm}},
  \bibinfo {author} {\bibfnamefont {K.}~\bibnamefont {Br{\'{a}}dler}}, \bibinfo
  {author} {\bibfnamefont {T.~R.}\ \bibnamefont {Bromley}}, \bibinfo {author}
  {\bibfnamefont {M.~J.}\ \bibnamefont {Collins}}, \bibinfo {author}
  {\bibfnamefont {I.}~\bibnamefont {Dhand}}, \bibinfo {author} {\bibfnamefont
  {A.}~\bibnamefont {Fumagalli}}, \bibinfo {author} {\bibfnamefont
  {T.}~\bibnamefont {Gerrits}}, \bibinfo {author} {\bibfnamefont
  {A.}~\bibnamefont {Goussev}}, \bibinfo {author} {\bibfnamefont {L.~G.}\
  \bibnamefont {Helt}}, \bibinfo {author} {\bibfnamefont {J.}~\bibnamefont
  {Hundal}}, \bibinfo {author} {\bibfnamefont {T.}~\bibnamefont {Isacsson}},
  \bibinfo {author} {\bibfnamefont {R.~B.}\ \bibnamefont {Israel}}, \bibinfo
  {author} {\bibfnamefont {J.}~\bibnamefont {Izaac}}, \bibinfo {author}
  {\bibfnamefont {S.}~\bibnamefont {Jahangiri}}, \bibinfo {author}
  {\bibfnamefont {R.}~\bibnamefont {Janik}}, \bibinfo {author} {\bibfnamefont
  {N.}~\bibnamefont {Killoran}}, \bibinfo {author} {\bibfnamefont {S.~P.}\
  \bibnamefont {Kumar}}, \bibinfo {author} {\bibfnamefont {J.}~\bibnamefont
  {Lavoie}}, \bibinfo {author} {\bibfnamefont {A.~E.}\ \bibnamefont {Lita}},
  \bibinfo {author} {\bibfnamefont {D.~H.}\ \bibnamefont {Mahler}}, \bibinfo
  {author} {\bibfnamefont {M.}~\bibnamefont {Menotti}}, \bibinfo {author}
  {\bibfnamefont {B.}~\bibnamefont {Morrison}}, \bibinfo {author}
  {\bibfnamefont {S.~W.}\ \bibnamefont {Nam}}, \bibinfo {author} {\bibfnamefont
  {L.}~\bibnamefont {Neuhaus}}, \bibinfo {author} {\bibfnamefont {H.~Y.}\
  \bibnamefont {Qi}}, \bibinfo {author} {\bibfnamefont {N.}~\bibnamefont
  {Quesada}}, \bibinfo {author} {\bibfnamefont {A.}~\bibnamefont {Repingon}},
  \bibinfo {author} {\bibfnamefont {K.~K.}\ \bibnamefont {Sabapathy}}, \bibinfo
  {author} {\bibfnamefont {M.}~\bibnamefont {Schuld}}, \bibinfo {author}
  {\bibfnamefont {D.}~\bibnamefont {Su}}, \bibinfo {author} {\bibfnamefont
  {J.}~\bibnamefont {Swinarton}}, \bibinfo {author} {\bibfnamefont
  {A.}~\bibnamefont {Sz{\'{a}}va}}, \bibinfo {author} {\bibfnamefont
  {K.}~\bibnamefont {Tan}}, \bibinfo {author} {\bibfnamefont {P.}~\bibnamefont
  {Tan}}, \bibinfo {author} {\bibfnamefont {V.~D.}\ \bibnamefont {Vaidya}},
  \bibinfo {author} {\bibfnamefont {Z.}~\bibnamefont {Vernon}}, \bibinfo
  {author} {\bibfnamefont {Z.}~\bibnamefont {Zabaneh}}, \ and\ \bibinfo
  {author} {\bibfnamefont {Y.}~\bibnamefont {Zhang}},\ }\bibfield  {title}
  {\enquote {\bibinfo {title} {{Quantum circuits with many photons on a
  programmable nanophotonic chip}},}\ }\href {\doibase
  10.1038/s41586-021-03202-1} {\bibfield  {journal} {\bibinfo  {journal}
  {Nature}\ }\textbf {\bibinfo {volume} {591}},\ \bibinfo {pages} {54--60}
  (\bibinfo {year} {2021})}\BibitemShut {NoStop}%
\bibitem [{\citenamefont {Miller}(2015)}]{Miller2015PerfectComponents}%
  \BibitemOpen
  \bibfield  {author} {\bibinfo {author} {\bibfnamefont {David A.~B.}\
  \bibnamefont {Miller}},\ }\bibfield  {title} {\enquote {\bibinfo {title}
  {{Perfect optics with imperfect components}},}\ }\href {\doibase
  10.1364/OPTICA.2.000747} {\bibfield  {journal} {\bibinfo  {journal} {Optica}\
  }\textbf {\bibinfo {volume} {2}},\ \bibinfo {pages} {747} (\bibinfo {year}
  {2015})}\BibitemShut {NoStop}%
\bibitem [{\citenamefont {Cabanillas}\ \emph {et~al.}(2019)\citenamefont
  {Cabanillas}, \citenamefont {Gevorgyan}, \citenamefont {Khilo},\ and\
  \citenamefont {Popovi{\'{c}}}}]{Cabanillas2019ExperimentalCouplers}%
  \BibitemOpen
  \bibfield  {author} {\bibinfo {author} {\bibfnamefont {Josep M~Fargas}\
  \bibnamefont {Cabanillas}}, \bibinfo {author} {\bibfnamefont {Hayk}\
  \bibnamefont {Gevorgyan}}, \bibinfo {author} {\bibfnamefont {Anatol}\
  \bibnamefont {Khilo}}, \ and\ \bibinfo {author} {\bibfnamefont {Miloš~A}\
  \bibnamefont {Popovi{\'{c}}}},\ }\bibfield  {title} {\enquote {\bibinfo
  {title} {{Experimental Demonstration of Rapid Adiabatic Couplers}},}\ }in\
  \href {\doibase 10.1364/CLEO{\_}SI.2019.SM3J.5} {\emph {\bibinfo {booktitle}
  {Conference on Lasers and Electro-Optics}}}\ (\bibinfo  {publisher} {Optical
  Society of America},\ \bibinfo {year} {2019})\ p.\ \bibinfo {pages}
  {SM3J.5}\BibitemShut {NoStop}%
\bibitem [{\citenamefont {Lu}\ \emph {et~al.}(2015)\citenamefont {Lu},
  \citenamefont {Celo}, \citenamefont {Dumais}, \citenamefont {Bernier},\ and\
  \citenamefont {Chrostowski}}]{Lu2015ComparisonPlatforms}%
  \BibitemOpen
  \bibfield  {author} {\bibinfo {author} {\bibfnamefont {Zeqin}\ \bibnamefont
  {Lu}}, \bibinfo {author} {\bibfnamefont {Dritan}\ \bibnamefont {Celo}},
  \bibinfo {author} {\bibfnamefont {Patrick}\ \bibnamefont {Dumais}}, \bibinfo
  {author} {\bibfnamefont {Eric}\ \bibnamefont {Bernier}}, \ and\ \bibinfo
  {author} {\bibfnamefont {Lukas}\ \bibnamefont {Chrostowski}},\ }\bibfield
  {title} {\enquote {\bibinfo {title} {{Comparison of photonic 2×2 3-dB
  couplers for 220 nm silicon-on-insulator platforms}},}\ }in\ \href {\doibase
  10.1109/Group4.2015.7305944} {\emph {\bibinfo {booktitle} {IEEE International
  Conference on Group IV Photonics GFP}}},\ Vol.\ \bibinfo {volume}
  {2015-October}\ (\bibinfo  {publisher} {IEEE Computer Society},\ \bibinfo
  {year} {2015})\ pp.\ \bibinfo {pages} {57--58}\BibitemShut {NoStop}%
\bibitem [{\citenamefont {Pai}\ \emph {et~al.}(2019)\citenamefont {Pai},
  \citenamefont {Bartlett}, \citenamefont {Solgaard},\ and\ \citenamefont
  {Miller}}]{Pai2019MatrixDevices}%
  \BibitemOpen
  \bibfield  {author} {\bibinfo {author} {\bibfnamefont {Sunil}\ \bibnamefont
  {Pai}}, \bibinfo {author} {\bibfnamefont {Ben}\ \bibnamefont {Bartlett}},
  \bibinfo {author} {\bibfnamefont {Olav}\ \bibnamefont {Solgaard}}, \ and\
  \bibinfo {author} {\bibfnamefont {David A.~B.}\ \bibnamefont {Miller}},\
  }\bibfield  {title} {\enquote {\bibinfo {title} {{Matrix Optimization on
  Universal Unitary Photonic Devices}},}\ }\href {\doibase
  10.1103/PhysRevApplied.11.064044} {\bibfield  {journal} {\bibinfo  {journal}
  {Physical Review Applied}\ }\textbf {\bibinfo {volume} {11}},\ \bibinfo
  {pages} {064044} (\bibinfo {year} {2019})}\BibitemShut {NoStop}%
\bibitem [{\citenamefont {Basani}\ \emph {et~al.}(2022)\citenamefont {Basani},
  \citenamefont {Vadlamani}, \citenamefont {Bandyopadhyay}, \citenamefont
  {Englund},\ and\ \citenamefont {Hamerly}}]{Basani2022AInterferometers}%
  \BibitemOpen
  \bibfield  {author} {\bibinfo {author} {\bibfnamefont {Jasvith~Raj}\
  \bibnamefont {Basani}}, \bibinfo {author} {\bibfnamefont {Sri~Krishna}\
  \bibnamefont {Vadlamani}}, \bibinfo {author} {\bibfnamefont {Saumil}\
  \bibnamefont {Bandyopadhyay}}, \bibinfo {author} {\bibfnamefont {Dirk~R.}\
  \bibnamefont {Englund}}, \ and\ \bibinfo {author} {\bibfnamefont {Ryan}\
  \bibnamefont {Hamerly}},\ }\bibfield  {title} {\enquote {\bibinfo {title} {{A
  Self-Similar Sine-Cosine Fractal Architecture for Multiport
  Interferometers}},}\ }\href {\doibase 10.48550/arxiv.2209.03335} {\bibfield
  {journal} {\bibinfo  {journal} {arXiv preprint}\ } (\bibinfo {year} {2022}),\
  10.48550/arxiv.2209.03335}\BibitemShut {NoStop}%
\bibitem [{\citenamefont {Jing}\ \emph {et~al.}(2017)\citenamefont {Jing},
  \citenamefont {Shen}, \citenamefont {Dubcek}, \citenamefont {Peurifoy},
  \citenamefont {Skirlo}, \citenamefont {LeCun}, \citenamefont {Tegmark},\ and\
  \citenamefont {Solja{\v{c}}i{\'{c}}}}]{Jing2017TunableRNNs}%
  \BibitemOpen
  \bibfield  {author} {\bibinfo {author} {\bibfnamefont {Li}~\bibnamefont
  {Jing}}, \bibinfo {author} {\bibfnamefont {Yichen}\ \bibnamefont {Shen}},
  \bibinfo {author} {\bibfnamefont {Tena}\ \bibnamefont {Dubcek}}, \bibinfo
  {author} {\bibfnamefont {John}\ \bibnamefont {Peurifoy}}, \bibinfo {author}
  {\bibfnamefont {Scott}\ \bibnamefont {Skirlo}}, \bibinfo {author}
  {\bibfnamefont {Yann}\ \bibnamefont {LeCun}}, \bibinfo {author}
  {\bibfnamefont {Max}\ \bibnamefont {Tegmark}}, \ and\ \bibinfo {author}
  {\bibfnamefont {Marin}\ \bibnamefont {Solja{\v{c}}i{\'{c}}}},\ }\bibfield
  {title} {\enquote {\bibinfo {title} {{Tunable Efficient Unitary Neural
  Networks (EUNN) and their application to RNNs}},}\ }in\ \href
  {http://proceedings.mlr.press/v70/jing17a} {\emph {\bibinfo {booktitle}
  {Proceedings of Machine Learning Research}}}\ (\bibinfo {year} {2017})\ pp.\
  \bibinfo {pages} {1733--1741}\BibitemShut {NoStop}%
\bibitem [{\citenamefont {Fang}\ \emph {et~al.}(2019)\citenamefont {Fang},
  \citenamefont {Manipatruni}, \citenamefont {Wierzynski}, \citenamefont
  {Khosrowshahi},\ and\ \citenamefont {DeWeese}}]{Fang2019DesignImprecisions}%
  \BibitemOpen
  \bibfield  {author} {\bibinfo {author} {\bibfnamefont {Michael Y.-S.}\
  \bibnamefont {Fang}}, \bibinfo {author} {\bibfnamefont {Sasikanth}\
  \bibnamefont {Manipatruni}}, \bibinfo {author} {\bibfnamefont {Casimir}\
  \bibnamefont {Wierzynski}}, \bibinfo {author} {\bibfnamefont {Amir}\
  \bibnamefont {Khosrowshahi}}, \ and\ \bibinfo {author} {\bibfnamefont
  {Michael~R.}\ \bibnamefont {DeWeese}},\ }\bibfield  {title} {\enquote
  {\bibinfo {title} {{Design of optical neural networks with component
  imprecisions}},}\ }\href {\doibase 10.1364/OE.27.014009} {\bibfield
  {journal} {\bibinfo  {journal} {Optics Express}\ }\textbf {\bibinfo {volume}
  {27}},\ \bibinfo {pages} {14009} (\bibinfo {year} {2019})}\BibitemShut
  {NoStop}%
\bibitem [{\citenamefont {Carolan}\ \emph {et~al.}(2015)\citenamefont
  {Carolan}, \citenamefont {Harrold}, \citenamefont {Sparrow}, \citenamefont
  {Mart{\'{i}}n-L{\'{o}}pez}, \citenamefont {Russell}, \citenamefont
  {Silverstone}, \citenamefont {Shadbolt}, \citenamefont {Matsuda},
  \citenamefont {Oguma}, \citenamefont {Itoh}, \citenamefont {Marshall},
  \citenamefont {Thompson}, \citenamefont {Matthews}, \citenamefont
  {Hashimoto}, \citenamefont {O'Brien},\ and\ \citenamefont
  {Laing}}]{Carolan2015UniversalOptics}%
  \BibitemOpen
  \bibfield  {author} {\bibinfo {author} {\bibfnamefont {Jacques}\ \bibnamefont
  {Carolan}}, \bibinfo {author} {\bibfnamefont {Christopher}\ \bibnamefont
  {Harrold}}, \bibinfo {author} {\bibfnamefont {Chris}\ \bibnamefont
  {Sparrow}}, \bibinfo {author} {\bibfnamefont {Enrique}\ \bibnamefont
  {Mart{\'{i}}n-L{\'{o}}pez}}, \bibinfo {author} {\bibfnamefont {Nicholas~J.}\
  \bibnamefont {Russell}}, \bibinfo {author} {\bibfnamefont {Joshua~W.}\
  \bibnamefont {Silverstone}}, \bibinfo {author} {\bibfnamefont {Peter~J.}\
  \bibnamefont {Shadbolt}}, \bibinfo {author} {\bibfnamefont {Nobuyuki}\
  \bibnamefont {Matsuda}}, \bibinfo {author} {\bibfnamefont {Manabu}\
  \bibnamefont {Oguma}}, \bibinfo {author} {\bibfnamefont {Mikitaka}\
  \bibnamefont {Itoh}}, \bibinfo {author} {\bibfnamefont {Graham~D.}\
  \bibnamefont {Marshall}}, \bibinfo {author} {\bibfnamefont {Mark~G.}\
  \bibnamefont {Thompson}}, \bibinfo {author} {\bibfnamefont {Jonathan~C.F.}\
  \bibnamefont {Matthews}}, \bibinfo {author} {\bibfnamefont {Toshikazu}\
  \bibnamefont {Hashimoto}}, \bibinfo {author} {\bibfnamefont {Jeremy~L.}\
  \bibnamefont {O'Brien}}, \ and\ \bibinfo {author} {\bibfnamefont {Anthony}\
  \bibnamefont {Laing}},\ }\bibfield  {title} {\enquote {\bibinfo {title}
  {{Universal linear optics}},}\ }\href {\doibase 10.1126/science.aab3642}
  {\bibfield  {journal} {\bibinfo  {journal} {Science}\ } (\bibinfo {year}
  {2015}),\ 10.1126/science.aab3642}\BibitemShut {NoStop}%
\bibitem [{\citenamefont {Dennis}(1996)}]{Dennis1996AProblems}%
  \BibitemOpen
  \bibfield  {author} {\bibinfo {author} {\bibfnamefont {Samuel~Y.}\
  \bibnamefont {Dennis}},\ }\bibfield  {title} {\enquote {\bibinfo {title} {{A
  Bayesian analysis of tree-structured statistical decision problems}},}\
  }\href {\doibase 10.1016/0378-3758(95)00112-3} {\bibfield  {journal}
  {\bibinfo  {journal} {Journal of Statistical Planning and Inference}\
  }\textbf {\bibinfo {volume} {53}},\ \bibinfo {pages} {323--344} (\bibinfo
  {year} {1996})}\BibitemShut {NoStop}%
\bibitem [{\citenamefont {Harris}\ \emph {et~al.}(2017)\citenamefont {Harris},
  \citenamefont {Steinbrecher}, \citenamefont {Prabhu}, \citenamefont {Lahini},
  \citenamefont {Mower}, \citenamefont {Bunandar}, \citenamefont {Chen},
  \citenamefont {Wong}, \citenamefont {Baehr-Jones}, \citenamefont {Hochberg},
  \citenamefont {Lloyd},\ and\ \citenamefont
  {Englund}}]{Harris2017QuantumProcessor}%
  \BibitemOpen
  \bibfield  {author} {\bibinfo {author} {\bibfnamefont {Nicholas~C.}\
  \bibnamefont {Harris}}, \bibinfo {author} {\bibfnamefont {Gregory~R.}\
  \bibnamefont {Steinbrecher}}, \bibinfo {author} {\bibfnamefont {Mihika}\
  \bibnamefont {Prabhu}}, \bibinfo {author} {\bibfnamefont {Yoav}\ \bibnamefont
  {Lahini}}, \bibinfo {author} {\bibfnamefont {Jacob}\ \bibnamefont {Mower}},
  \bibinfo {author} {\bibfnamefont {Darius}\ \bibnamefont {Bunandar}}, \bibinfo
  {author} {\bibfnamefont {Changchen}\ \bibnamefont {Chen}}, \bibinfo {author}
  {\bibfnamefont {Franco~N.C.}\ \bibnamefont {Wong}}, \bibinfo {author}
  {\bibfnamefont {Tom}\ \bibnamefont {Baehr-Jones}}, \bibinfo {author}
  {\bibfnamefont {Michael}\ \bibnamefont {Hochberg}}, \bibinfo {author}
  {\bibfnamefont {Seth}\ \bibnamefont {Lloyd}}, \ and\ \bibinfo {author}
  {\bibfnamefont {Dirk}\ \bibnamefont {Englund}},\ }\bibfield  {title}
  {\enquote {\bibinfo {title} {{Quantum transport simulations in a programmable
  nanophotonic processor}},}\ }\href {\doibase 10.1038/nphoton.2017.95}
  {\bibfield  {journal} {\bibinfo  {journal} {Nature Photonics}\ }\textbf
  {\bibinfo {volume} {11}},\ \bibinfo {pages} {447--452} (\bibinfo {year}
  {2017})}\BibitemShut {NoStop}%
\bibitem [{\citenamefont {Nikiforov}\ \emph {et~al.}(1991)\citenamefont
  {Nikiforov}, \citenamefont {Uvarov},\ and\ \citenamefont
  {Suslov}}]{Nikiforov1991ClassicalVariable}%
  \BibitemOpen
  \bibfield  {author} {\bibinfo {author} {\bibfnamefont {Arnold~F.}\
  \bibnamefont {Nikiforov}}, \bibinfo {author} {\bibfnamefont {Vasilii~B.}\
  \bibnamefont {Uvarov}}, \ and\ \bibinfo {author} {\bibfnamefont {Sergei~K.}\
  \bibnamefont {Suslov}},\ }\bibfield  {title} {\enquote {\bibinfo {title}
  {{Classical Orthogonal Polynomials of a Discrete Variable}},}\ }in\ \href
  {\doibase 10.1007/978-3-642-74748-9{\_}2} {\emph {\bibinfo {booktitle}
  {Classical Orthogonal Polynomials of a Discrete Variable}}}\ (\bibinfo
  {publisher} {Springer Berlin Heidelberg},\ \bibinfo {address} {Berlin,
  Heidelberg},\ \bibinfo {year} {1991})\ pp.\ \bibinfo {pages}
  {18--54}\BibitemShut {NoStop}%
\bibitem [{\citenamefont {{J. O. Irwin}}(1942)}]{J.O.Irwin1942OnWeighting}%
  \BibitemOpen
  \bibfield  {author} {\bibinfo {author} {\bibnamefont {{J. O. Irwin}}},\
  }\bibfield  {title} {\enquote {\bibinfo {title} {{On the Distribution of a
  Weighted Estimate of Variance and on Analysis of Variance in Certain Cases of
  Unequal Weighting}},}\ }\href {https://www.jstor.org/stable/pdf/2980611.pdf}
  {\bibfield  {journal} {\bibinfo  {journal} {Journal of the Royal Statistical
  Society}\ }\textbf {\bibinfo {volume} {105}},\ \bibinfo {pages} {115--118}
  (\bibinfo {year} {1942})}\BibitemShut {NoStop}%
\end{thebibliography}%

\end{document}